\documentclass[aoas]{imsart}
\RequirePackage{amsthm,amsmath,amsfonts,amssymb}
\RequirePackage[authoryear]{natbib}
\RequirePackage{graphicx}

\RequirePackage{booktabs}
\RequirePackage{enumerate}
\RequirePackage{mathtools}
\RequirePackage{multirow}
\RequirePackage{subcaption}
\RequirePackage{tabularx}
\RequirePackage{url}
\RequirePackage{pdfpages}

\RequirePackage[acronym]{glossaries}

\startlocaldefs
\theoremstyle{plain}
\newtheorem{theorem}{Theorem}
\newtheorem{lemma}[theorem]{Lemma}
\newtheorem{proposition}[theorem]{Proposition}

\theoremstyle{definition}

\newtheorem{remark}{Remark}

\newcommand{\btheta}{\boldsymbol{\theta}}
\newcommand{\fisher}{\mathcal{I}}

\newacronym{bh}{BH}{Benjamini--Hochberg}
\newacronym{bvm}{BvM}{Bernstein--von Mises}
\newacronym{fdr}{FDR}{false discovery rate}
\newacronym{gfi}{GFI}{generalized fiducial inference}
\newacronym{gfd}{GFD}{generalized fiducial distribution}
\newacronym{hgfi}{HGFI}{hybrid generalized fiducial inference}
\newacronym{iid}{i.i.d.}{independent and identically distributed}
\newacronym{irb}{IRB}{Institutional Review Board}
\newacronym{mh}{MH}{Metropolis--Hastings}
\newacronym{nb}{NB}{negative binomial}
\newacronym{pmf}{PMF}{probability mass function}

\newcommand{\nbderivationapp}{Section~1 of the Supplementary Material \citep{thomas_hannig_supp}}
\newcommand{\scorehessianapp}{Section~2 of the Supplementary Material \citep{thomas_hannig_supp}}
\newcommand{\jointfiducialproofapp}{Section~3 of the Supplementary Material \citep{thomas_hannig_supp}}
\newcommand{\blockdiagjacobianproofapp}{Section~4 of the Supplementary Material \citep{thomas_hannig_supp}}
\newcommand{\posdefproofapp}{Section~5 of the Supplementary Material \citep{thomas_hannig_supp}}
\newcommand{\bvmproofapp}{Section~6 of the Supplementary Material \citep{thomas_hannig_supp}}
\newcommand{\simulationapp}{Section~7 of the Supplementary Material \citep{thomas_hannig_supp}}
\newcommand{\discretecompapp}{Section~8 of the Supplementary Material \citep{thomas_hannig_supp}}
\newcommand{\permutationapp}{Section~9.1 of the Supplementary Material \citep{thomas_hannig_supp}}
\newcommand{\thresholdsensapp}{Section~9.2 of the Supplementary Material \citep{thomas_hannig_supp}}

\newcommand{\lanassump}{Assumption~1}
\endlocaldefs

\begin{document}

\begin{frontmatter}
\title{Generalized Fiducial Inference for Hybrid Discrete--Continuous Count Data: Theory and Application}
\runtitle{Hybrid Generalized Fiducial Inference for Count Data}

\begin{aug}
\author[A]{\fnms{Kendall}~\snm{Thomas}\ead[label=e1]{ketho@unc.edu}\orcid{0009-0004-6372-8960}}  
\author[A]{\fnms{Jan}~\snm{Hannig}\ead[label=e2]{jan.hannig@unc.edu}\orcid{0000-0002-4164-0173}}
\address[A]{Department of Statistics and Operations Research, University of North Carolina at Chapel Hill, Chapel Hill, NC, US\printead[presep={,\ }]{e1,e2}}
\end{aug}

\begin{abstract}
High-frequency athlete-tracking data can be summarized as quantile cubes (Thomas and Hannig, 2026): counts of time spent in bins of velocity, acceleration, and movement angle. In our application, each athlete--match session has $100$ such bins, with counts ranging from under a hundred to several thousand. Whether a bin's negative binomial distribution is well approximated by a normal distribution depends on its shape parameter, the product of its mean and dispersion, rather than on the count alone. We develop a hybrid generalized fiducial framework that uses exact negative binomial structural equations where the estimated shape is small and moment-matched normal approximations where it is large. This approach provides uncertainty quantification for hundreds of regression and dispersion parameters without specifying a prior distribution. Only the normal components contribute to the fiducial Jacobian, which has a closed-form block-diagonal structure, and we prove a Bernstein--von Mises theorem for the correctly specified hybrid model. Simulations show near-nominal coverage, agreement with a flat-prior Bayesian analysis, and the importance of partitioning by estimated shape rather than by the mean or observed responses. In $216$ sessions from $17$ professional women's soccer athletes, the estimated position effects concentrate at the velocity extremes, while the estimated age effects follow an acceleration gradient. However, an athlete-level permutation check indicates that these patterns are not distinguishable from those obtained under athlete-label permutations.
\end{abstract}


\end{frontmatter}

\section{Introduction}
\label{sec:intro}

Wearable tracking systems record an athlete's movement many times per second throughout a soccer match, but comparing athletes or matches requires summaries of the resulting trajectories \citep{bourdon_2017}. Common summaries, such as total distance covered and time above a fixed speed, describe overall workload but not how movement is distributed across combinations of velocity, acceleration and direction \citep{snyder_2024}. Two athletes can therefore have nearly the same totals but very different movement profiles. Quantile cubes provide a distributional representation of these profiles \citep{thomas_hannig_2026}. Velocity, acceleration, and movement angle are each divided into bins at empirical quantiles, and the time spent in each combination is recorded. With five velocity bins, five acceleration bins, and four angle bins (oriented to forward, right, backward, and left), each athlete--match session can be represented as a $d = 5 \times 5 \times 4 = 100$-dimensional vector of counts. Because the same cutoffs are used for all athlete--match sessions, the components have consistent interpretations across sessions. We analyze $216$ athlete--match sessions from $17$ athletes on a professional women's soccer team over one season, and ask how movement time is allocated across the velocity--acceleration--angle grid and how that allocation relates to athlete characteristics.

The quantile-cube counts vary widely. Each marginal bin holds an equal share of the pooled movement time by construction, but the joint combinations do not: common movement states hold thousands of deciseconds (tenths of a second), while unusual combinations of velocity, acceleration and angle hold hundreds. The counts are also overdispersed relative to a Poisson model, so we use a \gls{nb} model \citep[Sec.~4.2]{cameron_trivedi_2013}; see also \citet[Ch.~7]{hilbe_2011}. Whether a \gls{nb} distribution is far from normal is governed by its shape parameter, the mean times the dispersion, which is small mainly in sparsely occupied bins. In these bins, its discreteness, skewness and nonnegativity all affect inference, and the exact distribution matters \citep{hilbe_2017}. Where the shape is large the distribution is close to normal, and a normal approximation with the same mean and variance avoids the cost of repeatedly evaluating and inverting the \gls{nb} distribution function during fiducial sampling. We therefore use a hybrid representation, chosen separately for each bin of each athlete--match session: the exact \gls{nb} model where the estimated shape parameter is small and the normal approximation where it is large.

The number of parameters is a second challenge. Giving each movement bin its own regression coefficients and dispersion parameter yields $q=500$ parameters in our application, $100$ of them dispersions. \Gls{gfi} provides a distribution on this parameter space from the data-generating equations alone, with no prior to specify, which for $500$ parameters, many of them weakly informed, is a real saving. By the \gls{bvm} theorem of Section~\ref{sec:asymptotics}, that distribution agrees to first order with a flat-prior Bayesian analysis, and the simulations confirm the agreement at moderate sample sizes.

Several existing approaches address related features of heterogeneous count data. Continuous distributions are commonly used to approximate large counts \citep{anders_huber_2010}, while composite and working likelihoods give tractable inference when a full joint model is hard to specify \citep{varin_reid_firth_2011}. Mixture, zero-inflated and hurdle models handle other kinds of count heterogeneity, such as excess zeros and latent subpopulations \citep{mullahy_1986, lambert_1992, risso_et_al_2018, feng_2021}. None of these methods yields a generalized fiducial construction in which each observation is modeled either exactly or by a continuous approximation, and the fiducial Jacobian and asymptotic behavior of such a construction have not been studied.

We therefore develop a \gls{hgfi} framework for multivariate count data with many components.\footnote{The asymptotic theory holds the number of components and the parameter dimension fixed while the number of independent observations grows.} \Gls{gfi} produces a data-dependent distribution on the parameter space without a prior \citep{hannig_2009, hannig_et_al_2016}. It starts from a structural equation $\mathbf Y=G(\mathbf U,\btheta)$ that writes the data as a function of the parameter $\btheta$ and auxiliary variables $\mathbf U$ with a known distribution. Given the observed $\mathbf y$, the equation is inverted: the auxiliary variables are drawn from their known distribution, and the parameter values for which $G(\mathbf U,\btheta)$ reproduces $\mathbf y$ are kept. The \gls{gfd} is the limit of this procedure as the required agreement with $\mathbf y$ becomes exact. Under conditions given in \citet{hannig_et_al_2016}, the resulting density is
$$r(\btheta\mid\mathbf y)\propto L(\mathbf y\mid\btheta)\,J(\btheta,\mathbf y),$$
where $L$ is the likelihood and $J$ is a Jacobian-type factor determined by the structural equation. Asymptotic normality and \gls{bvm} results are available for broad classes of such procedures \citep{hannig_et_al_2016, borgert_hannig_2026}. We build on this theory to derive the hybrid generalized fiducial density and establish its asymptotic behavior.

The theory is stated for the oracle partition, a fixed partition determined by the true \gls{nb} shapes, under a correctly specified hybrid model with conditionally independent observations. It does not cover three features of our application: the dependence among movement bins induced by a fixed total time, repeated sessions from the same athlete, and a partition estimated from the data, for which we compare candidate rules by simulation (Section~\ref{sec:sims:partition}). Section~\ref{sec:discuss:limitations} discusses all three.

This paper makes four contributions. First, we derive the generalized fiducial density for the mixed discrete--continuous model and show that, when the continuous structural equations have full column rank and are injective, it factors into the hybrid likelihood and a Jacobian that comes only from the continuous components. Second, for the hybrid \gls{nb}--normal model with component-specific parameters, we show that this Jacobian is block diagonal with closed-form blocks, so each component can be computed on its own. Third, for components that are \gls{nb} in every observation and so contribute no Jacobian, we give an exact finite-sample representation of the discrete \gls{gfd} and a \gls{mh} sampler that computes it at the application's scale, where the single-site Gibbs sampler of \citet{hannig_iyer_wang_2007} does not mix (Section~\ref{sec:method:discrete} and \discretecompapp). Fourth, building on \citet{borgert_hannig_2026}, we prove a \gls{bvm} theorem for the parameters of components that have at least one normal observation, under the model conditions stated in Section~\ref{sec:asymptotics}.

The remainder of the paper is organized as follows. Section~\ref{sec:method} sets up the hybrid \gls{nb}--normal model and its generalized fiducial formulation, including the discrete construction for components that are \gls{nb} in every observation. Section~\ref{sec:theory} derives the joint fiducial density and the continuous Jacobian. Section~\ref{sec:asymptotics} proves the \gls{bvm} theorem. Section~\ref{sec:sims} presents the simulation study, Section~\ref{sec:app} applies the method to professional women's soccer data, and Section~\ref{sec:discuss} concludes with limitations and extensions.

\section{The Hybrid Model}
\label{sec:method}

This section specifies the hybrid model and its structural-equation representation. We call it a working model because it is the model under which inference is carried out, not a claim about how the data were generated: the exact \gls{nb} structural equation is kept for components with small \gls{nb} shape, and a normal structural equation with the same mean and variance replaces it where the shape is large. The two are combined within a single generalized fiducial construction \citep{hannig_et_al_2016}.

Let $Z_i=(X_i,t_i)$ denote the design pair for observation $i\in\{1,\ldots,n\}$, where $X_i\in\mathbb R^K$ is a covariate vector and $t_i>0$ is a known exposure. Under the random-design formulation used in the asymptotic theory, $Z_1,\ldots,Z_n$ are \gls{iid} draws from a distribution $Q$. The response for observation $i$ is the count vector $\mathbf y_i=(y_{i1},\ldots,y_{id})^\top\in\mathbb N_0^d$. We collect responses and covariates into $\mathbf y\in\mathbb N_0^{n\times d}$ and $\mathbf X\in\mathbb R^{n\times K}$. For each observation $i$, let $D_i$ denote the set of components represented by the exact \gls{nb} model and $C_i$ the set represented by the normal approximation, so that $D_i\cup C_i=\{1,\ldots,d\}$ and $D_i\cap C_i=\varnothing$ (the oracle partition is defined below), giving the observation-level sample space
$$\left(\prod_{j\in D_i}\mathbb N_0\right) \times \left(\prod_{j\in C_i}\mathbb R\right).$$

Conditional on the realized design and this partition, the working model treats response vectors as independent across observations and scalar component responses as conditionally independent within each observation. It does not impose the fixed-sum constraint that compositional count vectors satisfy \citep{aitchison_1982}, a limitation discussed in Section~\ref{sec:discuss:limitations}.

Conditional on $Z_i$, the component-specific \gls{nb} working model arises from a Poisson--Gamma hierarchy \citep[pp.~117--119]{cameron_trivedi_2013}:
$$Y_{ij}\mid\lambda_{ij}\sim\operatorname{Poisson}(\lambda_{ij}),
\qquad
\lambda_{ij}\sim\operatorname{Gamma}\big(\operatorname{shape}=\alpha_{ij}t_i,\ \operatorname{rate}=\nu_j\big),$$
with $\alpha_{ij}=\exp(X_i^\top\beta_j)$, component-specific regression coefficients $\beta_j\in\mathbb R^K$ and dispersion $\nu_j>0$. Marginalizing over $\lambda_{ij}$ gives $Y_{ij}\sim\operatorname{NB}\big(r_{ij},\ \nu_j/(1+\nu_j)\big)$ with shape $r_{ij}=\alpha_{ij}t_i$, which equals $\mu_{ij}\nu_j$, the mean times the dispersion, and
$$\mu_{ij}=\mathbb E_{\btheta}[Y_{ij}\mid Z_i]=\frac{\alpha_{ij}t_i}{\nu_j},
\qquad
\sigma_{ij}^2=\operatorname{Var}_{\btheta}(Y_{ij}\mid Z_i)=\frac{\alpha_{ij}t_i(\nu_j+1)}{\nu_j^2}
=\mu_{ij}\big(1+\nu_j^{-1}\big),$$
so smaller $\nu_j$ corresponds to greater overdispersion and the Poisson limit is recovered as $\nu_j\to\infty$. The derivation is given in \nbderivationapp. We write
$$\btheta=(\btheta_1^\top,\ldots,\btheta_d^\top)^\top,
\qquad
\btheta_j=(\beta_j^\top,\nu_j)^\top,
\qquad
q=d(K+1).$$
The dispersion $\nu_j$ is shared across observations within component $j$ but may vary across components.

The \gls{nb} distribution is well approximated by a normal with the same mean and variance when its shape parameter is large. Unlike the Poisson distribution, it does not approach normality as its mean grows at fixed shape \citep[p.~199]{hilbe_2011}. Its skewness is $(2+\nu_j)/\{r_{ij}(1+\nu_j)\}^{1/2}$. In the overdispersed regime relevant here, where $\nu_j$ is small, this is approximately $2/r_{ij}^{1/2}$ and is essentially free of the mean, whereas in the Poisson limit, as $\nu_j\to\infty$, it reduces to $\mu_{ij}^{-1/2}$. The accuracy of the approximation is governed by $r_{ij}$ rather than by the mean, and Section~\ref{sec:sims:dispersion} confirms this in terms of coverage. We therefore define the partition from the true shape $r_{ij}(\btheta_0)=\exp(X_i^\top\beta_{j,0})\,t_i$, which depends on the regression coefficients and exposure but not on the dispersion: given a fixed threshold $\omega>0$,
$$C_i=\{j:r_{ij}(\btheta_0)\geq\omega\},
\qquad
D_i=\{j:r_{ij}(\btheta_0)<\omega\}.$$
We refer to this as the oracle partition. Because $r_{ij}=\mu_{ij}\nu_j$, a threshold $c$ on the conditional mean corresponds to a shape threshold $\omega=c\,\nu_{j,0}$ that varies with the component, so a mean rule coincides with the shape rule only when the dispersion is common to all components.

Under the random-design formulation the oracle partition is random through $Z_i$. Conditional on the realized design and $\btheta_0$ it is fixed, does not depend on the realized response, and is not recomputed as the candidate parameter $\btheta$ varies in the likelihood, score, Hessian, Jacobian, or local asymptotic expansion. The partition is observation-specific: the same component may belong to $D_i$ for one observation and to $C_i$ for another. The component-specific hybrid working model is
$$Y_{ij}\mid Z_i
\sim
\begin{cases}
\operatorname{NB}\big(r_{ij}, \nu_j/(1+\nu_j)\big), & j\in D_i,\\[6pt]
\mathcal N\big(\mu_{ij},\ \sigma_{ij}^2\big), & j\in C_i.
\end{cases}$$

The structural-equation representation follows \citet{hannig_et_al_2016}. For a continuous contribution,
$$Y_{ij}=G_{ij}^{(C)}(U_{ij}^{C},\btheta_j)=\mu_{ij}(\btheta_j)+\sigma_{ij}(\btheta_j)\,\Phi^{-1}(U_{ij}^{C}),
\qquad U_{ij}^{C}\sim\operatorname{Uniform}(0,1),\ j\in C_i,$$
where $\Phi$ is the standard normal distribution function. For a discrete contribution, we use the inverse \gls{nb} structural equation of the generalized fiducial construction for discrete distributions \citep{hannig_2013,hannig_et_al_2016},
$$Y_{ij}=G_{ij}^{(D)}(U_{ij}^{D},\btheta_j)=F_{\mathrm{NB}}^{-1}\big(U_{ij}^{D};\,r_{ij},\,\nu_j/(1+\nu_j)\big),
\qquad U_{ij}^{D}\sim\operatorname{Uniform}(0,1),\ j\in D_i,$$
where $F_{\mathrm{NB}}^{-1}$ is the generalized inverse of the \gls{nb} distribution function, so that $G_{ij}^{(D)}(U_{ij}^{D},\btheta_j)=y_{ij}$ if and only if $F_{\mathrm{NB}}(y_{ij}-1;\btheta_j)<U_{ij}^{D}\leq F_{\mathrm{NB}}(y_{ij};\btheta_j)$, with $F_{\mathrm{NB}}(-1;\btheta_j)=0$. Collecting the auxiliary variables gives $\mathbf Y_C=G_C(\mathbf U_C,\btheta)$ and $\mathbf Y_D=G_D(\mathbf U_D,\btheta)$, where $\mathbf U_C$ and $\mathbf U_D$ are independent vectors of \gls{iid} $\operatorname{Uniform}(0,1)$ variables, held fixed when the continuous structural equations are differentiated with respect to $\btheta$. 

In practice $\btheta_0$ is unknown, so the oracle partition must be estimated, and an estimated shape involves the dispersion. The essential requirement is that a component not be assigned by its own realized response, which would let the response choose its own model and bias the working likelihood. A rule based on an aggregate estimator whose sampling variability vanishes converges instead to a fixed partition determined by the design alone, although not necessarily to the oracle one. The asymptotic results of Section~\ref{sec:asymptotics} are stated for the oracle partition and do not account for the variability of an estimated rule. Section~\ref{sec:sims:partition} compares candidate rules by simulation and Section~\ref{sec:app:fit} describes the rule used in the application.

\subsection{Computing the Discrete Construction}
\label{sec:method:discrete}

For a component $j$ with $j\in D_i$ for every observation, $\btheta_j$ appears in no continuous structural equation and Theorem~\ref{thm:jointdensity} below does not apply. Such a component is handled by the discrete construction of \citet{hannig_2013}. Fix one, suppress its index, and restrict $\btheta$ to a bounded set $\Theta$ given in \discretecompapp. Let $\mathbf U^\star$ be a vector of $n$ independent uniforms, drawn independently of the data. For a value $\mathbf u$ of $\mathbf U^\star$, let $\mathcal Q(\mathbf y,\mathbf u)=\{\btheta\in\Theta: G_D(\mathbf u,\btheta)=\mathbf y\}$ be the set of parameters that reproduce every observed count from $\mathbf u$: each observation $i$ restricts $\btheta$ to the values that reproduce $y_i$ from $u_i$, and $\mathcal Q(\mathbf y,\mathbf u)$ is where all $n$ restrictions hold at once. For most $\mathbf u$ no such $\btheta$ exists and the set is empty. When it is nonempty, it is a region rather than a point, since each count is compatible with a range of parameter values. The discrete \gls{gfd} is the distribution of the point produced by the following procedure: draw $\mathbf U^\star$, keep it only if $\mathcal Q(\mathbf y,\mathbf U^\star)$ is nonempty, and pick one point from that set by a rule $V$ \citep[eq.~(2.2)]{hannig_2013}. The rule $V$ is a choice. For \gls{iid} data discretized by a fixed finite partition, Theorem~4.1 of \citet{hannig_2013} shows that it is asymptotically immaterial. At fixed $\nu$, each observation with $y_i>0$ bounds $X_i^\top\beta$ from both sides and each with $y_i=0$ bounds it from above only, so $\mathcal Q(\mathbf y,\mathbf u)$ is a family of polytopes in $\beta$, one for each feasible value of $\log\nu$. Our rule $V$ draws $\log\nu$ uniformly on its feasible range and then $\beta$ uniformly on the polytope at that $\log\nu$. We write $v_{\mathbf u}$ for the density of this selection.

Let $\mathcal B(\btheta)=\{\mathbf u: G_D(\mathbf u,\btheta)=\mathbf y\}$ be the set of uniforms that reproduce $\mathbf y$ under $\btheta$. Then $\btheta\in\mathcal Q(\mathbf y,\mathbf u)$ if and only if $\mathbf u\in\mathcal B(\btheta)$, and since $G^{(D)}_{ij}(u_i,\btheta)=y_i$ exactly when $u_i$ lies in the cell $\big(F_{\mathrm{NB}}(y_i-1;\btheta),F_{\mathrm{NB}}(y_i;\btheta)\big]$, $\mathcal B(\btheta)$ is a product of these cells, a box with volume $L(\mathbf y\mid\btheta)$.

\begin{lemma}[Discrete Generalized Fiducial Density]
\label{lem:discrete}
If $V$ has a density $v_{\mathbf u}$ on $\mathcal Q(\mathbf y,\mathbf u)$ for almost every $\mathbf u$ with $\mathcal Q(\mathbf y,\mathbf u)\neq\varnothing$, the discrete \gls{gfd} has density
$$r(\btheta\mid\mathbf y)\;\propto\; L(\mathbf y\mid\btheta)\,\mathbb E_{\btheta}\big[v_{\mathbf U}(\btheta)\big],\qquad \mathbf U\sim\operatorname{Uniform}(\mathcal B(\btheta)),$$
where $r$ and $v_{\mathbf u}$ are densities with respect to the same measure on $\Theta$, and is the $\btheta$-margin of $p(\btheta,\mathbf u)\propto v_{\mathbf u}(\btheta)\,\mathbf 1\{\mathbf u\in\mathcal B(\btheta)\}$, whose $\mathbf u$-margin is uniform on
$\{\mathbf u:\mathcal Q(\mathbf y,\mathbf u)\neq\varnothing\}.$
\end{lemma}

The discrete \gls{gfd} is thus the likelihood times a correction. The likelihood is the volume of auxiliary values consistent with $\btheta$, as in Theorem~\ref{thm:jointdensity}. The correction $\mathbb E_{\btheta}[v_{\mathbf U}(\btheta)]$, the expected reciprocal size of the consistent set, plays the part of the Jacobian. If it were constant in $\btheta$, the discrete \gls{gfd} would equal the flat-prior posterior on $(\beta,\log\nu)$. The proof is a change in the order of integration (see \discretecompapp). The joint density also yields a sampler. Updating one uniform at a time \citep{hannig_iyer_wang_2007} mixes too slowly to be usable at $n=216$: each update is confined to a set of diameter $O_P(n^{-1})$ \citep[eq.~(B.5)]{hannig_2013}, while the target has spread $O(n^{-1/2})$. We instead run \gls{mh} on $\btheta$ and the uniforms together, proposing a new $\btheta$ and carrying each uniform along at its relative position within its cell, an idea used in a binomial setting by \citet[App.~D]{murph_hannig_williams_2023}. Both samplers agree with exact rejection draws on a small problem. The \gls{mh} sampler, its validation and its cost are given in \discretecompapp.

\section{Derivation of the Fiducial Density}
\label{sec:theory}

In this section we derive the joint generalized fiducial density for the mixed discrete--continuous construction. We write $r(\btheta\mid\mathbf y)$ for the density and reserve \gls{gfd} for the distribution itself. In this section and Section~\ref{sec:asymptotics} we assume that every component is \emph{active}: $P_Q(j\in C_i)>0$ for every $j\in\{1,\ldots,d\}$, so that each component receives the normal approximation for a positive fraction of observations. Entirely discrete components are handled by the construction of Section~\ref{sec:method:discrete} and are excluded from $d$ and from $q=d(K+1)$. We refer to this as the active-component convention. It is what makes the full-column-rank condition of Theorem~\ref{thm:jointdensity} attainable: the $K+1$ parameters of an entirely discrete component would appear in no continuous structural equation, and the continuous Jacobian matrix $J_C$ defined below would have a nontrivial null space.

\subsection{Joint Fiducial Density}
\label{sec:jointFiducial}

Throughout this subsection we condition on the realized design $\mathbf Z=(Z_1,\ldots,Z_n)$ and on the oracle partition determined at $\btheta_0$, which is held fixed as $\btheta$ varies. Stacking the $nd$ scalar responses with the normal-approximated ones first gives
$$\mathbf y=\begin{pmatrix}\mathbf y_C\\ \mathbf y_D\end{pmatrix},
\qquad \mathbf y_C\in\mathbb R^{n_C},\qquad \mathbf y_D\in\mathbb N_0^{n_D},
\qquad n_C+n_D=nd,$$
where $n_C=\sum_i|C_i|$ and $n_D=\sum_i|D_i|$, with the corresponding structural-equation representation
$$\mathbf Y=\bigl(G_C(\mathbf U_C,\btheta),\,G_D(\mathbf U_D,\btheta)\bigr),$$
where $\mathbf U_C$ and $\mathbf U_D$ are the independent vectors of uniforms introduced in Section~\ref{sec:method}.

\begin{theorem}[Joint Generalized Fiducial Density]
\label{thm:jointdensity}
Suppose that
\begin{enumerate}[(i)]
  \item the design pairs $Z_i=(X_i,t_i)$, $i=1,\ldots,n$, satisfy Condition~(A) of Section~\ref{sec:modelconditions}, and the parameter space $\Theta$ satisfies Condition~(D);
  \item conditional on $\mathbf Z$ and the oracle partition, the observation vectors $\mathbf Y_1,\ldots,\mathbf Y_n$ are independent under the hybrid working model of Section~\ref{sec:method}, and the scalar responses are conditionally independent within each observation;
  \item the continuous structural-equation Jacobian
  $$J_C(\mathbf u_C,\btheta)=\frac{\partial G_C(\mathbf u_C,\btheta)}{\partial\btheta}$$
  has rank $q$ almost surely in $\mathbf U_C$ for every $\btheta\in\Theta$, and its evaluation at the auxiliary values that reproduce the observed $\mathbf y_C$,
  $$J_C(\btheta):=J_C(\mathbf u_C,\btheta)\big|_{\mathbf u_C=G_C^{-1}(\mathbf y_C,\btheta)},$$
  has full column rank $q$ for every $\btheta\in\Theta$; and
  \item for every $\mathbf u_C$ in $\mathcal U_C(\mathbf y_C):=\{G_C^{-1}(\mathbf y_C,\btheta):\btheta\in\Theta\}$, the map $\btheta'\mapsto G_C(\mathbf u_C,\btheta')$ is injective on $\Theta$.
\end{enumerate}
Then the generalized fiducial density of $\btheta$ is
$$r(\btheta\mid\mathbf y,\mathbf Z)\propto L(\mathbf y\mid\btheta,\mathbf Z)\,J(\btheta,\mathbf y_C,\mathbf Z),$$
where $L(\mathbf y\mid\btheta,\mathbf Z)$ is the conditional working likelihood of the hybrid model and
$$J(\btheta,\mathbf y_C,\mathbf Z)=\sqrt{\det\!\left(n^{-1}J_C(\btheta)^\top J_C(\btheta)\right)}$$
is the Jacobian induced by the continuous structural equations.
\end{theorem}

\begin{remark}
The injectivity in (iv) does not follow from the rank condition in (iii). If the design includes an intercept, shifting the intercept by $c$ and multiplying $\nu_j$ by $e^{c}$ leaves every $\mu_{ij}$ unchanged, so when every continuous auxiliary variable sits at its median, $\Phi^{-1}(U^C_{ij})=0$, two different parameters produce the same $\mathbf Y_C$. This configuration has probability zero, and the requirement holds almost surely under the working model for every component with more than $2(K+1)$ normal-approximated observations, as shown in the remark in \jointfiducialproofapp.
\end{remark}

\begin{proof}[Proof sketch]
A full derivation appears in \jointfiducialproofapp. Because the design distribution $Q$ does not depend on $\btheta$, it contributes to neither the likelihood ratio nor the Jacobian.

For the discrete components, the responses lie on the integer lattice, where distinct points are at least one unit apart, so a tolerance $\epsilon<1$ forces exact matching, $G_D(\mathbf U_D,\btheta)=\mathbf y_D$. The set of $\btheta$ that satisfy this constraint depends on $\mathbf U_D$, so one cannot simply replace the constraint by its probability at a fixed $\btheta$. The supplementary proof first uses injectivity (iv) and compactness to rule out parameters far from the reference value that nonetheless produce a small continuous residual. It then uses the rank condition (iii) and a uniform Taylor expansion to show that the conditioning event confines $\btheta$ to an $O(\epsilon)$ neighborhood. Over that neighborhood the discrete constraint is unchanged except when some $U^D_{ij}$ lies within $O(\epsilon)$ of a cell boundary, an event of probability $O(\epsilon)$ that vanishes in the limit. Integrating the constraint over $\mathbf U_D$ at fixed $\btheta$ then yields the discrete likelihood factor $L_D(\mathbf y_D\mid\btheta,\mathbf Z)$ and no Jacobian.

For the continuous components, the rank condition and the regularity of the continuous structural equations give the standard contribution $L_C(\mathbf y_C\mid\btheta,\mathbf Z)\,J(\btheta,\mathbf y_C,\mathbf Z)$, with $J$ as in the statement.

Combining the two, using the independence of $\mathbf U_C$ and $\mathbf U_D$ and the conditional independence of the responses,
$$r(\btheta\mid\mathbf y,\mathbf Z)\propto L_D(\mathbf y_D\mid\btheta,\mathbf Z)\,L_C(\mathbf y_C\mid\btheta,\mathbf Z)\,J(\btheta,\mathbf y_C,\mathbf Z)=L(\mathbf y\mid\btheta,\mathbf Z)\,J(\btheta,\mathbf y_C,\mathbf Z).$$
\end{proof}

The discrete observations therefore enter only through the likelihood, and the Jacobian depends only on the continuous structural equations.

\subsection{Hybrid Likelihood}
\label{sec:likelihood}

Conditional on the covariates, exposures and oracle partition, and with the conditioning on $\mathbf Z$ suppressed from the notation, the working hybrid likelihood is
$$L(\mathbf y\mid\btheta)
= \prod_{i=1}^n
\left[ \prod_{j\in D_i} f_{\mathrm{NB}}(y_{ij}\mid\btheta_j) \right]
\left[ \prod_{j\in C_i} f_{\mathrm N}(y_{ij}\mid\btheta_j) \right] 
= L_D(\mathbf y_D\mid\btheta)\,L_C(\mathbf y_C\mid\btheta),$$
where $f_{\mathrm{NB}}(\cdot\mid\btheta_j)$ is the \gls{nb} \gls{pmf} with shape $r_{ij}$ and success probability $\nu_j/(1+\nu_j)$, as derived in \nbderivationapp, and $f_{\mathrm N}(\cdot\mid\btheta_j)$ is the normal density with mean $\mu_{ij}$ and variance $\sigma_{ij}^2$, all as in Section~\ref{sec:method}. Under the conditions of Theorem~\ref{thm:jointdensity}, substituting this decomposition into the theorem gives
$$r(\btheta\mid\mathbf y)
\propto
L_D(\mathbf y_D\mid\btheta)\,
L_C(\mathbf y_C\mid\btheta)\,
J(\btheta,\mathbf y_C).$$
The explicit \gls{nb} and normal factors, and the working log-likelihood obtained from them, are given in \scorehessianapp, where they are used to derive the score and Hessian. With the likelihood specified, what remains is the Jacobian factor $J(\btheta,\mathbf y_C)$.

\subsection{Jacobian Structure}
\label{sec:jacobian}

We now derive the explicit form of $J(\btheta,\mathbf y_C)$ under the conditions of Theorem~\ref{thm:jointdensity}, with $J_C(\btheta)$ and $J(\btheta,\mathbf y_C)$ as defined there. Write $\xi_{ij}:=\Phi^{-1}(U_{ij}^{C})\sim\mathcal N(0,1)$, so that the continuous structural equation of Section~\ref{sec:method} reads
$$Y_{ij}=\mu_{ij}(\btheta_j)+\sigma_{ij}(\btheta_j)\,\xi_{ij},\qquad j\in C_i,$$
with $\mu_{ij}(\btheta_j)$ and $\sigma_{ij}(\btheta_j)$ as given there. Each continuous structural equation involves only its own component's parameter $\btheta_j=(\beta_j^\top,\nu_j)^\top$, so the derivative of the $(i,j)$ equation with respect to $\btheta_{j'}$ is zero for $j'\neq j$. This gives $J_C(\btheta)$ a block structure, which we describe next.

\begin{proposition}[Block-Diagonal Jacobian]
\label{prop:jacobian}
Under the hybrid model of Section~\ref{sec:method} and the conditions of Theorem~\ref{thm:jointdensity}, order the continuous contributions by component. Then $J_C(\btheta)$ is block diagonal with rectangular blocks,
$$J_C(\btheta)=\operatorname{diag}\{J_1(\btheta),\ldots,J_d(\btheta)\},$$
where $J_j(\btheta)$ is the $n_j^C\times(K+1)$ matrix of derivatives of the continuous equations of component $j$ with respect to $\btheta_j=(\beta_j^\top,\nu_j)^\top$, and $n_j^C=|\{i:j\in C_i\}|$ is the number of observations in which component $j$ is normal-approximated. The full-column-rank condition implies $\operatorname{rank}\{J_j(\btheta)\}=K+1$, and hence $n_j^C\geq K+1$, for every active component.

Define weight vectors $\mathbf w_j^{(\beta)},\mathbf w_j^{(\nu)}\in\mathbb R^{n_j^C}$, indexed by the observations $i$ with $j\in C_i$, by
$$w_{ij}^{(\beta)}=\tfrac12(y_{ij}+\mu_{ij}),
\qquad
w_{ij}^{(\nu)}=-\frac{\nu_j\mu_{ij}+(\nu_j+2)\,y_{ij}}{2\nu_j(\nu_j+1)}.$$
Then
$$J_j(\btheta)=
\begin{bmatrix}
W_j^{(\beta)}X^{(j)} & \mathbf w_j^{(\nu)}
\end{bmatrix},$$
where $W_j^{(\beta)}=\operatorname{diag}(\mathbf w_j^{(\beta)})$ and $X^{(j)}\in\mathbb R^{n_j^C\times K}$ is the submatrix of $\mathbf X$ with the rows $i$ for which $j\in C_i$.
\end{proposition}

The proof is in \blockdiagjacobianproofapp.

The block structure factors the Jacobian over components:
$$J(\btheta,\mathbf y_C)
=\sqrt{\det\!\left(n^{-1}J_C(\btheta)^\top J_C(\btheta)\right)}
=\prod_{j=1}^d\sqrt{\det\widehat V_j(\btheta)},
\qquad
\widehat V_j(\btheta):=n^{-1}J_j(\btheta)^\top J_j(\btheta).$$
Each factor depends only on $\btheta_j$, which is what allows the Jacobian to be computed component by component. The scaling by $n^{-1}$ is the one under which $\widehat V_j(\btheta)$ has a finite limit: under the conditions of Section~\ref{sec:asymptotics}, $\widehat V_j(\btheta)\to V_j(\btheta)$ uniformly in probability on a neighborhood of $\btheta_0$, and so
$$J(\btheta,\mathbf y_C)\to\pi(\btheta):=\prod_{j=1}^d\sqrt{\det V_j(\btheta)},$$
the limiting Jacobian density that plays the role of a prior in the \gls{bvm} theorem of Section~\ref{sec:asymptotics}.

\section{Asymptotic Justification}
\label{sec:asymptotics}

We establish a \gls{bvm} theorem for the hybrid \gls{gfd} of Section~\ref{sec:jointFiducial} under the conditions of Theorem~\ref{thm:jointdensity} and the active-component convention of Section~\ref{sec:theory}. We work within the framework of \citet{borgert_hannig_2026}. The model conditions and the verification of their five regularity conditions are given below and in \bvmproofapp.

\subsection{Model Conditions}
\label{sec:modelconditions}

\begin{enumerate}[(A)]

    \item \textbf{(Design regularity)} The design pairs $(X_i,t_i)$ are \gls{iid} draws from a common distribution $Q$ on $\mathbb R^K\times(0,\infty)$. The marginal distribution of $X_i$ is supported on a compact set $\mathcal X\subset\mathbb R^K$. In particular, there exists a constant $C<\infty$ such that $\|X_i\|_2\leq C$ almost surely. No independence between $X_i$ and $t_i$ is required.

    \item \textbf{(Bounded exposure)} There exist constants $0<t_{\min}\leq t_{\max}<\infty$ such that
    $$P(t_{\min}\leq t_i\leq t_{\max})=1.$$

    \item \textbf{(Design richness)} The following design nondegeneracy conditions hold:
    \begin{enumerate}[(i)]
        \item The population second-moment matrix of the covariates is positive definite:
        $$\mathbb E[X_iX_i^\top]\succ0.$$

        \item For each active component $j$, the design restricted to the observations in which it is normal-approximated is nondegenerate:
        $$\mathbb E\!\left[\mathbf1\{j\in C_i\}X_iX_i^\top\right]\succ0.$$
    \end{enumerate}

    \item \textbf{(Compact valid parameter space and interior true parameter)} The parameter space $\Theta$ is a compact subset of
    $$\prod_{j=1}^d\left(\mathbb R^K\times(0,\infty)\right),$$
    has nonempty interior, and satisfies $\btheta_0\in\operatorname{int}(\Theta)$. Consequently, there exist constants $0<\nu_{\min}\leq\nu_{\max}<\infty$ such that
    $$\nu_{\min}\leq\nu_j\leq\nu_{\max}$$
    uniformly over $\btheta\in\Theta$ and $j\in\{1,\ldots,d\}$.

    \item \textbf{(Information positivity)} The population Fisher information matrix
    $$\mathcal I_{\btheta_0} :=\mathbb E_Q[\mathcal I_i(\btheta_0\mid Z_i)],$$
    where $\mathcal I_i(\btheta_0\mid Z_i)$ is the observation-level conditional Fisher information derived in \scorehessianapp, is positive definite.

\end{enumerate}

\noindent Conditions~(A), (B) and~(D) bound the design, the exposures, the component means and variances, and the dispersions uniformly, which the score, Hessian and Jacobian arguments all require. Condition~(C) supplies the design nondegeneracy needed for identifiability, with (C)(ii) required only within the normal-approximated observations of active components. Condition~(E) ensures the limiting Fisher information is nonsingular. The convergence of the averaged Fisher information and of the continuous-block Gram matrices $\widehat V_j(\btheta)$ is established in \bvmproofapp. The following lemma records the uniform positive definiteness of their limits that the Jacobian analysis requires.

\begin{lemma}\label{lem:posdef}
Under Conditions~(A)--(E), for every active component $j$ the limiting matrix $V_j(\btheta)$ of Section~\ref{sec:jacobian} is positive definite for every $\btheta\in\Theta$, and
$$\inf_{\btheta\in\Theta}\lambda_{\min}\{V_j(\btheta)\}>0.$$
\end{lemma}

The proof is in \posdefproofapp.

\subsection{Main Asymptotic Result}

Under Conditions~(A)--(E), the conditions of Theorem~\ref{thm:jointdensity} and the active-component convention of Section~\ref{sec:theory}, the five regularity conditions of \citet{borgert_hannig_2026} hold for the hybrid working model. They concern the local asymptotic normality of the model sequence, uniform convergence of the Jacobian to a limiting density, positivity and continuity of that density, a global bound on the Jacobian, and the existence of exponentially consistent tests. Their verification is in \bvmproofapp, and Lemma~\ref{lem:posdef} supplies the positivity of the limiting Jacobian density that the second and third conditions require. Theorem~3.1 of \citet{borgert_hannig_2026} then gives the following.

\begin{theorem}[Bernstein--von Mises for the Hybrid Generalized Fiducial Distribution]
\label{thm:bvm}
Assume Conditions~(A)--(E) of Section~\ref{sec:modelconditions}, the conditions of Theorem~\ref{thm:jointdensity} and the active-component convention of Section~\ref{sec:theory}. Let $\bar{\btheta}$ be a draw from the \gls{gfd} and let $R_{\sqrt n(\bar{\btheta}-\btheta_0)\mid\mathbf y}$ denote the conditional law of $\sqrt n(\bar{\btheta}-\btheta_0)$ given $\mathbf y$. Then
$$\left\|R_{\sqrt n(\bar{\btheta}-\btheta_0)\mid\mathbf y}
-N\!\left(\fisher_{\btheta_0}^{-1}\Delta_{n,\btheta_0},\ \fisher_{\btheta_0}^{-1}\right)\right\|_{TV}
\xrightarrow{P_{\btheta_0}^{(n)}}0,$$
where $\Delta_{n,\btheta_0}=n^{-1/2}\sum_{i=1}^n s_i(\btheta_0)$ is the central sequence of \lanassump{} of \bvmproofapp, $s_i(\btheta_0)$ is the observation-level score derived in \scorehessianapp, and $\fisher_{\btheta_0}$ is the population Fisher information of Condition~(E).
\end{theorem}

\begin{proof}
See \bvmproofapp.
\end{proof}

The \gls{gfd} is therefore asymptotically normal, centered at $\btheta_0+n^{-1/2}\fisher_{\btheta_0}^{-1}\Delta_{n,\btheta_0}$, which is the one-step approximation to the maximum likelihood estimate, with covariance $n^{-1}\fisher_{\btheta_0}^{-1}$. The limiting Jacobian $\pi(\btheta)$ of Section~\ref{sec:jacobian} enters only as a prior would, and like a smooth prior it disappears in the limit, so the \gls{gfd} and the flat-prior posterior on $(\beta,\log\nu)$ agree to first order. Section~\ref{sec:sims} examines how close they are at finite $n$.

\section{Simulation Study}
\label{sec:sims}

We assess the finite-sample behavior of the hybrid generalized fiducial procedure by simulation, with three aims: to examine frequentist coverage as the sample size grows, to compare the fiducial intervals with a flat-prior Bayesian comparator, and to determine how the discrete--continuous partition should be formed in practice. Data are generated from the exact \gls{nb} model throughout, in three regimes defined by the true \gls{nb} shape: two in which some or all components are normal-approximated, so that coverage reflects the effect of the approximation, and one in which every component is \gls{nb} in every observation and the discrete construction of Section~\ref{sec:method:discrete} applies.

\subsection{Design}
\label{sec:sims:design}

The simulation uses a single component ($d=1$), so we write $Y_i$ for the response and suppress the component index throughout this section. For each replication we draw a fresh design of size $n$: an intercept, one continuous covariate distributed as $\operatorname{Uniform}(-\sqrt3,\sqrt3)$, chosen so that it has mean zero and variance one, two indicators for a three-level categorical covariate with category probabilities $(0.40,0.30,0.30)$, and an exposure $t_i\sim\operatorname{Uniform}(60,90)$. A fresh \gls{iid} design in each replication matches the random-design and bounded-support conditions of Section~\ref{sec:modelconditions} and keeps the results from depending on any one realized design. The true regression coefficients are $\beta=(\beta_0,0.10,0.20,-0.15)$ on the log scale, with dispersion $\nu=0.05$, and the response is drawn from $Y_i\sim\operatorname{NB}\big(r_i=\exp(X_i^\top\beta)\,t_i,\ \nu/(1+\nu)\big)$.

To cover the discrete--continuous spectrum we calibrate the intercept $\beta_0$ to three regimes. In the heavy and borderline regimes it fixes the population geometric-mean count, and in the sparse regime the population geometric-mean \gls{nb} shape $r_i$. At $\nu=0.05$ the three regimes have shapes $5$, $25$ and $150$ and counts of about $100$, $500$ and $3000$:
\begin{itemize}
  \item \textbf{Sparse} (shape $5$, $\mu\approx100$): every observation is \gls{nb} and the fit is the discrete construction of Section~\ref{sec:method:discrete};
  \item \textbf{Borderline} (shape $25$, $\mu\approx500$): observations straddle the partition threshold, so a single fit mixes \gls{nb} and normal components;
  \item \textbf{Heavy} (shape $150$, $\mu\approx3000$): every observation is normal-approximated.
\end{itemize}
When $\nu$ varies (Sections~\ref{sec:sims:dispersion} and~\ref{sec:sims:partition}), the heavy and borderline regimes keep their count, so their shape moves with $\nu$, while the sparse regime keeps its shape. The regimes are named for their counts, but the shape governs which representation is appropriate (Section~\ref{sec:method}). The sparse regime's shape of $5$ is typical of the application's entirely discrete bins.

We use sample sizes $n\in\{50,100,200,500\}$ and $R=2000$ replications per cell, giving a Monte Carlo standard error of about $0.5\%$ for an empirical coverage near $95\%$. As a comparator we compute a flat-prior Bayesian posterior under the same working likelihood, by Stan in the heavy and borderline regimes and by an independence sampler on the \gls{nb} likelihood in the sparse regime. The fiducial and flat-prior Stan models sample $\log\nu$. The fiducial model multiplies the $\nu$-column of the Jacobian by $\nu$, which is the structural Jacobian in $(\beta,\log\nu)$ coordinates and supplies the change-of-variables factor, and the flat prior is flat on $(\beta,\log\nu)$. In the heavy and borderline regimes this comparison isolates the effect of the continuous Jacobian, and in the sparse regime that of the factor $\mathbb E_{\btheta}[v_{\mathbf U}(\btheta)]$ of Lemma~\ref{lem:discrete}. Fits whose chains had not converged after a refit are excluded: none in the coverage study, six of the $16{,}000$ in Section~\ref{sec:sims:dispersion}, which would change no cell by more than $0.001$, and the hybrid fits of Section~\ref{sec:sims:partition} described there. Every discrete fit is included (Sections~7 and~8 of the Supplementary Material \citep{thomas_hannig_supp}).

Sections~\ref{sec:sims:coverage} and~\ref{sec:sims:dispersion} fix $\nu$ within each design and form the partition by a mean threshold $\tau=500$, which at fixed $\nu$ is the shape rule of Section~\ref{sec:method} with $\omega=\tau\nu=25$. Section~\ref{sec:sims:partition} lets $\nu$ vary across designs, where the two rules differ, and compares partition rules.

\subsection{Coverage and Asymptotic Normality}
\label{sec:sims:coverage}

Across all $60$ combinations of regime, parameter and sample size, the empirical coverage of the $95\%$ fiducial intervals lay between $0.93$ and $0.96$, averaging $0.946$. \simulationapp{} gives the individual cells. Coverage was close to nominal in every regime: $0.947$ in the heavy regime, $0.942$ in the borderline regime and $0.950$ in the sparse regime. The regimes differ by less than one percentage point, and the borderline regime, where a single fit mixes \gls{nb} and normal components, is calibrated as well as the two homogeneous regimes. The flat-prior Bayesian comparator behaved almost identically, averaging $0.946$ overall with the same pattern across regimes. Fit to the same datasets, the two sets of $95\%$ intervals matched in median width to within $1.2\%$ in every cell and within $0.6\%$ at every sample size above $50$ (\simulationapp).

Figure~\ref{fig:bvm_local} displays the fiducial and flat-prior Bayesian distributions on the local scale $h=\sqrt n(\theta-\theta_0)$. The two differ slightly at the smallest sample size and become increasingly similar as $n$ grows, both approaching the normal limit of Theorem~\ref{thm:bvm}. Figure~\ref{fig:calibration} compares empirical and nominal coverage levels, and Figure~\ref{fig:widths} compares interval lengths.

\begin{figure}[tbp]
  \centering
  \includegraphics[width=4.9in]{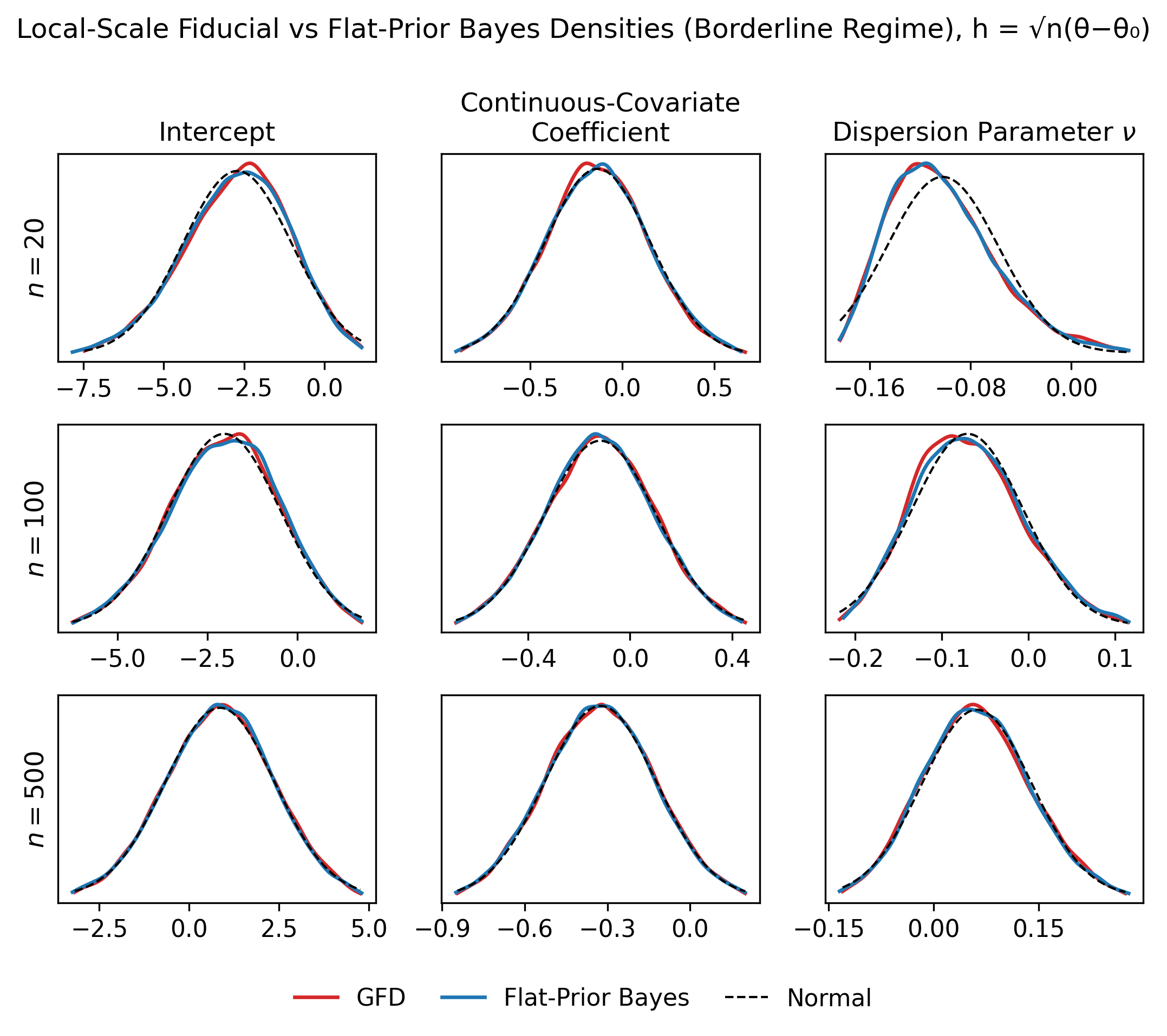}
    \caption{Fiducial (red) and flat-prior Bayesian (blue) distributions on the local scale $h=\sqrt{n}(\theta-\theta_0)$ in the borderline regime, for one simulated dataset at each of $n=20$, $100$ and $500$ (rows). The value $n=20$ lies below the coverage grid and is shown to make the small-sample difference visible. Columns are the intercept, the continuous-covariate coefficient and $\nu$. The dashed curve is a normal density matched to the fiducial draws in mean and standard deviation, so it checks shape rather than location or scale. The two distributions differ slightly at $n=20$, most visibly for $\nu$, and are nearly indistinguishable by $n=500$.}
  \label{fig:bvm_local}
\end{figure}

\begin{figure}[tbp]
  \centering
  \includegraphics[width=4.9in]{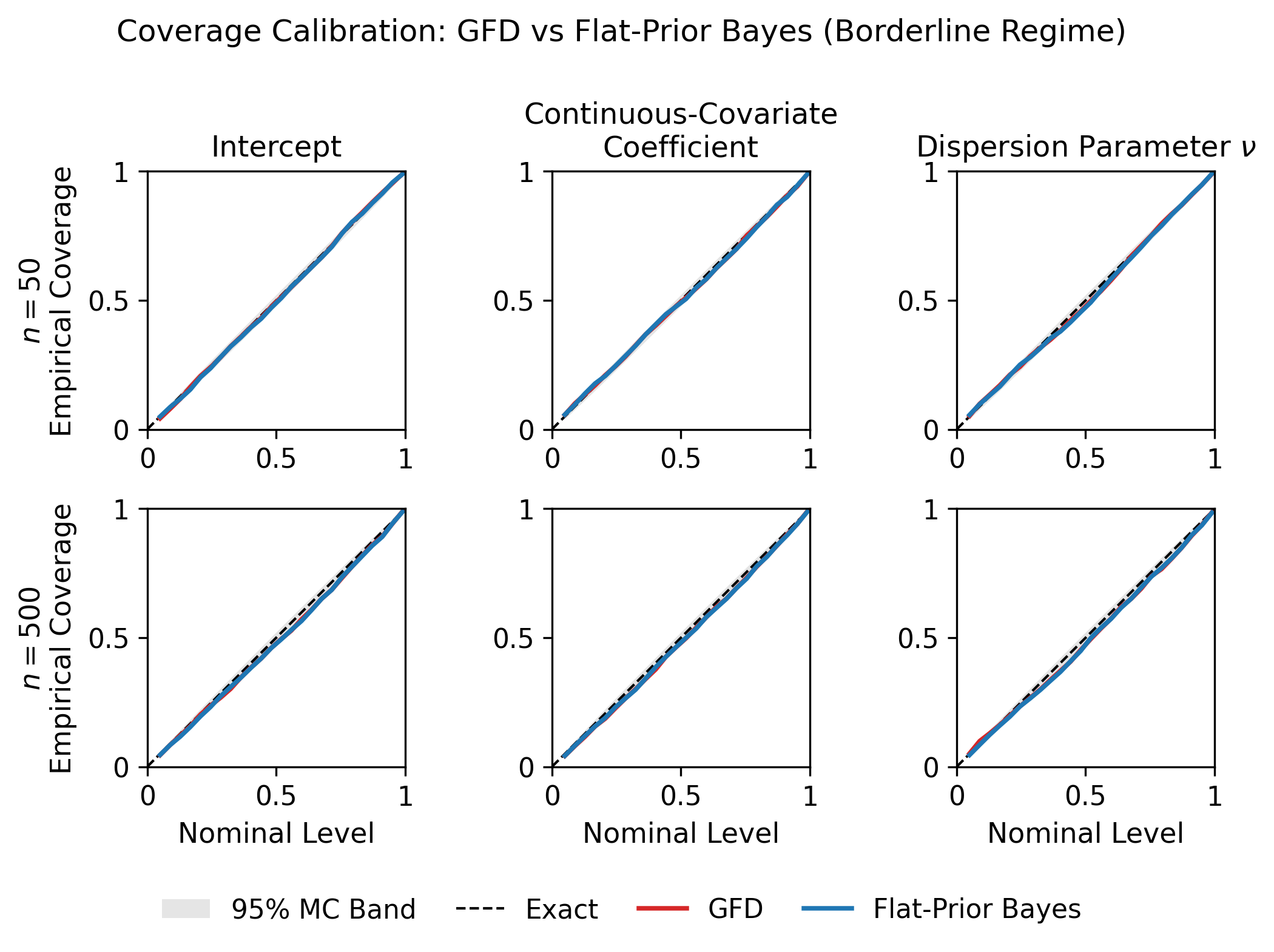}
  \caption{Empirical against nominal coverage for the fiducial (red) and flat-prior Bayesian (blue) intervals in the borderline regime, for the intercept, the continuous-covariate coefficient and dispersion parameter $\nu$ (columns) at $n=50$ and $500$ (rows). The dashed diagonal is exact calibration. The shaded region is a pointwise $95\%$ Monte Carlo band, $p\pm1.96\sqrt{p(1-p)/R}$ with $R=2000$. Both methods track the diagonal, sitting marginally below it at intermediate levels at $n=500$ by about the width of the band.}
  \label{fig:calibration}
\end{figure}

\begin{figure}[tbp]
  \centering
  \includegraphics[width=4.9in]{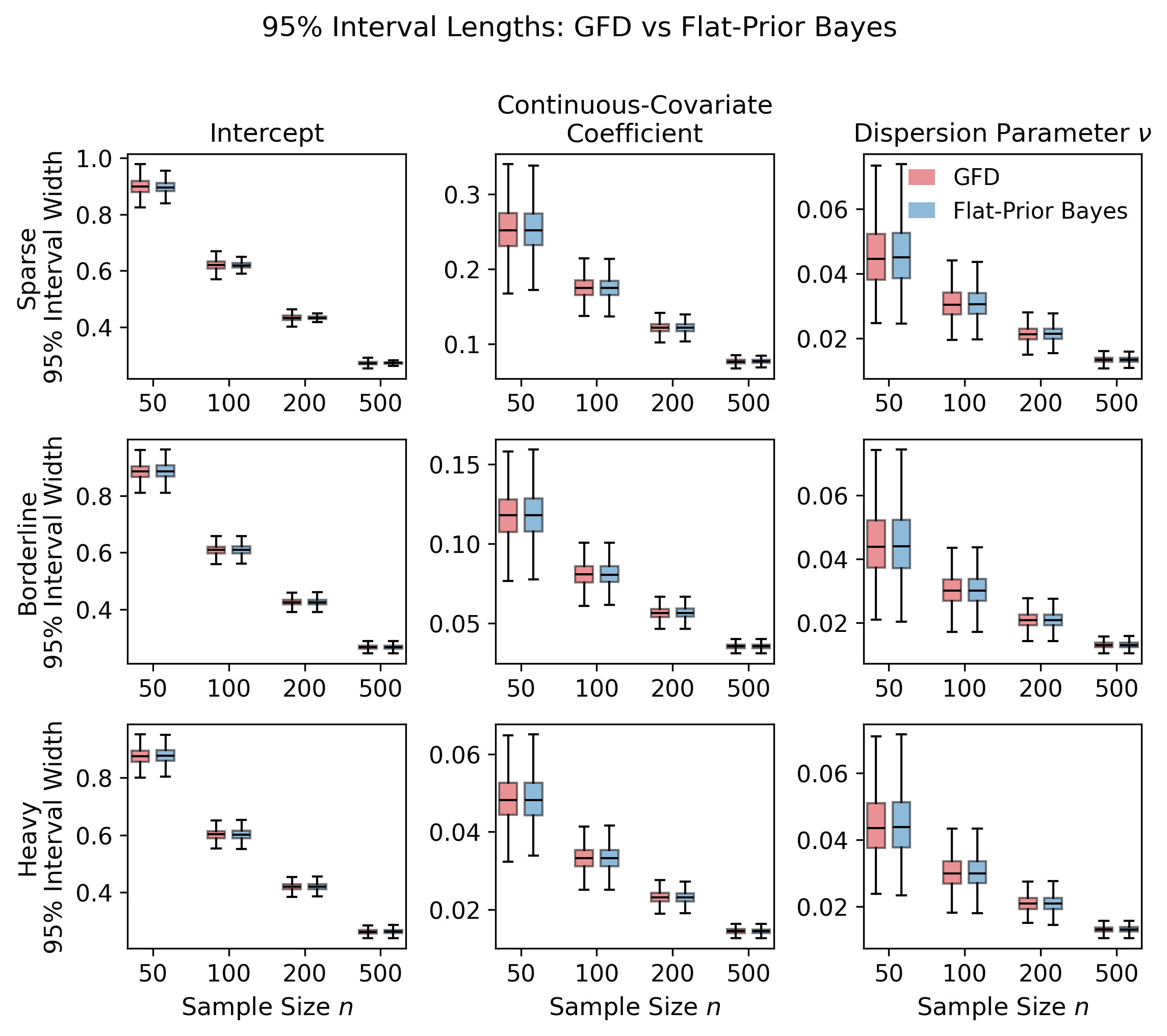}
    \caption{Distribution of $95\%$ interval lengths for the fiducial (red) and flat-prior Bayesian (blue) intervals, by regime (rows: sparse, borderline, heavy) and parameter (columns). In every regime the two methods give intervals of essentially equal length, contracting at the parametric rate as $n$ grows. In the sparse regime the fiducial intervals are from the discrete construction.}
  \label{fig:widths}
\end{figure}

\subsection{Sensitivity to the Dispersion}
\label{sec:sims:dispersion}

Under a partition by the true mean, an observation assigned to the normal component when $\mu_i\geq\tau$ has shape at least $\tau\nu$, a bound that depends on the dispersion (Section~\ref{sec:method}). The simulations above fix $\nu=0.05$, for which the shape at the threshold is $25$. To assess sensitivity to greater overdispersion, we repeat the heavy and borderline regimes at $\nu=0.02$ and $\nu=0.005$, which give shapes of $10$ and $2.5$ at the threshold. The sparse regime is omitted because it is calibrated to a fixed shape of $5$ and is fit entirely by the discrete construction at every $\nu$, so the normal approximation whose sensitivity is being tested never enters. All other settings are unchanged, and the regime intercepts are recalibrated to hold the target mean counts fixed. In $39$ borderline datasets with $n\le100$, the mean rule assigned fewer than $K+1$ observations to the normal component (too few for the rank condition of Theorem~\ref{thm:jointdensity}). As in Section~\ref{sec:sims:partition}, these are fit entirely by the discrete construction.

Table~\ref{tab:dispersion} reports average empirical coverage of the $95\%$ fiducial intervals across the five parameters, based on $R=1000$ replications per cell, except the $\nu=0.05$ rows, which reproduce the $R=2000$ main experiment. The Monte Carlo standard error for a coverage probability near $0.95$ is about $0.007$ ($0.005$ for the $\nu=0.05$ rows). At a fixed $\nu$, the shapes of the normal-assigned observations differ by regime: in the borderline regime they lie near the threshold, whereas in the heavy regime their means, and therefore their shapes, are substantially larger.

Coverage is close to nominal at $\nu=0.02$ and $\nu=0.05$ in both regimes. The lowest coverage for any single parameter in these settings is $0.925$. At $\nu=0.005$, coverage in the borderline regime falls by three to five percentage points, with averages from $0.900$ to $0.919$ and a lowest single-parameter coverage of $0.872$. The heavy regime is less affected, with averages from $0.929$ to $0.946$. In these simulations, coverage is close to nominal when the median \gls{nb} shape among normal-assigned observations is about $10$ or larger. Across those cells, fewer than $0.1\%$ of normal-assigned observations have shape below $7$, and none below $6$. Coverage deteriorates when the median shape falls to $2.7$, where shapes as small as $1.5$ are assigned to the normal component.

Coverage in the borderline $\nu=0.005$ cell does not improve with $n$: the distribution of assigned shapes does not change as $n$ grows, so the normal branch stays misspecified while the intervals contract. For the most affected parameter, the fiducial intervals are about $16\%$ too short at $n=50$ and $22\%$ too short at $n=500$ relative to the repeated-sampling variability under the \gls{nb} model. The coverage failures are split about evenly between the two tails, indicating underestimated uncertainty rather than directional bias. Section~\ref{sec:app:fit} relates the shapes assigned to the normal component in the application to these results.

\begin{table}[tbp]
\centering
\begin{tabular}{llccccc}
\toprule
& & & \multicolumn{4}{c}{Sample size $n$} \\
\cmidrule(lr){4-7}
Regime & $\nu$ & Normal-assigned shape: median (5th pct) & $50$ & $100$ & $200$ & $500$ \\
\midrule
\multirow{3}{*}{Borderline}
 & $0.005$ & $2.7$ \ \ ($2.0$) & 0.919 & 0.917 & 0.902 & 0.900 \\
 & $0.02$  & $10.9$ \ ($8.2$) & 0.945 & 0.940 & 0.938 & 0.940 \\
 & $0.05$  & $27.3$ ($20.5$) & 0.948 & 0.942 & 0.941 & 0.937 \\
\midrule
\multirow{3}{*}{Heavy}
 & $0.005$ & $15.0$ \ ($10.7$) & 0.946 & 0.943 & 0.929 & 0.943 \\
 & $0.02$  & $59.8$ ($42.9$) & 0.953 & 0.950 & 0.943 & 0.950 \\
 & $0.05$  & $149.7$ ($107.2$) & 0.950 & 0.944 & 0.948 & 0.947 \\
\bottomrule
\end{tabular}
\caption{Average empirical coverage of $95\%$ generalized fiducial intervals across the five parameters, by regime, dispersion and sample size. The shape column reports the median and fifth percentile of the \gls{nb} shape among observations assigned to the normal component. Coverage is close to nominal except in the borderline regime at $\nu=0.005$, where the assigned shapes are smallest.}
\label{tab:dispersion}
\end{table}

\subsection{Choice of Partition Rule}
\label{sec:sims:partition}

Assigning each observation by its own realized count, the most obvious rule, fails badly: intercept coverage falls from about $0.92$ at $n=50$ to $0.71$ at $n=500$ in the borderline regime, with a similar pattern for the dispersion and comparatively stable slopes (\simulationapp). The partition must not be formed from the realized response.

Assigning by an aggregate estimate instead avoids response-dependent assignment. With the pooled-rate estimate of the expected count,
$$\widehat\mu_i=\frac{t_i}{n}\sum_{k=1}^n\frac{Y_k}{t_k},$$
the \emph{mean rule} assigns observation $i$ to the normal component when $\widehat\mu_i\geq\tau$. The \emph{shape rule} assigns it when $\widehat\mu_i\widehat\nu\geq\omega$, where $\widehat\nu$ is the moment estimator implied by $\operatorname{Var}(Y_i)=\mu_i(1+\nu^{-1})$: with $\widehat m_i$ the fitted means from a Poisson regression of the counts on the covariates with a $\log t_i$ offset, $\widehat D=(n-K)^{-1}\sum_i(Y_i-\widehat m_i)^2/\widehat m_i$ and $\widehat\nu=1/(\widehat D-1)$. The \emph{lower-bound shape rule} assigns it when $\widehat\mu_i\widehat\nu\exp(-1.645\,\widehat s)\geq\omega$, where $\widehat s$ is a bootstrap standard error of the log estimated shape, so that sampling error in $\widehat\nu$ can move an observation only toward the exact \gls{nb} component. Under every rule, a component that would retain fewer than $K+1$ normal-assigned observations is modeled entirely by the \gls{nb} distribution, as the rank condition of Theorem~\ref{thm:jointdensity} requires.

To compare the rules when the dispersion varies, as it does across the application's bins, we draw $\nu$ in each replication from a log-uniform distribution on $[0.002,0.07]$, a range covering the fitted dispersions in Section~\ref{sec:app:fit}, and recalibrate the intercept to hold the regime's target fixed: the mean count in the heavy and borderline regimes, and the \gls{nb} shape in the sparse regime, whose mean therefore varies from about $70$ to $2{,}500$ while its shape stays at $5$. Each dataset is analyzed under the oracle shape rule, the mean rule with $\tau=500$, the shape rule with $\omega=10$, and the lower-bound shape rule with $\omega=10$ and $200$ bootstrap resamples, with $R=1000$ paired replications per cell. Entirely \gls{nb} fits use the discrete construction, fit once to each dataset and recorded against every rule that assigned nothing to the normal component. Of the $48{,}000$ rule--dataset fits, $25{,}999$ assign at least one observation to the normal component, and the $118$ of these that had not converged after a refit ($0.45\%$) are excluded. Most are heavy-regime oracle fits, whose retention would lower those cells by up to $0.022$, and every other cell moves by less than $0.002$ (\simulationapp). The other $22{,}001$ are entirely \gls{nb}, and all are included. In the sparse regime every non-converged hybrid fit is a mean-rule fit, so excluding them does not account for the mean rule's shortfall there.

Table~\ref{tab:partition} reports average coverage across the five parameters. In the sparse regime every observation has shape $5$, below $\omega=10$, so the three shape-based rules assign almost nothing to the normal component and coincide, and these fits are entirely \gls{nb}. The mean rule instead assigns nearly half, because a shape of $5$ gives a mean above $\tau=500$ whenever $\nu\lesssim0.01$, and it under-covers at every sample size. The shortfall falls on the slope coefficients, and only in the replications where an assignment occurs: their coverage there is about $0.91$ at all four sample sizes, against about $0.95$ where nothing is assigned. A large mean does not license the normal approximation when the shape is small.

In the heavy regime the mean rule assigns every observation to the normal component, well above the oracle, but coverage is unaffected, because large means keep the shape far above the range where the normal approximation fails. In the borderline regime it assigns half again as many observations as the oracle and under-covers at every sample size, with its lowest single-parameter coverage falling to $0.897$ at $n=500$. The shortfall does not diminish with $n$, consistent with Section~\ref{sec:sims:dispersion}. The shape rule matches the oracle to within $0.002$ in every cell, and the lower-bound rule to within $0.004$ while assigning fewer observations to the normal component. A partition by the mean is therefore harmless or harmful depending on the shape it happens to select, whereas a partition by the shape tracks the oracle everywhere. We use the lower-bound shape rule in the application.

\begin{table}[tbp]
\centering
\begin{tabular}{llccccc}
\toprule
& & & \multicolumn{4}{c}{Sample size $n$} \\
\cmidrule(lr){4-7}
Regime & Partition rule & Normal-assigned & $50$ & $100$ & $200$ & $500$ \\
\midrule
\multirow{4}{*}{Sparse}
 & Oracle shape             & $0\%$         & 0.954 & 0.948 & 0.950 & 0.955 \\
 & Mean, $\tau=500$         & $46$--$50\%$  & 0.940 & 0.934 & 0.937 & 0.939 \\
 & Shape, $\omega=10$       & $0$--$0.5\%$  & 0.954 & 0.948 & 0.950 & 0.955 \\
 & Lower-bound shape        & $0\%$         & 0.954 & 0.948 & 0.950 & 0.955 \\
\midrule
\multirow{4}{*}{Borderline}
 & Oracle shape             & $36$--$37\%$  & 0.944 & 0.952 & 0.945 & 0.948 \\
 & Mean, $\tau=500$         & $54$--$55\%$  & 0.927 & 0.934 & 0.919 & 0.917 \\
 & Shape, $\omega=10$       & $37\%$        & 0.945 & 0.951 & 0.945 & 0.948 \\
 & Lower-bound shape        & $26$--$34\%$  & 0.947 & 0.953 & 0.946 & 0.950 \\
\midrule
\multirow{4}{*}{Heavy}
 & Oracle shape             & $84$--$87\%$  & 0.949 & 0.950 & 0.943 & 0.949 \\
 & Mean, $\tau=500$         & $100\%$       & 0.947 & 0.949 & 0.940 & 0.946 \\
 & Shape, $\omega=10$       & $85$--$89\%$  & 0.949 & 0.950 & 0.943 & 0.951 \\
 & Lower-bound shape        & $75$--$83\%$  & 0.953 & 0.950 & 0.944 & 0.949 \\
\bottomrule
\end{tabular}
\caption{Average empirical coverage of $95\%$ generalized fiducial intervals across the five parameters when the dispersion varies across replications ($\nu$ log-uniform on $[0.002,0.07]$), by regime, partition rule and sample size, from $R=1000$ paired replications per cell with Monte Carlo standard error about $0.007$. The normal-assigned column gives the range, over sample sizes, of the average fraction of observations assigned to the normal component.}
\label{tab:partition}
\end{table}

\section{Application: Professional Women's Soccer}
\label{sec:app}

This section describes the data, the fitted model and its partition, the multiplicity adjustment, and the estimated effects.

\subsection{Data and Quantile Cube Construction}
\label{sec:app:data}

We apply the \gls{hgfi} framework to GPS tracking data, obtained under \gls{irb} 25-1234, from every match of one season for a professional women's soccer team. We kept the sessions in which an athlete played at least $60$ minutes of the match, and the athletes with at least three such sessions. This gave $17$ athletes, $26$ matches and $216$ athlete--match sessions. Not every athlete played in every included match.

Each raw GPS record corresponds to one athlete in one match and contains ten observations per second: a timestamp with the athlete's longitude and latitude. Following \citet{thomas_hannig_2026}, each session was converted to a quantile cube with cutoffs computed from the pooled match-play data of all sessions. Velocity and acceleration were each cut at their empirical quintiles. The movement angle, the direction of the acceleration relative to the direction of travel, was measured clockwise from a $-41^\circ$ baseline and cut at empirical angle quartiles, so that the four sectors correspond to forward, right, backward and left, giving $d=5\times5\times4=100$ movement bins. The response for session $i$ is the vector $\mathbf y_i\in\mathbb N_0^{100}$ of decisecond counts in each bin, accumulated over the whole session: warm-up, both halves and cool-down. Table~\ref{tab:vel_acc_quantiles} and Figure~\ref{fig:angle_quantiles} show the cutoffs.

\begin{table}[tbp]
\centering
\begin{tabular}{lccccc}
\toprule
Quantile (\%) & 0 & 20 & 40 & 60 & 80 \\
\midrule
Velocity (m/s)                          & 0.0100 & 0.3475 & 0.9721 & 1.4704 & 2.5606 \\
Acceleration (m/s\textsuperscript{2})   & 0.0010 & 3.8354 & 6.0756 & 8.8336 & 13.1388 \\
\bottomrule
\end{tabular}
\caption{Lower edges of the five velocity and acceleration bins, the empirical quantiles of the pooled first- and second-half data across all sessions. The cutoffs are computed from match play alone but applied to the whole session, so the bins hold equal time on match play (to within $0.1$ percentage points) but not on the analysis set.}
\label{tab:vel_acc_quantiles}
\end{table}

\begin{figure}[tbp]
  \centering
  \begin{minipage}[t]{0.40\textwidth}
    \centering
    \includegraphics[height=2.3in]{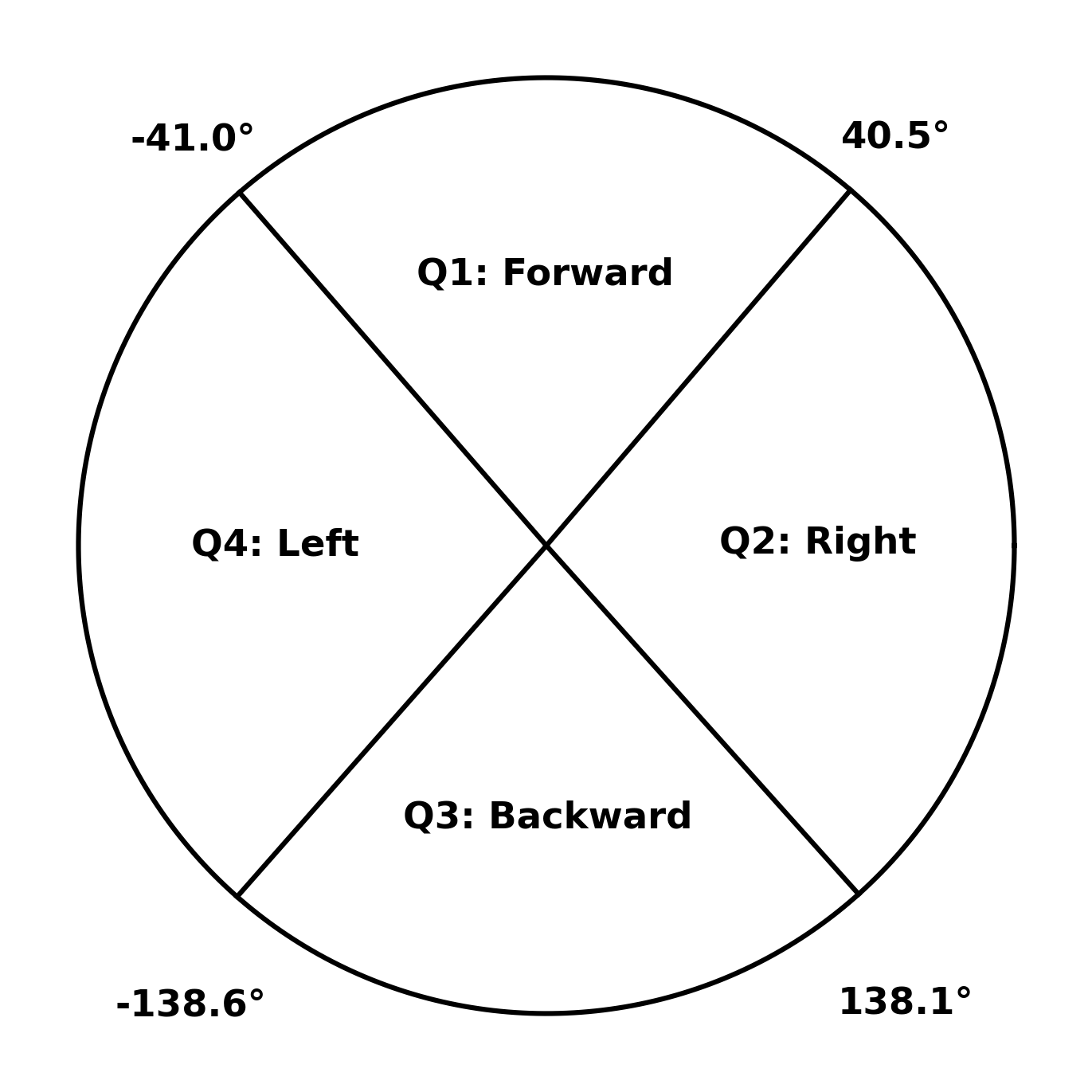}
    \caption{Movement-angle sectors of the quantile cube. The four boundaries are the empirical angle quartiles ($-41.0^\circ$, $40.5^\circ$, $138.1^\circ$, $-138.6^\circ$), measured clockwise from a forward baseline, and the sectors they define are the four directional bins (forward, right, backward, left).}
    \label{fig:angle_quantiles}
  \end{minipage}\hfill
  \begin{minipage}[t]{0.56\textwidth}
    \centering
    \includegraphics[height=2.3in]{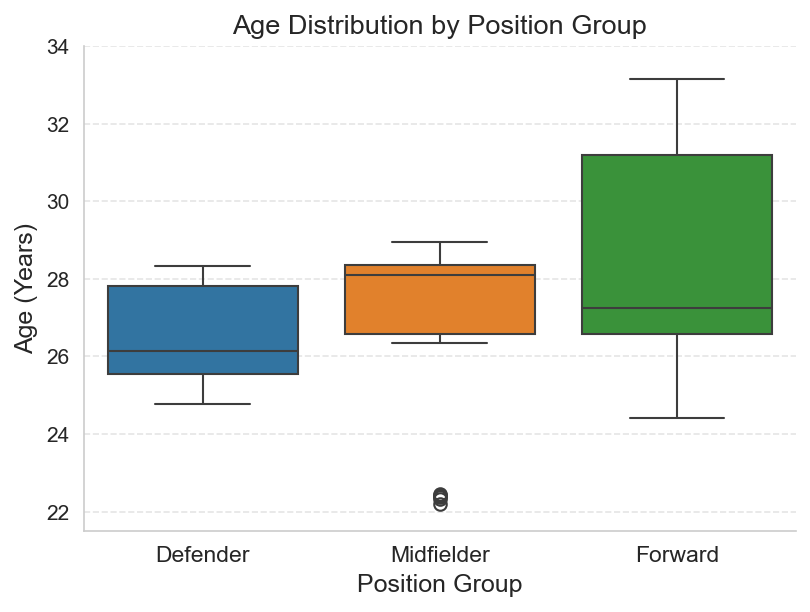}
    \caption{Age distribution by position group across the $216$ athlete--match sessions. Boxes show the interquartile range and median, whiskers extend to $1.5\times$ the interquartile range, and outliers are individual points. Forwards span a wider age range than defenders and midfielders, motivating the inclusion of both age and position in the design.}
    \label{fig:age_by_position}
  \end{minipage}
\end{figure}

The exposure $t_i$ is the time spent in prescribed activity in session $i$: playing time in the match plus the durations of that session's warm-up and cool-down. Warm-up and cool-down together contribute a median of $47.4$ minutes (interquartile range $38.9$ to $57.1$), giving a median exposure of $136$ minutes (range $94$ to $195$). Exposure varies across athletes and sessions, since some athletes begin their warm-up earlier or spend longer cooling down.

The design matrix $\mathbf X\in\mathbb R^{216\times4}$ has an intercept, standardized age (\texttt{Age\_z}, z-scored within the sample, mean $27.3$ years and standard deviation $2.1$ years), and indicators for forward and midfielder, with defender as the reference category. Ages ranged from $22.2$ to $33.1$ years, recomputed at each match date, and the position groups were six defenders, four midfielders and seven forwards. Goalkeepers are not in the dataset because their tracking data were not available. Figure~\ref{fig:age_by_position} shows the age distribution by position.

\subsection{Model Specification and Partition}
\label{sec:app:fit}

We fit the component-specific hybrid model of Section~\ref{sec:method} to the $216$ sessions: each of the $d=100$ bins has its own $\beta_j\in\mathbb R^4$ and dispersion $\nu_j$, for $500$ parameters in all. Observations are assigned by the lower-bound shape rule of Section~\ref{sec:sims:partition}. For bin $j$, the expected count is estimated by a pooled rate scaled by exposure, $\widehat\mu_{ij}=t_i\,n^{-1}\sum_kY_{kj}/t_k$, and the dispersion $\widehat\nu_j$ by the covariate-adjusted moment estimator. Because sessions are clustered within athletes, the uncertainty in the estimated shape is assessed by resampling athletes rather than sessions: with $\widehat s_j$ the standard deviation of the log estimated shape over $400$ athlete-level bootstrap resamples, observation $i$ is assigned to the normal component in bin $j$ when
$$\widehat\mu_{ij}\,\widehat\nu_j\exp(-1.645\,\widehat s_j)\geq\omega,
\qquad \omega=10.$$
The bootstrap standard errors have a median of $0.16$ (range $0.07$ to $0.25$), so the lower bound is typically about $0.77$ times the estimated shape. Section~\ref{sec:sims:dispersion} found coverage near nominal once the median assigned shape reached about $10$, so requiring it of every assigned observation is conservative.

Under this rule, $13{,}300$ of the $21{,}600$ observation--bin pairs ($61.6\%$) are assigned to the normal component (Figure~\ref{fig:shape_lcb}). Twenty-six bins are modeled by the \gls{nb} distribution for every session and are fit by the discrete construction of Section~\ref{sec:method:discrete}. Thirty-eight are assigned to the normal approximation for every session, and $36$ are mixed, so the $74$ active bins contribute $370$ parameters to the results of Sections~\ref{sec:theory} and~\ref{sec:asymptotics}. Figure~\ref{fig:partition_map} displays the partition across the movement grid. Because Figure~\ref{fig:shape_lcb} places substantial mass near the threshold, we refit the model at $\omega=7$ and $\omega=15$ and compare, for each covariate, the number of bins with a significant effect under the procedure of Section~\ref{sec:app:multiplicity}. The midfielder count is stable, the forward count falls at $\omega=15$, and the age count rises with the threshold, from $13$ at $\omega=7$ and $18$ at $\omega=10$ to $21$ at $\omega=15$ (\thresholdsensapp).

Every active bin has at least $11$ normal-assigned rows (the median mixed bin has $163$). We checked the rank conditions of Theorem~\ref{thm:jointdensity}(iii) numerically at fitted values: in all $74$ active bins the design matrix restricted to the normal-assigned rows has full column rank $K=4$ (smallest singular value of $X^{(j)}/\sqrt{n_j^C}$ at least $0.27$), and the evaluated structural Jacobian $J_j(\btheta)$ has full column rank $K+1=5$ at the fiducial mean and at $400$ fiducial draws per bin. These checks hold at the fitted parameters, not uniformly over $\Theta$ as the theorem assumes. Under Conditions~(A)--(E), the evaluated Gram matrix is positive definite uniformly on a neighborhood of $\btheta_0$ with probability tending to one, as shown in \bvmproofapp. This is what Theorem~\ref{thm:bvm} uses.

The fitted dispersions vary considerably across the grid, with fiducial medians from $0.0025$ to $0.070$ (overall median $0.028$, interquartile range $0.011$ to $0.041$), so a common mean threshold of $\tau=500$ would have imposed shape thresholds from $1.3$ to $35$. Under the shape rule, the fitted \gls{nb} shape among the $13{,}300$ normal-assigned pairs has a median of $28.7$, a fifth percentile of $14.2$ and a minimum of $11.2$, so no normal-assigned pair has skewness above about $0.60$.

Because the \gls{gfd} factorizes across bins (Proposition~\ref{prop:jacobian}), each bin is fit independently: the $74$ active bins by Hamiltonian Monte Carlo on the density of Theorem~\ref{thm:jointdensity} (four chains of $4{,}000$ draws after $1{,}000$ warmup, all $\widehat R<1.01$), and the $26$ entirely discrete bins by the sampler of Section~\ref{sec:method:discrete}, with diagnostics in \discretecompapp.

\begin{figure}[tbp]
  \centering
  \includegraphics[width=4.5in]{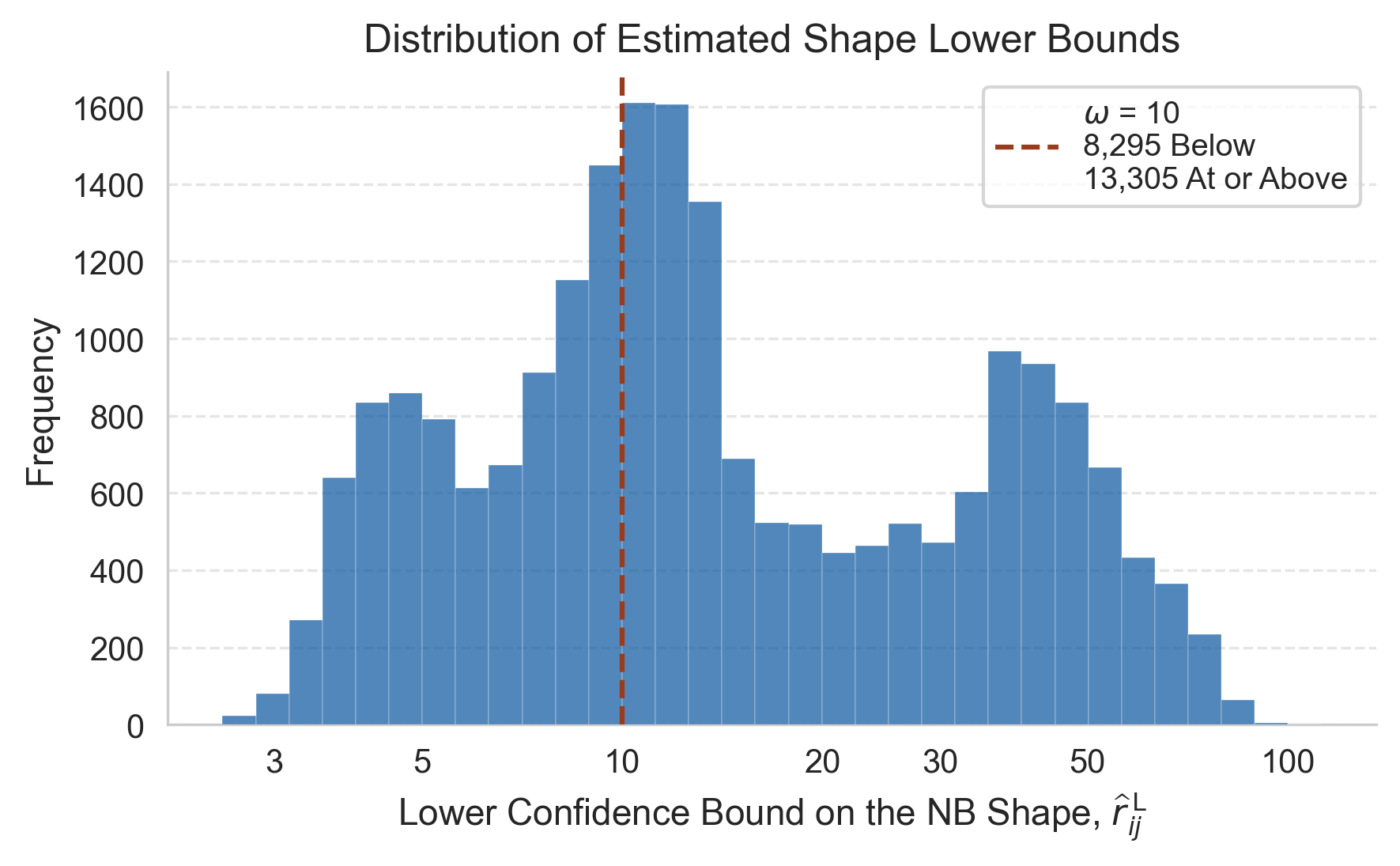}
    \caption{Distribution of the lower confidence bounds on the estimated \gls{nb} shape, $\widehat\mu_{ij}\widehat\nu_j\exp(-1.645\,\widehat s_j)$, across all $216\times100=21{,}600$ observation--bin pairs (log scale). The dashed line marks the partition threshold $\omega=10$: $8{,}295$ pairs lie below it and $13{,}305$ at or above it. After the rank rule of Section~\ref{sec:sims:partition} moves five of the latter to the \gls{nb} component, $8{,}300$ pairs are assigned to the exact \gls{nb} component and $13{,}300$ to the normal approximation.}
  \label{fig:shape_lcb}
\end{figure}

\begin{figure}[tbp]
  \centering
  \includegraphics[width=3.6in]{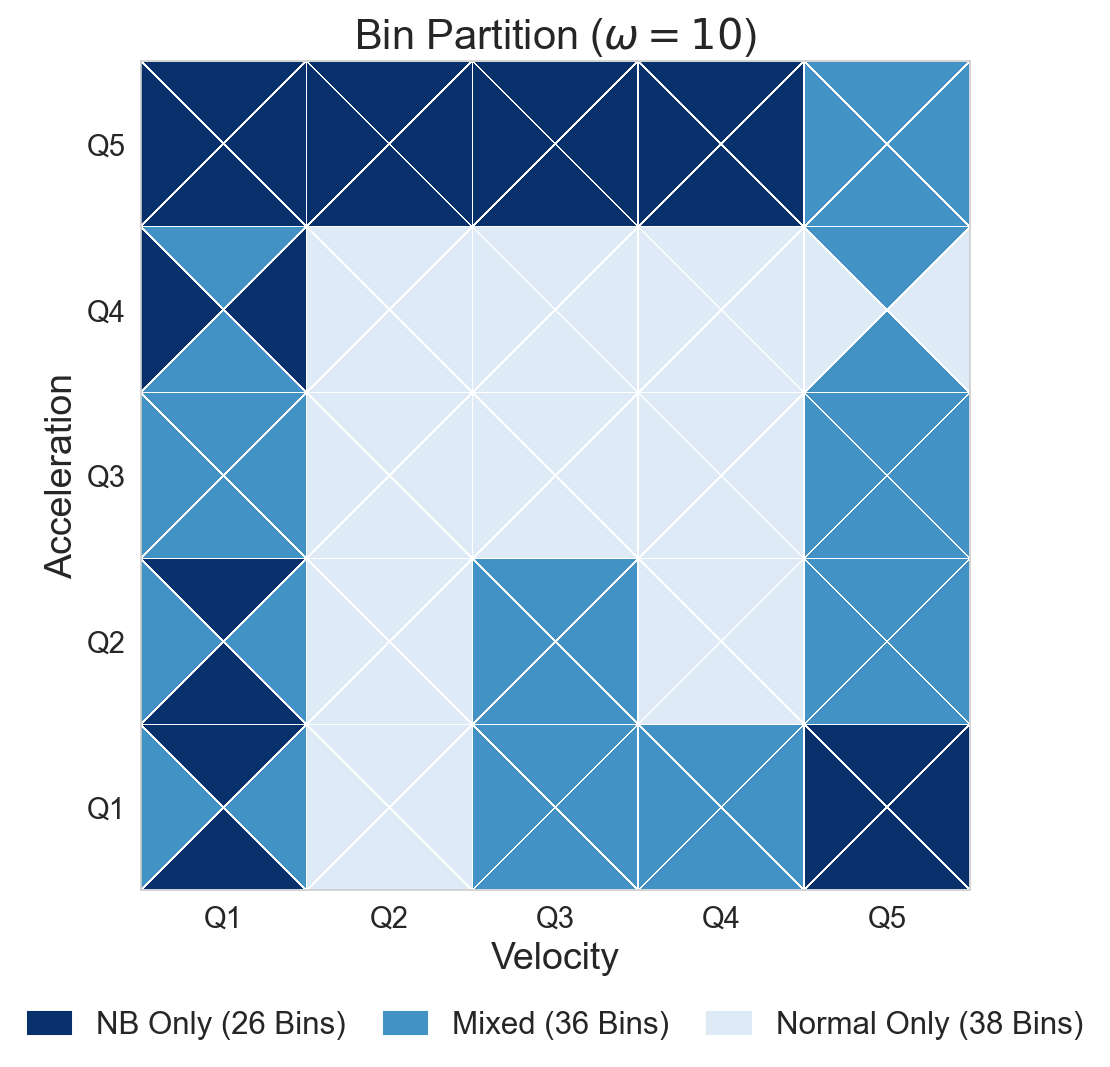}
  \caption{Partition of the $100$ movement bins under the lower-bound shape rule ($\omega=10$), on the velocity $\times$ acceleration grid with the four movement angles as triangles within each cell. Bins are entirely \gls{nb} ($26$), entirely normal ($38$) or mixed across observations ($36$). The entirely discrete bins are the highest-acceleration quintile at every velocity below the highest, the highest-velocity lowest-acceleration cell, and six angle sectors in the lowest-velocity column. The mixed bins lie mainly along the highest- and lowest-velocity columns and in the low-acceleration cells at intermediate velocity, surrounding a dense interior assigned entirely to the normal approximation.}
  \label{fig:partition_map}
\end{figure}

\subsection{Multiple Testing and a Calibration Check}
\label{sec:app:multiplicity}

Under the working model each bin is fit separately, so for a given covariate the $100$ per-bin effects are $100$ simultaneous inferences and call for a multiplicity adjustment. The bins are strongly dependent: $64\%$ of the residual correlations between bin pairs, after regressing the log rates on the covariates, are positive, and the effective number of independent bins by the method of \citet{li_ji_2005} is $16$ of $100$. The working model does not represent this structure (Section~\ref{sec:discuss:limitations}).

We pose the problem as one of sign determination, since a directional conclusion is what each colored bin reports. For each bin--covariate pair we compute the fiducial sign-error probability
$$\mathrm{FEP}_{jk}=\min\{P_r(\beta_{jk}>0\mid\mathbf y),\,P_r(\beta_{jk}<0\mid\mathbf y)\},$$
where $P_r(\cdot\mid\mathbf y)$ is probability under the \gls{gfd}: the fiducial probability of the sign opposite to the estimated effect \citep{hannig_et_al_2016}. Where Theorem~\ref{thm:bvm} applies, each tail probability is asymptotically a calibrated one-sided $p$-value, so $p_{jk}=2\,\mathrm{FEP}_{jk}$ is asymptotically a calibrated two-sided $p$-value. We apply the \gls{bh} procedure \citep{benjamini_hochberg_1995} to the $100$ two-sided $p$-values at level $0.05$ and attach the sign of the fiducial mean to each rejection, which controls the directional \gls{fdr} at that level given calibrated $p$-values \citep{benjamini_yekutieli_2005}. Theorem~\ref{thm:bvm} supplies that calibration asymptotically for the active bins under the working model, and \citet[Theorem~4.1]{hannig_2013} for the entirely discrete bins under its assumptions. Neither covers the estimated partition or athlete-level dependence, which the check below addresses. Throughout, ``significant'' means significant under this model. The fiducial probabilities are computed directly from the fiducial draws, with no normal approximation, and Monte Carlo error in the reported counts is small: a bootstrap over the draws gives a standard deviation of at most $0.9$ bins, and each chain analyzed separately reproduces the counts to within four bins.

\paragraph{Calibration check} The nominal level presumes $216$ independent sessions, but age and position vary almost entirely between the $17$ athletes. We therefore reassigned the athletes' position labels at random, every session of an athlete carrying its permuted label, refit all $100$ bins, and repeated this $100$ times. The same was done for age, reassigning athletes' mean ages while keeping each session's within-athlete deviation (\permutationapp). With no position effect present, the procedure declares a median of $31$ forward and $27.5$ midfielder bins significant, and with no age effect a median of $28.5$ age bins, against $23$, $48$ and $18$ observed. Significance under the working model is therefore not calibrated at the athlete level for these covariates: the relevant number of independent units is $17$, not $216$. The declared bins of Figure~\ref{fig:app_cubes} are where the fitted model places its largest and best-resolved contrasts among these $17$ athletes, and Section~\ref{sec:app:effects} describes the estimated surfaces in those terms.

\subsection{Effects of Age and Position}
\label{sec:app:effects}

For each covariate the fitted model yields a bin-specific effect $\beta_{jk}$ with a full \gls{gfd}, summarized by its fiducial mean and a $95\%$ equal-tailed fiducial interval. Because $\beta_{jk}$ enters through the log-linear link $\alpha_{ij}=\exp(X_i^\top\beta_j)$, a coefficient $\beta_{jk}$ corresponds to a multiplicative change of $e^{\beta_{jk}}$ in the expected count for bin $j$ per unit change in covariate $k$. For example, $\beta_{jk}=0.10$ corresponds to an $e^{0.10}\approx1.105$, or $10.5\%$, increase in expected count. Because the $100$ bins form a $5\times5\times4$ grid over velocity, acceleration and direction, we display each covariate's estimated effects as a quantile cube (Figure~\ref{fig:app_cubes}), with velocity and acceleration quintiles on the axes, the four movement angles as triangles within each cell, and the bins significant under the working model after \gls{bh} correction shown in color.

Under the working model, at the \gls{bh}-corrected level, the standardized age effect is significant in $18$ of the $100$ bins, the forward contrast (relative to defenders) in $23$, and the midfielder contrast in $48$. The clearest feature of the estimated surfaces is the difference in magnitude between age and position. The estimated age effects are modest, with $|\beta_{jk}|<0.09$ throughout the grid, the largest corresponding to an $8.5\%$ change in expected count, whereas the estimated position contrasts reach $|\beta_{jk}|\approx0.29$ (midfielder, highest-velocity bins), more than three times as large. In the fitted model, then, age is associated with a comparatively fine-grained redistribution of movement intensity, and position with substantially larger differences in the movement distribution.

The estimated age surface follows an acceleration gradient. All six declared bins with positive estimated age effects lie in the highest acceleration quintile, across the upper velocity range ($\beta_{jk}$ up to $0.08$), while all $12$ with negative estimates lie at lower acceleration, most at intermediate velocity ($\beta_{jk}$ as low as $-0.07$). Taken at face value, the fitted model places more of an older athlete's session time in the most demanding acceleration bins and less at lower acceleration, at comparable velocities. The estimated position contrasts are concentrated at the velocity extremes and show a more complex acceleration-dependent structure. Relative to defenders, the estimated midfielder contrast is positive in both the lowest- and the highest-velocity bins, largest in the highest-velocity bins ($\beta_{jk}\approx0.14$ to $0.29$) and smaller in the lowest-velocity bins ($\beta_{jk}$ up to $0.21$), with a band of negative estimates confined to the two highest acceleration quintiles (Q4 and Q5) at low-to-mid velocity ($\beta_{jk}$ as low as $-0.24$). The estimated forward contrast shows the same positive low- and high-velocity pattern ($\beta_{jk}$ up to $0.14$), together with a weaker band of negative estimates across the mid-velocity range at intermediate acceleration ($\beta_{jk}$ as low as $-0.09$). Section~\ref{sec:discuss:soccer} discusses the soccer-tactical interpretation of these patterns and the comparison with existing studies.

\begin{figure}[tbp]
  \centering
  \begin{subfigure}[b]{5in}
    \centering
    \includegraphics[width=2.5in]{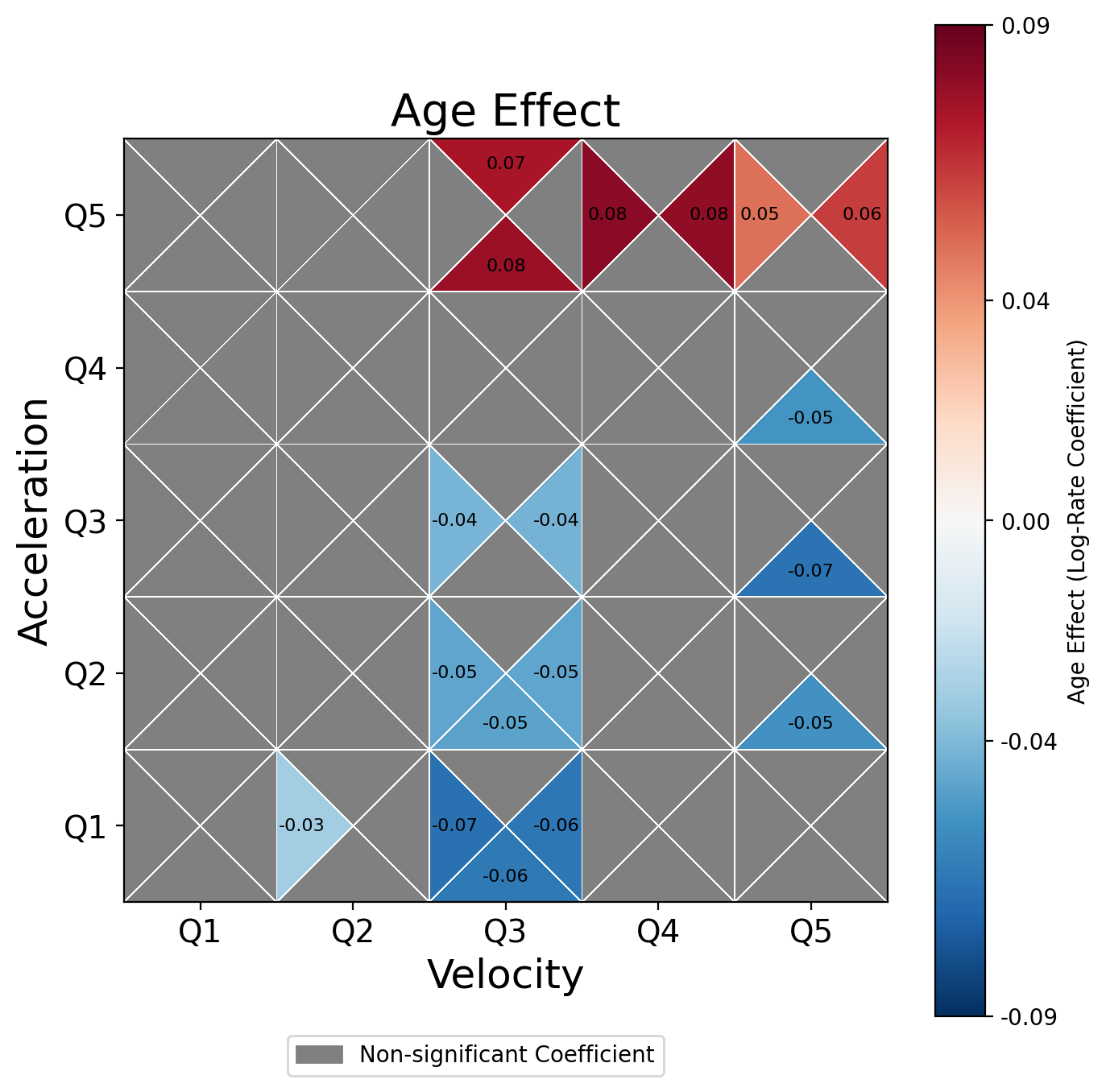}
    \caption{Age}
    \label{fig:cube_age}
  \end{subfigure}
  \\[1ex]
  \begin{subfigure}[b]{2.5in}
    \centering
    \includegraphics[width=2.5in]{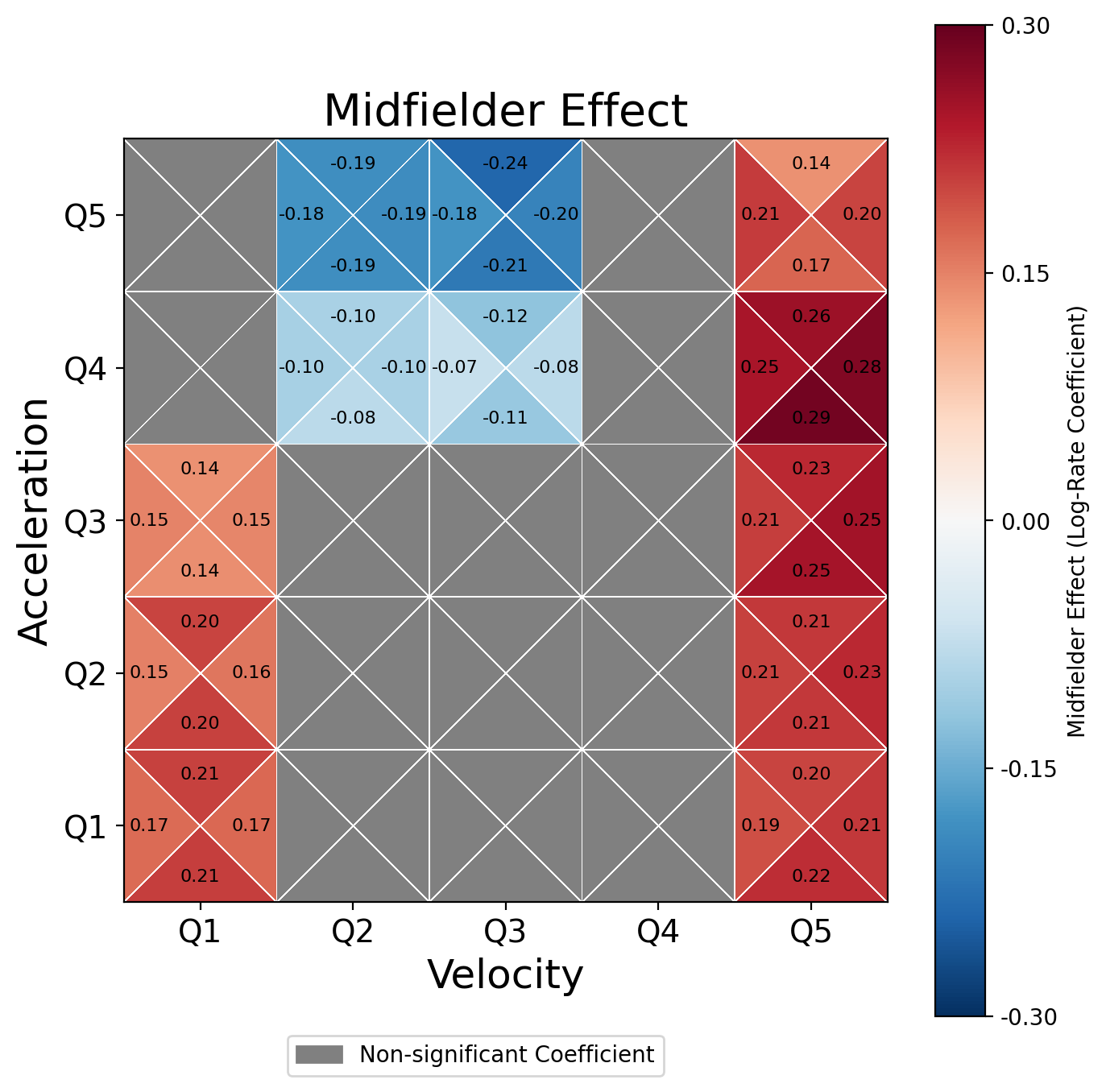}
    \caption{Midfielder vs.\ Defender}
    \label{fig:cube_mid}
  \end{subfigure}
  \hfill
  \begin{subfigure}[b]{2.5in}
    \centering
    \includegraphics[width=2.5in]{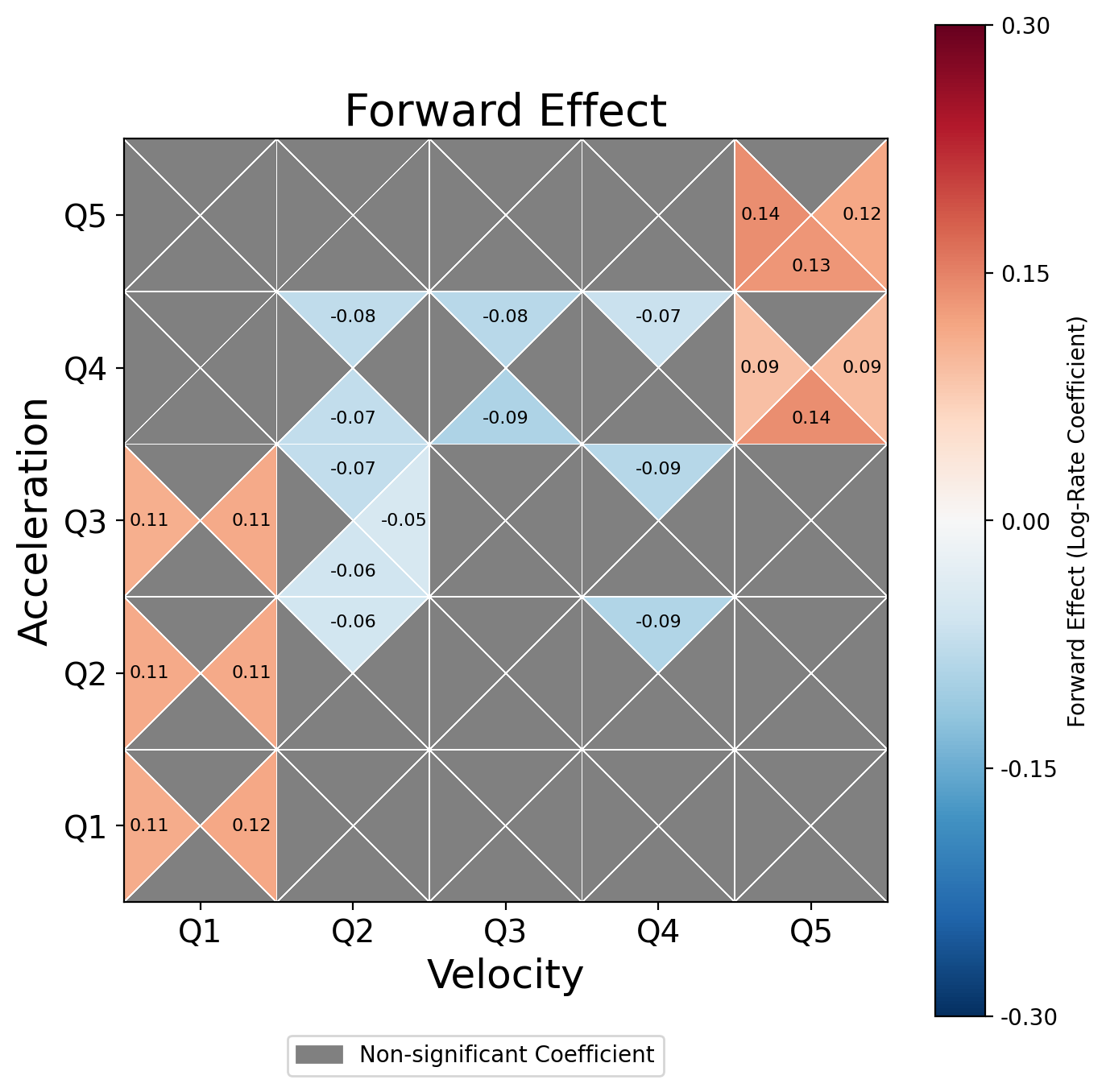}
    \caption{Forward vs.\ Defender}
    \label{fig:cube_fwd}
  \end{subfigure}
  \caption{Fiducial mean covariate effects by movement bin from the application model, displayed as quantile cubes (velocity $\times$ acceleration quintiles, with the four movement angles as triangles within each cell). Bins whose effect is significant under the working model at directional \gls{fdr} $0.05$ are shown in color, and bins with no sign call are greyed. Section~\ref{sec:app:multiplicity} shows that this criterion is not calibrated at the athlete level, so color marks where the fitted model resolves its largest contrasts among the $17$ athletes, not where an effect is established.}
  \label{fig:app_cubes}
\end{figure}

\subsection{A Smooth Basis-Function Extension}
\label{sec:app:basis}

The per-bin model of Section~\ref{sec:app:fit} estimates each covariate effect $\beta_{jk}$ independently across the $100$ movement bins, ignoring the smooth geometric structure of the velocity--acceleration--angle grid. We therefore consider an extension that borrows strength across neighboring bins by expressing the covariate effects as a smooth function of a bin's grid coordinates.

With $r_{ij}$ the \gls{nb} shape of Section~\ref{sec:method}, so that $\mu_{ij}=r_{ij}/\nu_j$, we place the smooth structure on the linear predictor,
\begin{equation}
\label{eq:basis}
\log\!\left(\frac{r_{ij}}{t_i}\right)=\alpha_j+\tilde X_i^\top\gamma\,\phi_j,
\end{equation}
where $\alpha_j$ is a free bin-specific intercept, $\tilde X_i\in\mathbb R^{K-1}$ collects the non-intercept covariates ($K-1=3$ here, since $\alpha_j$ absorbs the intercept), $\phi_j\in\mathbb R^M$ is a fixed vector of basis functions evaluated at the grid coordinates of bin $j$, and $\gamma\in\mathbb R^{(K-1)\times M}$ is a shared coefficient matrix. Equivalently, the effect of covariate $k$ in bin $j$ is
$$\beta_{jk}=\sum_{m=1}^M\gamma_{km}\phi_{jm},$$
so the covariate effects vary smoothly over the movement grid rather than being $100$ unrestricted bin-specific coefficients. Because $\nu_j$ does not depend on the covariates, the induced log mean rate is $\log(\mu_{ij}/t_i)=\alpha_j-\log\nu_j+\tilde X_i^\top\gamma\,\phi_j$, so $\beta_{jk}$ keeps the multiplicative interpretation of Section~\ref{sec:app:effects}.

\paragraph{Basis functions} Let $v_j,a_j\in\{1,\ldots,5\}$ denote the velocity and acceleration quintiles of bin $j$, centered as $v_j^c=v_j-3$ and $a_j^c=a_j-3$, and let $d^{LR}_j,d^{FB}_j\in\{-1,0,1\}$ encode the four movement directions (left/right and forward/backward). We consider a linear basis $\phi_j=(v_j^c,\,a_j^c,\,d^{LR}_j,\,d^{FB}_j,\,v_j^ca_j^c)$ with $M=5$, and a quadratic extension adding $(v_j^c)^2$ and $(a_j^c)^2$ with $M=7$. The basis columns are scaled to comparable magnitude to improve conditioning. The dispersion is modeled either as bin-specific ($\nu_j$, as before) or as a smooth function of the same basis, $\log\nu_j=\delta_0+\phi_j^\top\delta$, giving the four variants of Table~\ref{tab:basis_variants}.

\begin{table}[tbp]
\centering
\begin{tabular}{ccccc}
\toprule
Variant & Basis & Dispersion & Parameter count & Total \\
\midrule
1 & $M=5$ & bin-specific $\nu_j$ & $2d+(K-1)M=200+15$        & $215$ \\
2 & $M=7$ & bin-specific $\nu_j$ & $2d+(K-1)M=200+21$        & $221$ \\
3 & $M=5$ & smooth $\nu$         & $d+(K-1)M+M+1=100+15+6$   & $121$ \\
4 & $M=7$ & smooth $\nu$         & $d+(K-1)M+M+1=100+21+8$   & $129$ \\
\bottomrule
\end{tabular}
\caption{Basis-function model variants, with parameter counts for $d=100$ movement bins and $K-1=3$ covariates (standardized age, forward and midfielder indicators). The smooth-dispersion variants replace the $d=100$ bin-specific dispersion parameters with $M+1$ basis coefficients.}
\label{tab:basis_variants}
\end{table}

\paragraph{Fiducial density} Because the coefficient matrix $\gamma$ is shared across bins, the structural-equation Jacobian is no longer block diagonal across bins as in Proposition~\ref{prop:jacobian}. It couples all bins through $\gamma$, and the Jacobian factor is the determinant of the Gram matrix $J_C^\top J_C$ formed from the parameter directions that appear in the continuous structural equations. Those directions are $\gamma$, which enters the continuous row of every normal-assigned observation, together with the intercepts and dispersions of the bins that have at least one normal-assigned observation. The intercept and dispersion of each of the $26$ entirely discrete bins appear in no continuous row and, following the active-component convention of Section~\ref{sec:theory}, are excluded from the Jacobian. They enter the fiducial density through the \gls{nb} likelihood alone, without the correction $\mathbb E_{\btheta}[v_{\mathbf U}(\btheta)]$ of Lemma~\ref{lem:discrete}. This is the flat-prior treatment of those bins, and the comparison in \discretecompapp{} shows that the omitted correction moves their fiducial means by about $0.02$ posterior standard deviations on average, and by at most $0.07$ in any bin. The total parameter dimension ranges from $121$ to $221$ across the variants of Table~\ref{tab:basis_variants}. For the preferred variant it is $221$, of which $169$ enter the Jacobian. The model in \eqref{eq:basis} is thus a lower-dimensional constrained extension of the component-specific hybrid model, and the mixed discrete--continuous argument of Theorem~\ref{thm:jointdensity} motivates the same likelihood-times-Jacobian construction for its active parameters, with a dense rather than block-diagonal Gram matrix. The \gls{bvm} verification of Section~\ref{sec:asymptotics}, however, covers only the component-specific parameterization (Section~\ref{sec:discuss:limitations}).

\paragraph{Fitting and model selection} The basis models use the same observation-level partition as the per-bin model of Section~\ref{sec:app:fit}: an observation of a bin contributes a normal density and a Jacobian row when it is assigned to the normal component, and an \gls{nb} mass otherwise. All four variants were fit by Hamiltonian Monte Carlo with identical settings and converged ($\widehat R\leq1.013$, no divergent transitions). The quadratic basis ($M=7$) is necessary. The linear basis ($M=5$) systematically over-smooths the position contrasts, shrinking effects below the per-bin $95\%$ interval in $17$ to $19$ (forward) and $34$ to $35$ (midfielder) bins, whereas the quadratic basis reduces this to at most one. The age surface is smoother to begin with, with at most three bins shrunk below the interval under the linear basis and at most two under the quadratic basis. The dispersion model has little effect on the mean effects: the bin-specific and smooth-dispersion surfaces agree with correlations of $0.95$ to $1.00$ across covariates under the quadratic basis, and $0.87$ to $1.00$ under the linear basis. The smooth dispersion is, however, a poor summary of the bin-specific estimates themselves, correlating $0.70$ with the per-bin dispersions under the quadratic basis against $1.00$ for the bin-specific parameterization. We therefore adopt variant~2 (quadratic basis, bin-specific dispersion).

\paragraph{Comparison with the per-bin model} Relative to the independent per-bin fits of Section~\ref{sec:app:effects}, the basis model yields substantially narrower intervals. The median $95\%$ interval width is roughly $0.27$ times the corresponding per-bin width for every covariate, reflecting the smoothness constraint imposed across the grid. The two estimated surfaces are highly correlated for the midfielder contrast ($0.92$) and age ($0.84$), and less strongly for the forward contrast ($0.62$). The proportion of basis fiducial means falling within the per-bin $95\%$ intervals is $93\%$ for age, $83\%$ for midfielders and $75\%$ for forwards. Where the two models disagree, the pattern is consistent with smoothing: the basis surface rarely falls below the per-bin lower bound (two bins for age, none for midfielders and one for forwards) but more often exceeds the upper bound ($5$, $17$ and $24$ bins, respectively).

The colored-bin counts are not directly comparable across the two analyses. Figure~\ref{fig:app_cubes} uses directional \gls{bh} adjustment across the component-specific per-bin effects, whereas Figure~\ref{fig:basis_effects} uses simultaneous fiducial bands across the basis-function surface, reflecting joint uncertainty in that surface. The comparisons above concern the fitted surfaces and intervals themselves and do not depend on either significance rule.

\begin{figure}[tbp]
  \centering
  \begin{subfigure}[b]{5in}
    \centering
    \includegraphics[width=2.5in]{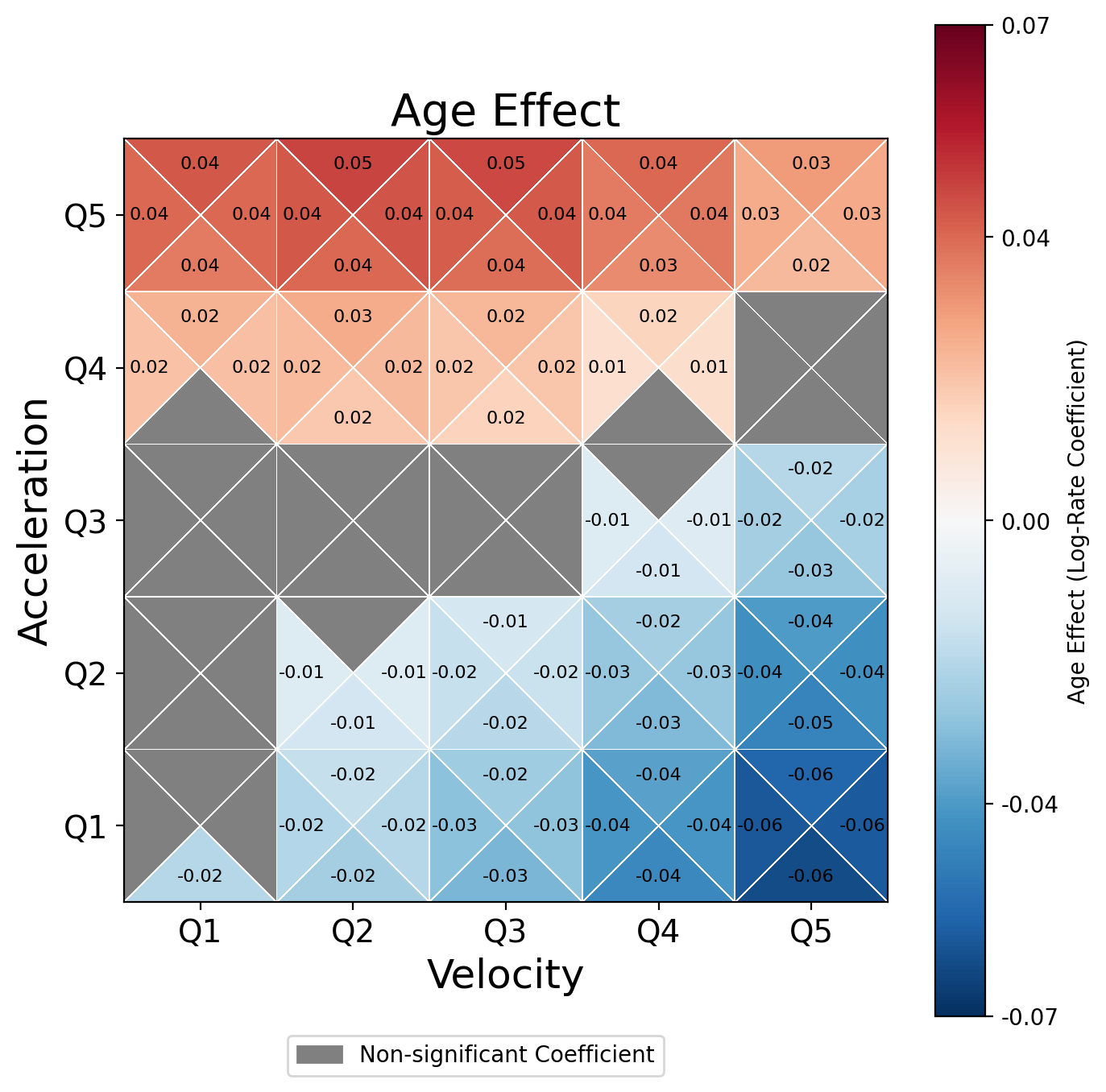}
    \caption{Age}
    \label{fig:basis_age}
  \end{subfigure}
  \\[1ex]
  \begin{subfigure}[b]{2.5in}
    \centering
    \includegraphics[width=2.5in]{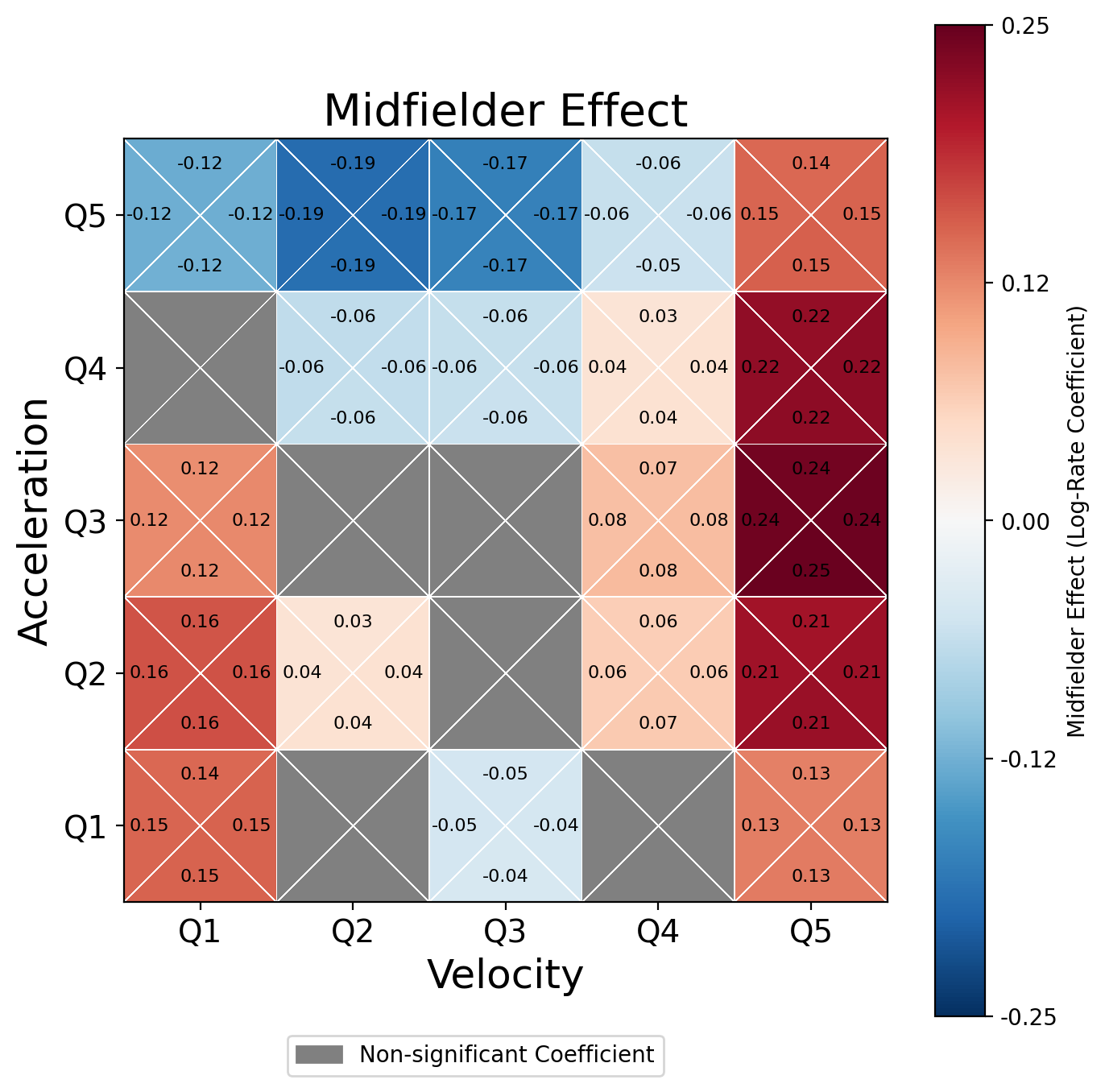}
    \caption{Midfielder vs.\ Defender}
    \label{fig:basis_mid}
  \end{subfigure}
    \hfill
  \begin{subfigure}[b]{2.5in}
    \centering
    \includegraphics[width=2.5in]{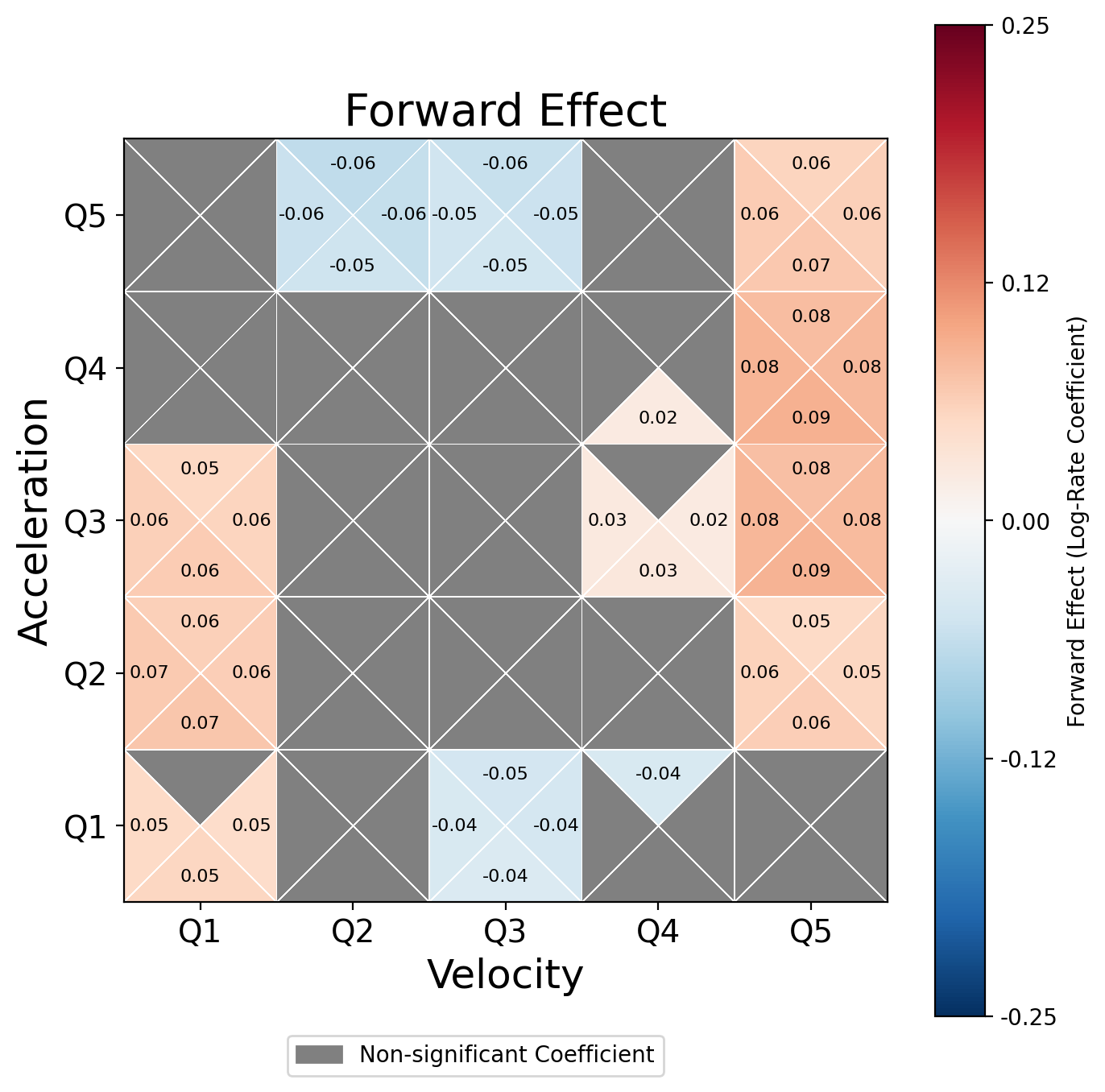}
    \caption{Forward vs.\ Defender}
    \label{fig:basis_fwd}
  \end{subfigure}
  \caption{Each panel shows the fiducial mean effect $\beta_{jk}=\sum_m \gamma_{km}\phi_{jm}$ for one covariate on the velocity $\times$ acceleration grid, with the four movement-angle sectors shown as triangles within each cell. Colors indicate the sign of effects whose $95\%$ simultaneous fiducial bands under the working-independence model exclude zero. Uncolored bins have bands that include zero. These bands inherit the calibration limitation of Section~\ref{sec:app:multiplicity}. The bands are constructed from the joint fiducial draws for each covariate surface using the maximum absolute studentized deviation across all $100$ bins. Their critical values range from $3.16$ to $3.20$, compared with the pointwise value of $1.96$. Position contrasts are relative to defenders, and the two position panels share a common color scale.}
  \label{fig:basis_effects}
\end{figure}

\section{Discussion}
\label{sec:discuss}

We have developed a \gls{hgfi} framework for multivariate count data in which each observation is partitioned into components with small \gls{nb} shape, modeled exactly through a \gls{nb} structural equation, and components with large shape, modeled through a normal approximation with the same mean and variance. The generalized fiducial density factors into the working likelihood and a Jacobian generated entirely by the continuous structural equations, block diagonal across components, and the discrete construction of Section~\ref{sec:method:discrete} handles components that are \gls{nb} for every observation. The simulation study found empirical coverage close to nominal across the discrete, mixed and continuous regimes, and it showed that the partition must not be formed from the realized response and should be based on the \gls{nb} shape rather than the conditional mean: when the dispersion varies, a partition by an estimated mean under-covers in both the sparse and mixed regimes, whereas a partition by an estimated shape matches the oracle. The same structure, discrete components entering through the likelihood alone and continuous components also generating the geometry, applies to other settings that allocate constrained counts across components, including microbiome, genomic and ecological studies \citep{linden_mantyniemi_2011, anders_huber_2010, chen_li_2013, combettes_muller_2021, devalpine_harmon_2013}.

Fiducial and flat-prior Bayesian intervals are essentially indistinguishable in length and calibration once $n$ is moderate, as Theorem~\ref{thm:bvm} predicts. For the entirely discrete bins the difference is about $0.02$ posterior standard deviations (\discretecompapp). The construction’s contribution is a default rather than a different answer: the Jacobian is derived from the structural equations rather than chosen, and it agrees with the flat prior to first order without an argument for that prior.

\subsection{Interpretation of the Application}
\label{sec:discuss:soccer}

Section~\ref{sec:app:multiplicity} showed that the estimated surfaces cannot be distinguished from label noise at the athlete level, so what follows is interpretation of the fitted patterns, not a claim that they are established. Two stand out.

The age effect connects our analysis to the literature on athletic aging curves, which relates age to average performance outcomes and traces a rise to a peak followed by decline \citep{fair_2008, bradbury_2009, page_et_al_2013}, with the peak in soccer at roughly $25$ to $27$ years \citep{dendir_2016}. Our model asks instead whether age reshapes how an athlete moves, holding position and exposure fixed, not only how well they perform. The acceleration gradient of Section~\ref{sec:app:effects} speaks to that question, but our design is cross-sectional across athletes rather than longitudinal within athletes, and age is confounded with role, selection and survivorship in a single-season roster. The gradient is therefore consistent with, but cannot establish, a within-athlete aging effect, and may instead reflect experience-driven role selection. Separating these explanations would require multi-season, within-athlete data, which also motivates the athlete-level random-effects extension of Section~\ref{sec:discuss:future}.

The midfielder pattern of Section~\ref{sec:app:effects}, with more time at both velocity extremes and less in the highest-acceleration bins at low-to-mid velocity, may reflect a role that combines sustained low-speed possession and positional play with high-speed transitions while involving fewer maximal-acceleration bursts than defenders \citep{griffin_2021, panduro_2022}. The forward pattern is consistent with an intensity profile that alternates sprinting with lower-speed holding phases. These are interpretations of contrasts against a defender baseline, not established tactical facts, and they concern the distribution of movement rather than the aggregate distance and high-speed-running totals of prior work \citep{datson_2017, panduro_2022}. Distributional changes of this kind may inform athlete monitoring and individualized training \citep{bourdon_2017, snyder_2024, thomas_hannig_2026} in women's soccer, which remains underrepresented in this literature \citep{ferraz_2023}.

The two model formulations answer complementary questions. By borrowing strength across neighboring bins, the smooth basis-function model yields substantially narrower intervals under its smoothness constraint and provides a useful summary of broad positional structure. The independent per-bin model remains preferable when localized effects are of interest, including the band of negative effects at mid-range velocities that the basis-function model attenuates. Reporting both helps avoid treating either the smoothed surface or the noisier per-bin map as the single definitive representation.

\subsection{Limitations}
\label{sec:discuss:limitations}

The primary limitation is the independence assumption of the working model. The \gls{gfd} and the \gls{bvm} theorem are established for the conditionally independent hybrid model, in which the component responses are independent given the design and partition. The soccer counts, however, satisfy a fixed-sum constraint: each athlete--match session allocates every decisecond of the observation window across the $100$ bins, so conditional on the window length each row of the conditional covariance matrix sums to zero. The working model does not represent this cross-component dependence, our theory does not establish robustness to it, and the application must accordingly be read as a working-model analysis rather than as inference under a specified compositional probability model. The most direct way to assess the cost of the approximation is to generate data with the relevant cross-component dependence and compare the empirical coverage of the working-model fiducial intervals with their nominal level (Section~\ref{sec:discuss:future}).

A second departure concerns dependence across observations. The $216$ sessions come from only $17$ athletes and $26$ matches: an athlete's movement profile is correlated across their own sessions, and match conditions are shared within a match. This matters most for the covariates that vary mainly between athletes rather than within them, age and position, for which the reported intervals should be read as descriptive rather than as calibrated frequentist statements. The permutation check of Section~\ref{sec:app:multiplicity} quantifies this. A procedure calibrated at the athlete level would need either athlete-level random effects (Section~\ref{sec:discuss:future}) or a reference distribution built from athlete-level resampling. Elite-athlete tracking studies often have this structure, with many high-frequency sensor observations but few independent athletes, teams or sessions, so finite-sample calibration remains a practical concern even where the asymptotic theory provides a formal justification.

Three limitations concern the scope of the theory. The first is the misspecification of the normal branch under the \gls{nb} data-generating process: the discrepancy between the working-model information and the repeated-sampling variance is of order $1/r_{ij}$ and does not vanish in $n$, so the routing threshold, not the sample size, controls calibration. Theorem~\ref{thm:bvm} is therefore a statement about the working model, and Section~\ref{sec:sims:dispersion} shows what happens when the threshold is too low. The second is the plug-in partition of Sections~\ref{sec:method} and~\ref{sec:app:fit}: the asymptotic results are stated for the oracle partition, and a formal treatment of the estimated one remains open. One route would replace the hard threshold by a smooth transition between the two representations, making the working likelihood a continuous function of the estimated partition and avoiding the discontinuous assignments that complicate the asymptotic argument. Within the \gls{gfi} framework this requires a structural-equation formulation of the transition, since the Jacobian is generated by the continuous structural equations rather than by the likelihood alone.

The third is the basis-function extension, which couples the bins through a shared coefficient matrix and so falls outside the block-diagonal structure of Proposition~\ref{prop:jacobian}. As noted in Section~\ref{sec:app:basis}, the mixed discrete--continuous argument of Theorem~\ref{thm:jointdensity} motivates the same likelihood-times-Jacobian construction there, but the asymptotic conditions of Section~\ref{sec:asymptotics} were verified only for the component-specific parameterization. The same extension would let the intercept and dispersion of an entirely discrete bin in the coupled model be handled by the discrete construction rather than the likelihood alone. The representation of Lemma~\ref{lem:discrete} extends to parameters pinned only by discrete constraints once the remaining parameters are fixed, and at fixed dispersion the consistent set for a bin's intercept is an interval, the one-dimensional case of the polytope in the per-bin construction, so the computation is cheap. The corresponding density, however, has not been derived.

A last limitation is computational. The discrete construction costs about three orders of magnitude more than the corresponding Stan fit. The simulation study and the athlete-level permutation check use it for every entirely discrete fit, the latter with shorter chains (\permutationapp).

\subsection{Extensions and Future Directions}
\label{sec:discuss:future}

A central next step is to relax the working model's independence assumption. Athlete-tracking data exhibit several forms of dependence: the movement bins are compositionally dependent because every decisecond of an observation window is allocated to one bin, observations from the same athlete are correlated across sessions, and adjacent time segments may be temporally dependent. Work on \gls{gfi} for linear mixed models provides a precedent for hierarchical structure \citep{cisewski_hannig_2012}, but extending these ideas to discrete, overdispersed count data is an open problem. Natural next steps are hybrid \gls{nb} models with athlete- and session-level random effects \citep{martin_2019, ma_dunton_hedeker_2025}, explicit temporal dependence across adjacent time segments, a compositional construction that imposes the fixed-sum constraint directly, and a high-dimensional regime in which the number of components grows with the sample size.

The hybrid construction itself suggests two further extensions. First, scalable and distributed computation for \gls{gfi} \citep{lai_hannig_lee_2021, du_hannig_etal_2025} could support fits with more components and observations than the componentwise scheme used here accommodates. Second, because Theorem~\ref{thm:jointdensity} depends on the discrete components only through the working likelihood, the construction may be adapted to other discrete families, including zero-inflated or hurdle components \citep{mullahy_1986, lambert_1992, risso_et_al_2018, feng_2021}, provided the continuous structural equations retain full column rank. This suggests applications to count data with excess zeros or latent subpopulations.

\begin{acks}[Acknowledgments]
The authors would like to extend their sincere thanks to the professional women's soccer team, coaching staff, and support personnel who facilitate data collection and generously share their time and expertise. 
\end{acks}

\begin{funding}
The work was supported in part by the National Science Foundation under Grant Nos.~DMS-2210337 and DMS-2515303, and by the United States--Israel Binational Science Foundation (BSF), Jerusalem, under Grant No. 2024055.
\end{funding}

\begin{supplement}
\stitle{Supplementary Material for ``Generalized Fiducial Inference for Hybrid Discrete--Continuous Count Data: Theory and Application''}
\sfilename{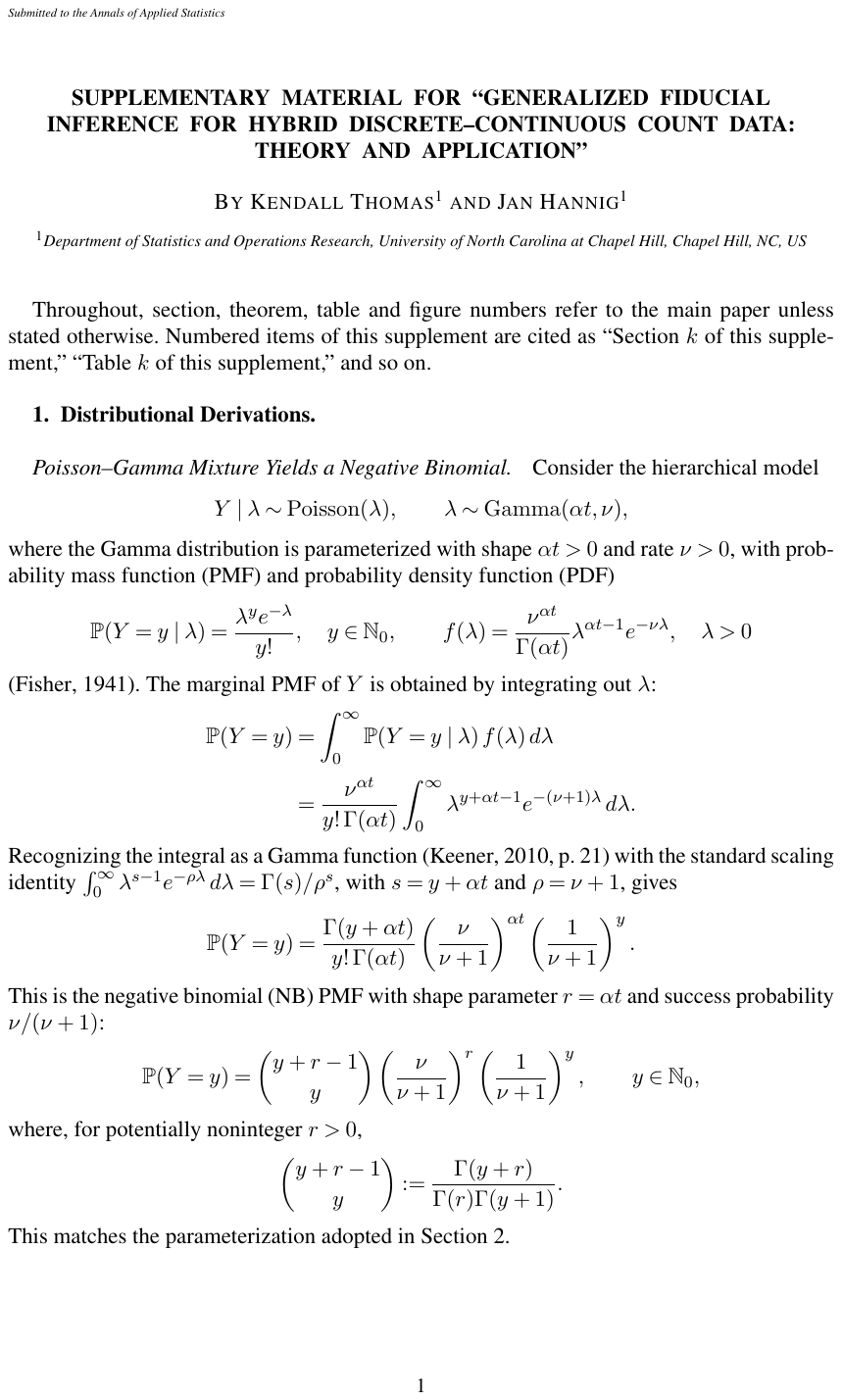}
\sdescription{Contains the distributional derivations for the \gls{nb} working model; the score, Hessian, and Fisher information for the hybrid model; the proofs of Theorem~\ref{thm:jointdensity}, Proposition~\ref{prop:jacobian}, Lemma~\ref{lem:posdef}, and Theorem~\ref{thm:bvm}, including verification of the five regularity conditions of \citet{borgert_hannig_2026}; the full coverage results of the simulation study; the proof of Lemma~\ref{lem:discrete} together with the sampler for the discrete generalized fiducial distribution, its validation against exact draws, its comparison with the flat-prior posterior, and the diagnostics for the application's entirely discrete bins; and the athlete-level permutation check and threshold sensitivity for the application.}
\end{supplement}

\begin{supplement}
\stitle{Code for the simulation study}
\sfilename{thomas\_hannig\_code\_supplement.zip}
\sdescription{Stan models and Python scripts for the simulation study of Section~\ref{sec:sims}, with the result files from which its tables and figures are rebuilt, together with the sampler for the discrete generalized fiducial distribution of Section~\ref{sec:method:discrete}, the single-site Gibbs sampler, and the exact rejection reference used to validate them (\discretecompapp). The GPS tracking data analyzed in Section~\ref{sec:app} were collected under \gls{irb} 25-1234 subject to a data-sharing agreement that does not permit redistribution, so neither those data nor the application-specific code are included.}
\end{supplement}

\bibliographystyle{imsart-nameyear} 
\bibliography{references}       



\section*{Supplementary material}

# SUPPLEMENTARY MATERIAL FOR "GENERALIZED FIDUCIAL INFERENCE FOR HYBRID DISCRETE–CONTINUOUS COUNT DATA: THEORY AND APPLICATION"


BY KENDALL THOMAS[1] AND JAN HANNIG[1]

[1]*Department of Statistics and Operations Research, University of North Carolina at Chapel Hill, Chapel Hill, NC, US*


Throughout, section, theorem, table and figure numbers refer to the main paper unless stated otherwise. Numbered items of this supplement are cited as "Section $k$ of this supplement," "Table $k$ of this supplement," and so on.

## 1. Distributional Derivations.

*Poisson–Gamma Mixture Yields a Negative Binomial.* Consider the hierarchical model

$$Y \mid \lambda \sim \text{Poisson}(\lambda), \qquad \lambda \sim \text{Gamma}(\alpha t, \nu),$$

where the Gamma distribution is parameterized with shape $\alpha t > 0$ and rate $\nu > 0$, with probability mass function (PMF) and probability density function (PDF)

$$\mathbb{P}(Y = y \mid \lambda) = \frac{\lambda^y e^{-\lambda}}{y!}, \quad y \in \mathbb{N}_0, \qquad f(\lambda) = \frac{\nu^{\alpha t}}{\Gamma(\alpha t)} \lambda^{\alpha t - 1} e^{-\nu \lambda}, \quad \lambda > 0$$

(Fisher, 1941). The marginal PMF of $Y$ is obtained by integrating out $\lambda$:

$$\begin{aligned}\mathbb{P}(Y = y) &= \int_0^\infty \mathbb{P}(Y = y \mid \lambda)\, f(\lambda)\, d\lambda \\ &= \frac{\nu^{\alpha t}}{y!\,\Gamma(\alpha t)} \int_0^\infty \lambda^{y + \alpha t - 1} e^{-(\nu+1)\lambda}\, d\lambda.\end{aligned}$$

Recognizing the integral as a Gamma function (Keener, 2010, p. 21) with the standard scaling identity $\int_0^\infty \lambda^{s-1} e^{-\rho \lambda}\, d\lambda = \Gamma(s)/\rho^s$, with $s = y + \alpha t$ and $\rho = \nu + 1$, gives

$$\mathbb{P}(Y = y) = \frac{\Gamma(y + \alpha t)}{y!\,\Gamma(\alpha t)} \left(\frac{\nu}{\nu + 1}\right)^{\alpha t} \left(\frac{1}{\nu + 1}\right)^{y}.$$

This is the negative binomial (NB) PMF with shape parameter $r = \alpha t$ and success probability $\nu/(\nu + 1)$:

$$\mathbb{P}(Y = y) = \binom{y + r - 1}{y} \left(\frac{\nu}{\nu + 1}\right)^{r} \left(\frac{1}{\nu + 1}\right)^{y}, \qquad y \in \mathbb{N}_0,$$

where, for potentially noninteger $r > 0$,

$$\binom{y + r - 1}{y} := \frac{\Gamma(y + r)}{\Gamma(r)\Gamma(y + 1)}.$$

This matches the parameterization adopted in Section 2.

*Mean and Variance.* Using the law of total expectation (Casella and Berger, 2002, Theorem 4.4.3) and variance (Casella and Berger, 2002, Theorem 4.4.7),

$$\mathbb{E}[Y] = \mathbb{E}_\lambda[\mathbb{E}[Y \mid \lambda]] = \mathbb{E}_\lambda[\lambda] = \frac{\alpha t}{\nu},$$

$$\begin{aligned}\mathrm{Var}(Y) &= \mathbb{E}_\lambda[\mathrm{Var}(Y \mid \lambda)] + \mathrm{Var}_\lambda(\mathbb{E}[Y \mid \lambda]) \\ &= \mathbb{E}_\lambda[\lambda] + \mathrm{Var}_\lambda[\lambda] = \frac{\alpha t}{\nu} + \frac{\alpha t}{\nu^2} = \frac{\alpha t(\nu+1)}{\nu^2},\end{aligned}$$

where the Gamma mean and variance follow from Casella and Berger (2002, p. 624) with shape $\alpha t$ and scale $\beta = 1/\nu$, giving $\mathbb{E}_\lambda[\lambda] = \alpha t \cdot \beta = \alpha t/\nu$ and $\mathrm{Var}_\lambda(\lambda) = \alpha t \cdot \beta^2 = \alpha t/\nu^2$. The NB variance satisfies $\mathrm{Var}(Y) = \mathbb{E}[Y]\,(1 + 1/\nu)$, implying overdispersion relative to the Poisson model.

*Skewness.* The Poisson–Gamma mixture gives the moment generating function of $Y$ directly: $M_Y(s) = \mathbb{E}_\lambda\left[e^{\lambda(e^s-1)}\right] = M_\lambda(e^s - 1)$, using the Poisson and Gamma moment generating functions (Casella and Berger, 2002, pp. 622, 624), with $M_\lambda(u) = \{\nu/(\nu-u)\}^{\alpha t}$ for $u < \nu$. The cumulant generating function is therefore $K_Y(s) = -\alpha t \log\{1 - (e^s-1)/\nu\}$, and differentiating three times at $s = 0$ gives the cumulants

$$\kappa_1 = \frac{\alpha t}{\nu}, \qquad \kappa_2 = \frac{\alpha t(\nu+1)}{\nu^2}, \qquad \kappa_3 = \frac{\alpha t(\nu+1)(\nu+2)}{\nu^3},$$

the first two of which agree with the mean and variance above. With $r = \alpha t$, the skewness is

$$\frac{\kappa_3}{\kappa_2^{3/2}} = \frac{r(\nu+1)(\nu+2)/\nu^3}{\{r(\nu+1)/\nu^2\}^{3/2}} = \frac{\nu+2}{\{r(\nu+1)\}^{1/2}},$$

the expression used in Section 2. As $\nu \to 0$ it tends to $2/r^{1/2}$, and as $\nu \to \infty$ with $r = \mu\nu$ it tends to $\mu^{-1/2}$, the Poisson value.

**2. Score, Hessian, and Fisher Information for the Hybrid Model.** This section derives the score, Hessian, and Fisher information for the hybrid working model. Conditions (A)–(E) referenced throughout are those of Section 4.1. The working likelihood factorizes into discrete and continuous parts, $L(\mathbf{y} \mid \boldsymbol{\theta}) = L_D(\mathbf{y}_D \mid \boldsymbol{\theta}) L_C(\mathbf{y}_C \mid \boldsymbol{\theta})$, where $\alpha_{ij} = \exp(X_i^\top \beta_j)$ and $r_{ij} = \alpha_{ij} t_i$. The NB factor is

$$L_D(\mathbf{y}_D \mid \boldsymbol{\theta}) = \prod_{i=1}^n \prod_{j \in D_i} \frac{\Gamma(y_{ij} + \alpha_{ij} t_i)}{\Gamma(y_{ij}+1)\Gamma(\alpha_{ij} t_i)} \left(\frac{\nu_j}{1+\nu_j}\right)^{\alpha_{ij} t_i} \left(\frac{1}{1+\nu_j}\right)^{y_{ij}},$$

and the normal factor is

$$L_C(\mathbf{y}_C \mid \boldsymbol{\theta}) = \prod_{i=1}^n \prod_{j \in C_i} \frac{1}{\sqrt{2\pi\sigma_{ij}^2}} \exp\left\{-\frac{(y_{ij} - \mu_{ij})^2}{2\sigma_{ij}^2}\right\}.$$

Dropping terms that do not depend on $\boldsymbol{\theta}$, including $-\log\Gamma(y_{ij}+1)$ for the NB contributions and $-\frac{1}{2}\log(2\pi)$ for the normal contributions, the corresponding working log-likelihood is

$$\begin{aligned}\ell(\boldsymbol{\theta}; \mathbf{y}) = \sum_{i=1}^n \sum_{j \in D_i} \Big[ &\log\Gamma(y_{ij} + \alpha_{ij} t_i) - \log\Gamma(\alpha_{ij} t_i) \\ &+ \alpha_{ij} t_i \log \nu_j - (\alpha_{ij} t_i + y_{ij}) \log(1+\nu_j) \Big] \\ + \sum_{i=1}^n \sum_{j \in C_i} &\left[ -\log\sigma_{ij} - \frac{(y_{ij} - \mu_{ij})^2}{2\sigma_{ij}^2} \right].\end{aligned}$$

The score and Hessian below are obtained by differentiating $\ell$. The parameter vector decomposes as

$$\boldsymbol{\theta} = (\boldsymbol{\theta}_1^\top, \dots, \boldsymbol{\theta}_d^\top)^\top, \qquad \boldsymbol{\theta}_j = (\beta_j^\top, \nu_j)^\top,$$

where each component-specific parameter vector has dimension $K+1$.

We first derive the quantities for a single component. To simplify notation, we suppress the component index $j$ throughout the component-specific derivations and write

$$r_i = r_{ij} = \exp(X_i^\top \beta_j) t_i, \qquad \beta = \beta_j, \qquad \nu = \nu_j,$$

and

$$\theta = (\beta^\top, \nu)^\top.$$

Similarly, we write $y_i$ for $y_{ij}$ and $\theta_0$ for the true value of the component-specific parameter $\boldsymbol{\theta}_{j,0}$. Thus,

$$r_{i,0} = \exp(X_i^\top \beta_0) t_i, \qquad \theta_0 = (\beta_0^\top, \nu_0)^\top.$$

In the final subsection, $\boldsymbol{\theta}_0$ denotes the full $q$-dimensional true parameter. All calculations below apply identically to every component $j \in \{1, \dots, d\}$, and the full $q$-dimensional quantities used in Section 6.1 are assembled in Section 2.3.

Under Condition (A), the design pairs $Z_i = (X_i, t_i)$ are random and independent and identically distributed (i.i.d.). The score, Hessian, and observation-level Fisher information derived in this section are conditional on the realized design pair $Z_i$ and the oracle partition determined at $\boldsymbol{\theta}_0$. When differentiating the working log-likelihood with respect to $\boldsymbol{\theta}$, the realized covariates, exposures, and oracle partition are held fixed. Throughout the component-specific derivations, conditional expectations and variances are written without explicitly displaying this conditioning when no ambiguity arises.

More precisely, the observation-level conditional Fisher information is

$$\mathcal{I}_i(\boldsymbol{\theta}_0 \mid Z_i) := \mathbb{E}_{\boldsymbol{\theta}_0}\left[-H_i(\boldsymbol{\theta}_0; \mathbf{Y}_i) \mid Z_i\right].$$

For notational simplicity, we write $\mathcal{I}_i(\boldsymbol{\theta}_0)$ for $\mathcal{I}_i(\boldsymbol{\theta}_0 \mid Z_i)$ throughout this section. Thus, $\mathcal{I}_i(\boldsymbol{\theta}_0)$ is a random matrix through its dependence on $Z_i$. The population information appearing in Condition (E) and the local asymptotic normality (LAN) limit is

$$\mathcal{I}_{\boldsymbol{\theta}_0} = \mathbb{E}_Q\left[\mathcal{I}_i(\boldsymbol{\theta}_0 \mid Z_i)\right],$$

where the outer expectation is with respect to the joint design distribution $Q$.

2.1. *Negative Binomial Component.* For notational compactness, define

$$A_{i,0} = \psi(y_i + r_{i,0}) - \psi(r_{i,0}) + \log \frac{\nu_0}{1+\nu_0},$$

$$T_{i,0} = \psi^{(1)}(y_i + r_{i,0}) - \psi^{(1)}(r_{i,0}).$$

Here $A_{i,0}$ is the derivative of the NB log-likelihood with respect to $r_i$ at $\theta_0$, and $T_{i,0}$ is its second derivative. By the information identity (Casella and Berger, 2002, Lemma 7.3.11), which applies because both working densities are smooth in $\boldsymbol{\theta}$, conditionally on $Z_i$ and on the oracle assignment $j \in D_i$,

$$\mathbb{E}_{\theta_0}\left[A_{i,0} \mid Z_i\right] = 0,$$

and

$$-\mathbb{E}_{\theta_0}\left[T_{i,0} \mid Z_i\right] = \mathcal{I}^{(D)}_{r_i r_i}(Z_i) = \psi^{(1)}(r_{i,0}) - \mathbb{E}_{\theta_0}\left[\psi^{(1)}(Y_i + r_{i,0}) \mid Z_i\right] \geq 0.$$

*Score.*

$$s_i^{(D)}(\theta_0; y_i) = \begin{pmatrix} r_{i,0}\, A_{i,0}\, X_i \\ \dfrac{r_{i,0} - \nu_0 y_i}{\nu_0(1+\nu_0)} \end{pmatrix} \in \mathbb{R}^{K+1} \tag{1}$$

where the upper block is a $K$-vector and the lower block is a scalar.

*Hessian.*

$$H_i^{(D)}(\theta_0; y_i) = \begin{pmatrix} [H_i^{(D)}]_{\beta\beta} & [H_i^{(D)}]_{\beta\nu} \\ [H_i^{(D)}]_{\nu\beta} & [H_i^{(D)}]_{\nu\nu} \end{pmatrix} \in \mathbb{R}^{(K+1)\times(K+1)} \tag{2}$$

where

$$[H_i^{(D)}]_{\beta\beta} = \left(r_{i,0}A_{i,0} + r_{i,0}^2\, T_{i,0}\right) X_i X_i^\top \tag{3}$$

$$[H_i^{(D)}]_{\beta\nu} = \frac{r_{i,0}}{\nu_0(1+\nu_0)} X_i \tag{4}$$

$$[H_i^{(D)}]_{\nu\beta} = \frac{r_{i,0}}{\nu_0(1+\nu_0)} X_i^\top \tag{5}$$

$$[H_i^{(D)}]_{\nu\nu} = -\frac{r_{i,0}}{\nu_0^2} + \frac{r_{i,0} + y_i}{(1+\nu_0)^2}. \tag{6}$$

*Fisher information.* The conditional Fisher information is

$$\mathcal{I}_i^{(D)}(\theta_0 \mid Z_i) = \mathbb{E}_{\theta_0}\left[-H_i^{(D)}(\theta_0; Y_i) \mid Z_i\right],$$

computed block by block using $\mathbb{E}_{\theta_0}[A_{i,0} \mid Z_i] = 0$ and $\mathbb{E}_{\theta_0}[Y_i \mid Z_i] = r_{i,0}/\nu_0$:

$$\mathcal{I}_i^{(D)}(\theta_0) = \begin{pmatrix} r_{i,0}^2\, \mathcal{I}_{r_i r_i}^{(D)}(Z_i)\, X_i X_i^\top & -\dfrac{r_{i,0}}{\nu_0(1+\nu_0)} X_i \\ -\dfrac{r_{i,0}}{\nu_0(1+\nu_0)} X_i^\top & \dfrac{r_{i,0}}{\nu_0^2(1+\nu_0)} \end{pmatrix} \in \mathbb{R}^{(K+1)\times(K+1)}. \tag{7}$$

2.2. *Normal Component.*

*Score.*

$$s_i^{(C)}(\theta_0; y_i) = \begin{pmatrix} \dfrac{1}{2}\left[\dfrac{\nu_0^2 y_i^2 - r_{i,0}^2}{r_{i,0}(1+\nu_0)} - 1\right] X_i \\ \dfrac{1}{\nu_0} - \dfrac{1}{2(1+\nu_0)} + \dfrac{r_{i,0}^2 + 2r_{i,0}y_i - \nu_0(\nu_0+2)y_i^2}{2r_{i,0}(1+\nu_0)^2} \end{pmatrix} \in \mathbb{R}^{K+1}, \tag{8}$$

where the upper block is a $K$-vector and the lower block is a scalar.

*Hessian.*

$$H_i^{(C)}(\theta_0; y_i) = \begin{pmatrix} [H_i^{(C)}]_{\beta\beta} & [H_i^{(C)}]_{\beta\nu} \\ [H_i^{(C)}]_{\nu\beta} & [H_i^{(C)}]_{\nu\nu} \end{pmatrix} \in \mathbb{R}^{(K+1)\times(K+1)}, \tag{9}$$

where

$$[H_i^{(C)}]_{\beta\beta} = -\frac{1}{2(1+\nu_0)}\left[\frac{\nu_0^2 y_i^2}{r_{i,0}} + r_{i,0}\right] X_i X_i^\top \tag{10}$$

$$[H_i^{(C)}]_{\beta\nu} = \frac{1}{2(1+\nu_0)^2}\left[\frac{\nu_0(\nu_0+2)y_i^2}{r_{i,0}} + r_{i,0}\right] X_i \tag{11}$$

$$[H_i^{(C)}]_{\nu\beta} = \frac{1}{2(1+\nu_0)^2}\left[\frac{\nu_0(\nu_0+2)y_i^2}{r_{i,0}} + r_{i,0}\right] X_i^\top \tag{12}$$

$$[H_i^{(C)}]_{\nu\nu} = -\frac{1}{\nu_0^2} + \frac{1}{2(1+\nu_0)^2} - \frac{(r_{i,0}+y_i)^2}{r_{i,0}(1+\nu_0)^3}. \tag{13}$$

*Fisher information.* The conditional Fisher information is

$$\mathcal{I}_i^{(C)}(\theta_0 \mid Z_i) = \mathbb{E}_{\theta_0}\left[-H_i^{(C)}(\theta_0; Y_i) \mid Z_i\right].$$

It is computed block by block using the conditional normal moments

$$\mathbb{E}_{\theta_0}[Y_i \mid Z_i] = \frac{r_{i,0}}{\nu_0} \quad \text{and} \quad \mathbb{E}_{\theta_0}[Y_i^2 \mid Z_i] = \frac{r_{i,0}(1+\nu_0+r_{i,0})}{\nu_0^2},$$

which follow from

$$Y_i \mid Z_i \sim \mathcal{N}\left(\frac{r_{i,0}}{\nu_0}, \frac{r_{i,0}(1+\nu_0)}{\nu_0^2}\right).$$

Writing

$$c_{i,0} := \frac{1}{\nu_0(1+\nu_0)}\left[r_{i,0} + \frac{\nu_0+2}{2}\right], \qquad v_{i,0} := \frac{1+\nu_0+r_{i,0}}{\nu_0^2(1+\nu_0)} + \frac{2-\nu_0^2}{2\nu_0^2(1+\nu_0)^2},$$

the conditional Fisher information is

$$\mathcal{I}_i^{(C)}(\theta_0) = \begin{pmatrix} \left[\frac{r_{i,0}}{1+\nu_0} + \frac{1}{2}\right] X_i X_i^\top & -c_{i,0} X_i \\ -c_{i,0} X_i^\top & v_{i,0} \end{pmatrix} \in \mathbb{R}^{(K+1)\times(K+1)}. \tag{14}$$

2.3. *Full Hybrid Model.* For the full $d$-component model, the component-specific quantities derived above are assembled according to the observation-specific oracle partition. For each observation $i$ and component $j$, the contribution is normal if $j \in C_i$ and NB if $j \in D_i$. The component index $j$, suppressed above, is restored here to assemble the full $q$-dimensional quantities.

*Score.*

$$s_i(\boldsymbol{\theta}_0; \mathbf{y}_i) = \begin{pmatrix} s_{i1}(\boldsymbol{\theta}_{1,0}; y_{i1}) \\ \vdots \\ s_{id}(\boldsymbol{\theta}_{d,0}; y_{id}) \end{pmatrix} \in \mathbb{R}^q, \tag{15}$$

where, for each component $j$,

$$s_{ij}(\boldsymbol{\theta}_{j,0}; y_{ij}) = \begin{cases} s_i^{(C)}(\boldsymbol{\theta}_{j,0}; y_{ij}), & j \in C_i, \\ s_i^{(D)}(\boldsymbol{\theta}_{j,0}; y_{ij}), & j \in D_i. \end{cases} \tag{16}$$

*Hessian.*

$$H_i(\boldsymbol{\theta}_0;\mathbf{y}_i) = \operatorname{diag}\{H_{i1}(\boldsymbol{\theta}_{1,0};y_{i1}),\ldots,H_{id}(\boldsymbol{\theta}_{d,0};y_{id})\} \in \mathbb{R}^{q\times q}, \tag{17}$$

where

$$H_{ij}(\boldsymbol{\theta}_{j,0};y_{ij}) = \begin{cases} H_i^{(C)}(\boldsymbol{\theta}_{j,0};y_{ij}), & j\in C_i,\\ H_i^{(D)}(\boldsymbol{\theta}_{j,0};y_{ij}), & j\in D_i. \end{cases} \tag{18}$$

*Fisher information.* Conditional on the realized design pair $Z_i$ and the oracle partition, the observation-level Fisher information is

$$\mathcal{I}_i(\boldsymbol{\theta}_0 \mid Z_i) = \operatorname{diag}\{\mathcal{I}_{i1}(\boldsymbol{\theta}_{1,0}\mid Z_i),\ldots,\mathcal{I}_{id}(\boldsymbol{\theta}_{d,0}\mid Z_i)\} \in \mathbb{R}^{q\times q}, \tag{19}$$

where

$$\mathcal{I}_{ij}(\boldsymbol{\theta}_{j,0}\mid Z_i) = \begin{cases} \mathcal{I}_i^{(C)}(\boldsymbol{\theta}_{j,0}\mid Z_i), & j\in C_i,\\ \mathcal{I}_i^{(D)}(\boldsymbol{\theta}_{j,0}\mid Z_i), & j\in D_i. \end{cases} \tag{20}$$

Because the working log-likelihood is additive across components, the cross-component derivatives vanish:

$$\frac{\partial^2 \ell_i(\boldsymbol{\theta};\mathbf{y}_i)}{\partial\boldsymbol{\theta}_j\,\partial\boldsymbol{\theta}_k^\top} = \mathbf{0}, \qquad j\neq k.$$

Consequently, the observation-level score is obtained by stacking the component-specific scores, while the observation-level Hessian is block diagonal across components. The negative expected Hessian is therefore also block diagonal.

By the information identity, each component-specific score has conditional mean zero and conditional variance equal to the conditional information:

$$\mathbb{E}_{\boldsymbol{\theta}_0}\left[s_{ij}(\boldsymbol{\theta}_{j,0};Y_{ij}) \mid Z_i\right] = \mathbf{0}$$

and

$$\operatorname{Var}_{\boldsymbol{\theta}_0}\left[s_{ij}(\boldsymbol{\theta}_{j,0};Y_{ij}) \mid Z_i\right] = -\mathbb{E}_{\boldsymbol{\theta}_0}\left[H_{ij}(\boldsymbol{\theta}_{j,0};Y_{ij}) \mid Z_i\right] = \mathcal{I}_{ij}(\boldsymbol{\theta}_{j,0}\mid Z_i).$$

Under the working conditional-independence model, the component responses are independent conditional on $Z_i$ and the oracle partition. Because each component-specific score is a function of the corresponding component response, for $j\neq k$,

$$\operatorname{Cov}_{\boldsymbol{\theta}_0}\left[s_{ij}(\boldsymbol{\theta}_{j,0};Y_{ij}), s_{ik}(\boldsymbol{\theta}_{k,0};Y_{ik}) \mid Z_i\right] = \mathbf{0}.$$

Consequently,

$$\mathbb{E}_{\boldsymbol{\theta}_0}\left[s_i(\boldsymbol{\theta}_0;\mathbf{Y}_i) \mid Z_i\right] = \mathbf{0}$$

and

$$\operatorname{Var}_{\boldsymbol{\theta}_0}\left[s_i(\boldsymbol{\theta}_0;\mathbf{Y}_i) \mid Z_i\right] = \mathcal{I}_i(\boldsymbol{\theta}_0\mid Z_i).$$

In particular,

$$\mathcal{I}_i(\boldsymbol{\theta}_0\mid Z_i) = -\mathbb{E}_{\boldsymbol{\theta}_0}\left[H_i(\boldsymbol{\theta}_0;\mathbf{Y}_i) \mid Z_i\right].$$

Taking expectations over the joint design distribution $Q$ and using the law of total expectation (Casella and Berger, 2002, Theorem 4.4.3) gives

$$\mathbb{E}_{\boldsymbol{\theta}_0}\left[s_i(\boldsymbol{\theta}_0;\mathbf{Y}_i)\right] = \mathbf{0}.$$

Moreover, because the conditional score mean is zero, the law of total variance (Casella and Berger, 2002, Theorem 4.4.7) gives

$$\mathrm{Var}_{\boldsymbol{\theta}_0}\left[s_i(\boldsymbol{\theta}_0;\mathbf{Y}_i)\right]=\mathbb{E}_Q\left[\mathrm{Var}_{\boldsymbol{\theta}_0}\left(s_i(\boldsymbol{\theta}_0;\mathbf{Y}_i)\mid Z_i\right)\right]=\mathbb{E}_Q\left[\mathcal{I}_i(\boldsymbol{\theta}_0\mid Z_i)\right]=\mathcal{I}_{\boldsymbol{\theta}_0}.$$

Similarly,

$$-\mathbb{E}_{\boldsymbol{\theta}_0}\left[H_i(\boldsymbol{\theta}_0;\mathbf{Y}_i)\right]=\mathbb{E}_Q\left[\mathcal{I}_i(\boldsymbol{\theta}_0\mid Z_i)\right]=\mathcal{I}_{\boldsymbol{\theta}_0}.$$

Thus, the population information appearing in Condition (E) is the design expectation of the observation-level conditional information derived above.

### 3. Proof of Theorem 2 (Joint Fiducial Density).

REMARK 1. The injectivity requirement in Theorem 2(iv) is an assumption rather than a consequence of the rank condition in Theorem 2(iii), and it can fail. Writing $\xi_{ij}=\Phi^{-1}(u_{ij}^C)$ as in Section 3.3, the continuous structural equation is $\mu_{ij}(\boldsymbol{\theta}_j)+\sigma_{ij}(\boldsymbol{\theta}_j)\xi_{ij}$. When the design includes an intercept, replacing $(\beta_j,\nu_j)$ by $(\beta_j+ce_1,\,\nu_j e^c)$ leaves every $\mu_{ij}$ unchanged, so $\boldsymbol{\theta}'\mapsto G_C(\mathbf{u}_C,\boldsymbol{\theta}')$ is not injective at $\mathbf{u}_C=\frac{1}{2}\mathbf{1}$, where $\xi=\mathbf{0}$ and the scale term drops out. More generally, for any pair $\boldsymbol{\theta}\neq\boldsymbol{\theta}'$ there is a $\xi$ at which the two structural equations agree, so injectivity is a restriction on the auxiliary configurations the proof actually uses, namely $\mathbf{u}_C(\boldsymbol{\theta})=G_C^{-1}(\mathbf{y}_C,\boldsymbol{\theta})$ for $\boldsymbol{\theta}\in\Theta$.

Substituting $\xi_{ij}=\{y_{ij}-\mu_{ij}(\boldsymbol{\theta})\}/\sigma_{ij}(\boldsymbol{\theta})$ shows that injectivity fails at some $\mathbf{u}_C(\boldsymbol{\theta})$ only if, for some pair $\boldsymbol{\theta}\neq\boldsymbol{\theta}'$,

$$y_{ij}\left\{\sigma_{ij}(\boldsymbol{\theta})-\sigma_{ij}(\boldsymbol{\theta}')\right\}=\mu_{ij}(\boldsymbol{\theta}')\sigma_{ij}(\boldsymbol{\theta})-\mu_{ij}(\boldsymbol{\theta})\sigma_{ij}(\boldsymbol{\theta}')$$

for every continuous $i$ in the block. The set of $\mathbf{y}_C$ satisfying this for some such pair is the image of a $2(K+1)$-parameter family in $\mathbb{R}^{n_j^C}$ and therefore has Lebesgue measure zero whenever $n_j^C>2(K+1)$. Under the continuous working model the requirement thus holds almost surely for every block with $n_j^C>2(K+1)$, and for the observed data it holds outside that null set.

PROOF. Throughout this proof, we condition on the realized design $\mathbf{Z}=(Z_1,\ldots,Z_n)$ with $Z_i=(X_i,t_i)$, and on the oracle partition determined at $\boldsymbol{\theta}_0$. The design distribution $Q$ does not depend on $\boldsymbol{\theta}$ and therefore contributes no term to the conditional working likelihood or to the fiducial Jacobian. For notational simplicity, the conditioning on $\mathbf{Z}$ and the oracle partition is suppressed below.

Let $\mathcal{Y}_D\subset\mathbb{N}_0^{n_D}$ be the support of the discrete component $G_D(\mathbf{U}_D^\star,\boldsymbol{\theta})$. Because $\mathcal{Y}_D$ lies on the integer lattice, any two distinct points $\mathbf{y}_1,\mathbf{y}_2\in\mathcal{Y}_D$ satisfy

$$\|\mathbf{y}_1-\mathbf{y}_2\|_2\geq 1.$$

Since we will take $\epsilon\to 0$, we assume the fiducial tolerance $\epsilon<1$ throughout.

**1. Setup and Minimization Criteria.** The generalized fiducial distribution (GFD) is characterized by the random variable $\bar{\boldsymbol{\theta}}$, defined as the weak limit of the conditional estimator that minimizes the $L_2$ norm of the structural error (Hannig et al., 2016, Definition 1):

$$\bar{\boldsymbol{\theta}}=\lim_{\epsilon\to 0}\left[\arg\min_{\boldsymbol{\theta}^\star}\|\mathbf{y}-G(\mathbf{U}^\star,\boldsymbol{\theta}^\star)\|_2\,\middle|\,\min_{\boldsymbol{\theta}^\star}\|\mathbf{y}-G(\mathbf{U}^\star,\boldsymbol{\theta}^\star)\|_2\leq\epsilon\right],$$

where $G(\mathbf{U}^\star,\boldsymbol{\theta}^\star)$ is the data generating equation and $\mathbf{U}^\star$ is an independent copy of the random error vector (the fiducial sample).

Because the oracle partition is determined at $\boldsymbol{\theta}_0$, it is held fixed throughout the minimization over $\boldsymbol{\theta}$. We partition the data vector and structural equations into continuous ($C$) and discrete ($D$) components according to the oracle partition of Theorem 2:

$$\mathbf{y} = \begin{pmatrix} \mathbf{y}_C \\ \mathbf{y}_D \end{pmatrix}, \quad G(\mathbf{U}^\star, \boldsymbol{\theta}) = \begin{pmatrix} G_C(\mathbf{U}_C^\star, \boldsymbol{\theta}) \\ G_D(\mathbf{U}_D^\star, \boldsymbol{\theta}) \end{pmatrix},$$

where $\mathbf{y}_C \in \mathbb{R}^{n_C}$ corresponds to the normal approximation components and $\mathbf{y}_D \in \mathbb{N}_0^{n_D}$ corresponds to the NB components.

Since $\arg\min \|\cdot\|_2 \equiv \arg\min \|\cdot\|_2^2$, we work with the squared objective $S(\boldsymbol{\theta}^\star, \mathbf{U}^\star) = \|\mathbf{y} - G(\mathbf{U}^\star, \boldsymbol{\theta}^\star)\|_2^2$, which decomposes additively along the block structure:

$$S(\boldsymbol{\theta}^\star, \mathbf{U}^\star) = \underbrace{\|\mathbf{y}_D - G_D(\mathbf{U}_D^\star, \boldsymbol{\theta}^\star)\|_2^2}_{\text{Discrete Term}} + \underbrace{\|\mathbf{y}_C - G_C(\mathbf{U}_C^\star, \boldsymbol{\theta}^\star)\|_2^2}_{\text{Continuous Term}}.$$

The conditioning event is then

$$\mathcal{C}_\epsilon = \left\{ \min_{\boldsymbol{\theta}^\star} S(\boldsymbol{\theta}^\star, \mathbf{U}^\star) \le \epsilon^2 \right\},$$

which we analyze for radius $\epsilon < 1$ in the subsequent sections. When the minimum is attained, we use a measurable selection from the set of minimizers, with a fixed measurable tie-breaking rule if the minimizer is not unique. Because the discrete inverse images are defined by half-open inequalities, the infimum need not be attained for every realization of the auxiliary variables. In that case, we use a measurable approximate minimizer $\boldsymbol{\theta}_\epsilon^\star(\mathbf{U}^\star)$ satisfying

$$S(\boldsymbol{\theta}_\epsilon^\star(\mathbf{U}^\star), \mathbf{U}^\star) \le \inf_{\boldsymbol{\theta} \in \Theta} S(\boldsymbol{\theta}, \mathbf{U}^\star) + a_\epsilon, \qquad a_\epsilon = o(\epsilon^2). \tag{21}$$

This convention does not alter the limiting calculation, because replacing the radius $\epsilon$ by $\sqrt{\epsilon^2 + a_\epsilon} = \epsilon\{1 + o(1)\}$ changes neither the limiting conditional law nor its density. The argument below also shows that configurations near the half-open boundaries of the discrete inverse images are negligible in the limit, so the limiting conditional law is independent of the treatment of those boundary configurations.

**2. The Discrete Component.** We analyze the conditioning event $\mathcal{C}_\epsilon = \{\min_{\boldsymbol{\theta}^\star} S(\boldsymbol{\theta}^\star, \mathbf{U}^\star) \le \epsilon^2\}$ for the discrete components. Since the observed count data $\mathbf{y}_D$ and the structural equation output $G_D$ both reside on the integer lattice $\mathbb{N}_0^{n_D}$, the Euclidean distance between them is discrete. The squared distance between any two distinct integer vectors is at least 1:

$$G_D(\mathbf{U}_D^\star, \boldsymbol{\theta}^\star) \ne \mathbf{y}_D \implies \|\mathbf{y}_D - G_D(\mathbf{U}_D^\star, \boldsymbol{\theta}^\star)\|_2^2 \ge 1.$$

However, our conditioning event requires the total sum of squares to be less than $\epsilon^2$. For a radius $\epsilon < 1$, this imposes a strict constraint on the discrete term, effectively forcing an exact match:

$$\|\mathbf{y}_D - G_D(\mathbf{U}_D^\star, \boldsymbol{\theta}^\star)\|_2^2 \le \epsilon^2 < 1 \implies \|\mathbf{y}_D - G_D(\mathbf{U}_D^\star, \boldsymbol{\theta}^\star)\|_2^2 = 0.$$

Consequently, the joint minimization in $\mathcal{C}_\epsilon$ reduces to a continuous minimization over the feasible set of parameters consistent with the discrete data:

$$\mathcal{C}_\epsilon = \left\{ \min_{\boldsymbol{\theta}^\star : G_D(\mathbf{U}_D^\star, \boldsymbol{\theta}^\star) = \mathbf{y}_D} \|\mathbf{y}_C - G_C(\mathbf{U}_C^\star, \boldsymbol{\theta}^\star)\|_2 \le \epsilon \right\}.$$

The probability of the discrete match event, $\{G_D(\mathbf{U}_D^\star, \boldsymbol{\theta}^\star) = \mathbf{y}_D\}$, corresponds to the measure of the set of random errors $\mathbf{U}_D^\star$ that generate the observed counts given parameters

$\boldsymbol{\theta}^\star$. The discrete structural equation $G_D(\cdot, \boldsymbol{\theta}^\star)$ is the inverse cumulative distribution function (CDF) of the NB, so $G_D(u, \boldsymbol{\theta}^\star) = y$ if and only if $u \in (F(y-1; \boldsymbol{\theta}^\star), F(y; \boldsymbol{\theta}^\star)]$, an interval whose length equals the PMF $P(Y = y \mid \boldsymbol{\theta}^\star)$. Thus, for a single discrete contribution $(i, j)$ with $j \in D_i$, $P\{G_D(U_{ij}^\star, \boldsymbol{\theta}^\star) = y_{ij}\} = P(Y_{ij} = y_{ij} \mid \boldsymbol{\theta}^\star)$. Since the auxiliary variables $U_{ij}^\star$ are independent across contributions under the generalized fiducial inference (GFI) construction of Hannig et al. (2016), the joint probability factors as the product likelihood: $P_{\mathbf{U}_D^\star}(G_D(\mathbf{U}_D^\star, \boldsymbol{\theta}^\star) = \mathbf{y}_D) = L_D(\mathbf{y}_D \mid \boldsymbol{\theta}^\star)$. Thus, for each fixed $\boldsymbol{\theta}$, integration over the discrete auxiliary variables gives

$$\int \mathbf{1}\{G_D(\mathbf{u}_D, \boldsymbol{\theta}) = \mathbf{y}_D\}\, d\mathbf{u}_D = L_D(\mathbf{y}_D \mid \boldsymbol{\theta}).$$

Because the discrete exact-match probability does not shrink with $\epsilon$, it contributes this likelihood factor but no $\epsilon$-dependent geometric scaling. The shrinking-neighborhood argument, and hence the fiducial Jacobian, is therefore determined entirely by the continuous structural equations.

**3. The Continuous Component.**

We next identify the contribution of the continuous structural equations. Because the discrete feasible set $\{\boldsymbol{\theta}^\star \in \Theta : G_D(\mathbf{u}_D, \boldsymbol{\theta}^\star) = \mathbf{y}_D\}$ varies with $\mathbf{u}_D$, this calculation does not by itself justify replacing the discrete constraint by its fixed-$\boldsymbol{\theta}$ probability. Part 4 below shows that the moving boundaries of that set are negligible in the limit, so that away from them the feasibility indicator is constant in $\boldsymbol{\theta}$ and the remaining shrinking-neighborhood calculation is that of the continuous structural equation $\mathbf{Y}_C = G_C(\mathbf{U}_C, \boldsymbol{\theta})$ alone.

Let $q = d(K+1)$ denote the full parameter dimension, and define the $n_C \times q$ Jacobian matrix

$$J_C(\mathbf{u}_C, \boldsymbol{\theta}) = \frac{\partial G_C(\mathbf{u}_C, \boldsymbol{\theta})}{\partial \boldsymbol{\theta}}.$$

We verify Assumptions A.1–A.4 of Hannig et al. (2016, Appendix A) for the continuous structural equation. For each continuous contribution $(i, j)$, the structural equation is

$$G_{ij}^{(C)}(u_{ij}, \boldsymbol{\theta}) = \mu_{ij}(\boldsymbol{\theta}) + \sigma_{ij}(\boldsymbol{\theta})\Phi^{-1}(u_{ij}), \qquad u_{ij} \in (0, 1),$$

where

$$\mu_{ij}(\boldsymbol{\theta}) = \frac{\exp(X_i^\top \beta_j) t_i}{\nu_j} \quad \text{and} \quad \sigma_{ij}(\boldsymbol{\theta}) = \frac{\exp(\frac{1}{2} X_i^\top \beta_j)\sqrt{t_i(1 + \nu_j)}}{\nu_j}.$$

*Verification of Assumption A.1.* For $\nu_j > 0$ and $u_{ij} \in (0, 1)$, the functions $\mu_{ij}(\boldsymbol{\theta})$, $\sigma_{ij}(\boldsymbol{\theta})$, and $\Phi^{-1}(u_{ij})$ are continuously differentiable in their arguments. Consequently, $G_C(\mathbf{u}_C, \boldsymbol{\theta})$ has continuous partial derivatives with respect to every component of $\mathbf{u}_C$ and $\boldsymbol{\theta}$.

*Verification of Assumption A.2.* For every observed $y_{ij} \in \mathbb{R}$ and admissible $\boldsymbol{\theta}$, the continuous structural equation has the unique inverse

$$u_{ij} = \Phi\left(\frac{y_{ij} - \mu_{ij}(\boldsymbol{\theta})}{\sigma_{ij}(\boldsymbol{\theta})}\right) \in (0, 1),$$

because $\sigma_{ij}(\boldsymbol{\theta}) > 0$. Thus, $\mathbf{u}_C = G_C^{-1}(\mathbf{y}_C, \boldsymbol{\theta})$ is uniquely determined.

Moreover, because each continuous structural equation depends only on its corresponding auxiliary variable, the Jacobian with respect to $\mathbf{u}_C$ is diagonal. Its diagonal entries are

$$\frac{\partial G_{ij}^{(C)}(u_{ij}, \boldsymbol{\theta})}{\partial u_{ij}} = \frac{\sigma_{ij}(\boldsymbol{\theta})}{\phi\{\Phi^{-1}(u_{ij})\}} > 0,$$

where $\phi$ is the standard normal density. Therefore,

$$\det\left(\frac{\partial G_C(\mathbf{u}_C,\boldsymbol{\theta})}{\partial \mathbf{u}_C}\right) = \prod_{(i,j):\, j\in C_i} \frac{\sigma_{ij}(\boldsymbol{\theta})}{\phi\{\Phi^{-1}(u_{ij})\}} \neq 0.$$

*Verification of Assumption A.3.* By the first part of the full-column-rank condition in Theorem 2, $\operatorname{rank}\{J_C(\mathbf{U}_C,\boldsymbol{\theta})\} = q$ almost surely for every $\boldsymbol{\theta}\in\Theta$. The second part of the rank condition gives $\operatorname{rank}\{J_C(\boldsymbol{\theta})\} = q$ for every $\boldsymbol{\theta}\in\Theta$, where

$$J_C(\boldsymbol{\theta}) = J_C(\mathbf{u}_C,\boldsymbol{\theta})\big|_{\mathbf{u}_C = G_C^{-1}(\mathbf{y}_C,\boldsymbol{\theta})}.$$

This evaluated-rank condition ensures that the generalized determinant appearing in the fiducial density is strictly positive throughout $\Theta$.
*Verification of Assumption A.4.* The entries of $J_C(\mathbf{u}_C,\boldsymbol{\theta})$ are formed from derivatives of $\mu_{ij}(\boldsymbol{\theta})$ and $\sigma_{ij}(\boldsymbol{\theta})$ and are affine in $\Phi^{-1}(u_{ij})$. Since $\mu_{ij}$ and $\sigma_{ij}$ are smooth for $\nu_j > 0$, and $\Phi^{-1}$ is continuously differentiable on $(0,1)$, the entries of $J_C(\mathbf{u}_C,\boldsymbol{\theta})$ have continuous partial derivatives with respect to every component of $\mathbf{u}_C$ and $\boldsymbol{\theta}$. In particular, $J_C(\mathbf{u}_C,\boldsymbol{\theta})$ is continuous in $(\mathbf{u}_C,\boldsymbol{\theta})$.

Assumptions A.1–A.4 of Hannig et al. (2016) are therefore satisfied. By part (iii) of Theorem 1 of Hannig et al. (2016), the generalized fiducial contribution from the continuous structural equations under the $\ell_2$ norm is

$$L_C(\mathbf{y}_C \mid \boldsymbol{\theta}) D\{J_C(\mathbf{u}_C,\boldsymbol{\theta})\}\Bigg|_{\mathbf{u}_C = G_C^{-1}(\mathbf{y}_C,\boldsymbol{\theta})},$$

where $D(A) = \sqrt{\det(A^\top A)}$ is the generalized determinant of Hannig et al. (2016). Writing

$$J_C(\boldsymbol{\theta}) := J_C(\mathbf{u}_C,\boldsymbol{\theta})\Bigg|_{\mathbf{u}_C = G_C^{-1}(\mathbf{y}_C,\boldsymbol{\theta})},$$

we obtain

$$r_C(\boldsymbol{\theta}\mid \mathbf{y}_C) \propto L_C(\mathbf{y}_C\mid\boldsymbol{\theta})\sqrt{\det[J_C(\boldsymbol{\theta})^\top J_C(\boldsymbol{\theta})]}.$$

Following the normalization used in Borgert and Hannig (2026), we place $n^{-1}$ inside the determinant. Since $J_C(\boldsymbol{\theta})^\top J_C(\boldsymbol{\theta})$ is a $q\times q$ matrix,

$$\sqrt{\det[n^{-1}J_C(\boldsymbol{\theta})^\top J_C(\boldsymbol{\theta})]} = n^{-q/2}\sqrt{\det[J_C(\boldsymbol{\theta})^\top J_C(\boldsymbol{\theta})]}.$$

The factor $n^{-q/2}$ is independent of $\boldsymbol{\theta}$ and therefore does not affect the generalized fiducial density up to proportionality. Consequently,

$$r_C(\boldsymbol{\theta}\mid \mathbf{y}_C) \propto L_C(\mathbf{y}_C\mid\boldsymbol{\theta})\sqrt{\det[n^{-1}J_C(\boldsymbol{\theta})^\top J_C(\boldsymbol{\theta})]}.$$

**4. Moving Discrete Boundaries and Conclusion.**

The identity in Part 2 holds for each fixed $\boldsymbol{\theta}$. The generalized fiducial construction, however, conditions on $\mathcal{C}_\epsilon$ and lets $\epsilon\downarrow 0$, so the relevant object is the limiting law of a minimizer over the discrete feasible set

$$\Theta_D(\mathbf{u}_D) := \{\boldsymbol{\theta}^\star\in\Theta:\ G_D(\mathbf{u}_D,\boldsymbol{\theta}^\star) = \mathbf{y}_D\},$$

which varies with $\mathbf{u}_D$. We now show that the moving boundaries of this set are negligible in the limit, so that the discrete feasibility indicator may be evaluated at a fixed $\boldsymbol{\theta}$.

Let $\mathcal{K}_D := \{(i,j) : j \in D_i\}$, so that $|\mathcal{K}_D| = n_D$, and for $k = (i,j) \in \mathcal{K}_D$ define the inverse-CDF endpoints

$$a_k(\boldsymbol{\theta}) := F_{\mathrm{NB}}(y_{ij} - 1; \boldsymbol{\theta}_j), \qquad b_k(\boldsymbol{\theta}) := F_{\mathrm{NB}}(y_{ij}; \boldsymbol{\theta}_j).$$

The inverse-CDF characterization of Part 2 gives

$$G_D(\mathbf{u}_D, \boldsymbol{\theta}) = \mathbf{y}_D \iff a_k(\boldsymbol{\theta}) < u_k \leq b_k(\boldsymbol{\theta}) \quad \text{for every } k \in \mathcal{K}_D. \tag{22}$$

For the fixed observed sample, $a_k$ and $b_k$ are finite sums of NB masses, each smooth in $(r_{ij}, \nu_j)$ for $r_{ij} > 0$ and $\nu_j > 0$, and are therefore continuously differentiable on the valid parameter space. That space is convex, so the convex hull $\mathrm{conv}(\Theta)$ is a compact subset of it, and since $\mathcal{K}_D$ is finite there exists $B < \infty$ with

$$\sup_{\boldsymbol{\theta} \in \mathrm{conv}(\Theta)} \max_{k \in \mathcal{K}_D} \{\|\nabla a_k(\boldsymbol{\theta})\|, \|\nabla b_k(\boldsymbol{\theta})\|\} \leq B.$$

Because the segment joining any $\boldsymbol{\theta}, \boldsymbol{\theta}' \in \Theta$ lies in $\mathrm{conv}(\Theta)$, the mean value inequality gives

$$\max_{k \in \mathcal{K}_D} \big\{|a_k(\boldsymbol{\theta}') - a_k(\boldsymbol{\theta})|, |b_k(\boldsymbol{\theta}') - b_k(\boldsymbol{\theta})|\big\} \leq B\|\boldsymbol{\theta}' - \boldsymbol{\theta}\| \tag{23}$$

for all $\boldsymbol{\theta}, \boldsymbol{\theta}' \in \Theta$.

We next bound the parameter displacement compatible with the conditioning event. For $\boldsymbol{\theta} \in \Theta$, write

$$\mathbf{u}_C(\boldsymbol{\theta}) := G_C^{-1}(\mathbf{y}_C, \boldsymbol{\theta}).$$

The evaluated inverse graph

$$\mathcal{G}_C := \big\{\big(\mathbf{u}_C(\boldsymbol{\theta}), \boldsymbol{\theta}\big) : \boldsymbol{\theta} \in \Theta\big\}$$

is compact, because $\Theta$ is compact and the explicit normal inverse established in Part 3 is continuous.

We first exclude parameter values that remain a fixed positive distance from $\boldsymbol{\theta}$. Define

$$H(\boldsymbol{\theta}, \boldsymbol{\theta}') := \big\|G_C\big(\mathbf{u}_C(\boldsymbol{\theta}), \boldsymbol{\theta}'\big) - \mathbf{y}_C\big\|_2.$$

Since $G_C\big(\mathbf{u}_C(\boldsymbol{\theta}), \boldsymbol{\theta}\big) = \mathbf{y}_C$, and the injectivity requirement of Theorem 2(iv) makes

$$\boldsymbol{\theta}' \longmapsto G_C\big(\mathbf{u}_C(\boldsymbol{\theta}), \boldsymbol{\theta}'\big)$$

injective on $\Theta$, we have

$$H(\boldsymbol{\theta}, \boldsymbol{\theta}') = 0 \iff \boldsymbol{\theta}' = \boldsymbol{\theta}.$$

For every fixed $\delta > 0$ the set

$$\mathcal{A}_\delta := \big\{(\boldsymbol{\theta}, \boldsymbol{\theta}') \in \Theta \times \Theta : \|\boldsymbol{\theta}' - \boldsymbol{\theta}\|_2 \geq \delta\big\}$$

is compact, and $H$ is continuous and strictly positive on $\mathcal{A}_\delta$. Consequently

$$\eta_\delta := \inf_{(\boldsymbol{\theta}, \boldsymbol{\theta}') \in \mathcal{A}_\delta} H(\boldsymbol{\theta}, \boldsymbol{\theta}') > 0.$$

Hence if $\big\|G_C\big(\mathbf{u}_C(\boldsymbol{\theta}), \boldsymbol{\theta}^\star\big) - \mathbf{y}_C\big\|_2 \leq \epsilon$ with $\epsilon < \eta_\delta$, then necessarily $\|\boldsymbol{\theta}^\star - \boldsymbol{\theta}\|_2 < \delta$. Uniformly in $\boldsymbol{\theta} \in \Theta$, candidate parameters satisfying the continuous residual constraint therefore cannot remain a fixed positive distance from $\boldsymbol{\theta}$ as $\epsilon \downarrow 0$.

By compactness of $\mathcal{G}_C$, continuity of $J_C$, and the evaluated full-column-rank condition of Theorem 2(iii),

$$c_0 := \inf_{\boldsymbol{\theta} \in \Theta} \sigma_{\min}\big\{J_C\big(\mathbf{u}_C(\boldsymbol{\theta}), \boldsymbol{\theta}\big)\big\} > 0.$$

The explicit form of the continuous structural equation established in Part 3 is twice continuously differentiable in $\boldsymbol{\theta}$ on the valid parameter space. The set $\mathcal{U}_C(\mathbf{y}_C) = \{\mathbf{u}_C(\boldsymbol{\theta}) : \boldsymbol{\theta} \in \Theta\}$ is compact. Moreover, $\mathrm{conv}(\Theta)$ is a compact subset of the valid parameter space: in particular, Condition (D) bounds every dispersion coordinate uniformly away from zero, and this property is preserved under convex combinations. Hence, the second derivatives of $(\mathbf{u}_C, \boldsymbol{\theta}') \longmapsto G_C(\mathbf{u}_C, \boldsymbol{\theta}')$ with respect to $\boldsymbol{\theta}'$ are uniformly bounded on $\mathcal{U}_C(\mathbf{y}_C) \times \mathrm{conv}(\Theta)$. There therefore exists $M < \infty$ such that, for every $\boldsymbol{\theta}, \boldsymbol{\theta}' \in \Theta$, Taylor's theorem along the segment joining $\boldsymbol{\theta}$ and $\boldsymbol{\theta}'$ gives

$$G_C\big(\mathbf{u}_C(\boldsymbol{\theta}), \boldsymbol{\theta}'\big) - G_C\big(\mathbf{u}_C(\boldsymbol{\theta}), \boldsymbol{\theta}\big) = J_C\big(\mathbf{u}_C(\boldsymbol{\theta}), \boldsymbol{\theta}\big)(\boldsymbol{\theta}' - \boldsymbol{\theta}) + R(\boldsymbol{\theta}, \boldsymbol{\theta}'),$$

where

$$\|R(\boldsymbol{\theta}, \boldsymbol{\theta}')\|_2 \le \frac{M}{2}\|\boldsymbol{\theta}' - \boldsymbol{\theta}\|_2^2.$$

Set $\delta_0 := c_0/M$. For $\|\boldsymbol{\theta}' - \boldsymbol{\theta}\|_2 \le \delta_0$ it follows that

$$\big\|G_C\big(\mathbf{u}_C(\boldsymbol{\theta}), \boldsymbol{\theta}'\big) - G_C\big(\mathbf{u}_C(\boldsymbol{\theta}), \boldsymbol{\theta}\big)\big\|_2 \ge c_0\|\boldsymbol{\theta}' - \boldsymbol{\theta}\|_2 - \frac{M}{2}\|\boldsymbol{\theta}' - \boldsymbol{\theta}\|_2^2 \ge \frac{c_0}{2}\|\boldsymbol{\theta}' - \boldsymbol{\theta}\|_2$$

whenever $\|\boldsymbol{\theta}' - \boldsymbol{\theta}\|_2 \le \delta_0$.

Apply the global-separation argument with $\delta = \delta_0$. For every $\epsilon < \eta_{\delta_0}$, any $\boldsymbol{\theta}^\star$ satisfying

$$\big\|G_C\big(\mathbf{u}_C(\boldsymbol{\theta}), \boldsymbol{\theta}^\star\big) - \mathbf{y}_C\big\|_2 \le \epsilon$$

lies within $\delta_0$ of $\boldsymbol{\theta}$, so the local inequality applies to it. Since $G_C\big(\mathbf{u}_C(\boldsymbol{\theta}), \boldsymbol{\theta}\big) = \mathbf{y}_C$,

$$\frac{c_0}{2}\|\boldsymbol{\theta}^\star - \boldsymbol{\theta}\|_2 \le \big\|G_C\big(\mathbf{u}_C(\boldsymbol{\theta}), \boldsymbol{\theta}^\star\big) - \mathbf{y}_C\big\|_2 \le \epsilon.$$

Therefore

$$\|\boldsymbol{\theta}^\star - \boldsymbol{\theta}\|_2 \le C_\Theta\,\epsilon, \qquad C_\Theta := \frac{2}{c_0}, \tag{24}$$

for all sufficiently small $\epsilon$, uniformly in $\boldsymbol{\theta} \in \Theta$.

Under the approximate-minimizer convention (21), the effective continuous radius is at most $\sqrt{\epsilon^2 + a_\epsilon} = \epsilon\{1 + o(1)\}$, because $a_\epsilon = o(\epsilon^2)$. Hence, for all sufficiently small $\epsilon$, the same conclusion holds with any fixed constant larger than $C_\Theta$. Relabeling that constant as $C_\Theta$, we retain (24).

For each fixed $\boldsymbol{\theta} \in \Theta$, define the discrete boundary layer

$$\mathcal{N}_\epsilon(\boldsymbol{\theta}) := \bigcup_{k \in \mathcal{K}_D} \left[\{\mathbf{u}_D : |u_k - a_k(\boldsymbol{\theta})| \le BC_\Theta\epsilon\} \cup \{\mathbf{u}_D : |u_k - b_k(\boldsymbol{\theta})| \le BC_\Theta\epsilon\}\right]. \tag{25}$$

Each of the two strips associated with a given $k$ has Lebesgue measure at most $2BC_\Theta\epsilon$ in $[0,1]^{n_D}$, so uniformly in $\boldsymbol{\theta}$,

$$\mathrm{Leb}\{\mathcal{N}_\epsilon(\boldsymbol{\theta})\} \le 4n_D BC_\Theta\,\epsilon = O(\epsilon). \tag{26}$$

Suppose $\mathbf{u}_D \notin \mathcal{N}_\epsilon(\boldsymbol{\theta})$ and let $\boldsymbol{\theta}'$ be any parameter value compatible with the continuous residual constraint. By (24), $\|\boldsymbol{\theta}' - \boldsymbol{\theta}\|_2 \le C_\Theta\epsilon$, so (23) bounds the movement of every endpoint by

$$\max_{k \in \mathcal{K}_D} \max\big\{|a_k(\boldsymbol{\theta}') - a_k(\boldsymbol{\theta})|,\, |b_k(\boldsymbol{\theta}') - b_k(\boldsymbol{\theta})|\big\} \le B\,\|\boldsymbol{\theta}' - \boldsymbol{\theta}\|_2 \le BC_\Theta\epsilon.$$

By the definition of the boundary layer in (25), $\mathbf{u}_D \notin \mathcal{N}_\epsilon(\boldsymbol{\theta})$ means that every $u_k$ lies at distance strictly greater than $BC_\Theta\epsilon$ from both $a_k(\boldsymbol{\theta})$ and $b_k(\boldsymbol{\theta})$. No endpoint can therefore cross its auxiliary value as the parameter moves from $\boldsymbol{\theta}$ to $\boldsymbol{\theta}'$, so each of the inequalities

$a_k(\cdot) < u_k \le b_k(\cdot)$ characterizing feasibility in (22) holds at $\boldsymbol{\theta}'$ if and only if it holds at $\boldsymbol{\theta}$. Hence

$$\mathbf{1}\left\{G_D(\mathbf{u}_D, \boldsymbol{\theta}') = \mathbf{y}_D\right\} = \mathbf{1}\left\{G_D(\mathbf{u}_D, \boldsymbol{\theta}) = \mathbf{y}_D\right\}. \tag{27}$$

Away from the boundary layer, the discrete feasibility indicator is therefore constant over every parameter displacement compatible with the continuous residual constraint.

It remains to show that the boundary layer is negligible in the local calculation at a fixed $\boldsymbol{\theta}$. Once (24) restricts the relevant parameter values to an $O(\epsilon)$ neighborhood of $\boldsymbol{\theta}$, the continuous shrinking-neighborhood contribution depends only on $\mathbf{u}_C$. Away from $\mathcal{N}_\epsilon(\boldsymbol{\theta})$, the discrete feasibility indicator is constant throughout that neighborhood by (27). Since $\mathbf{U}_C$ and $\mathbf{U}_D$ are independent, the local continuous factor is the same whether $\mathbf{u}_D$ belongs to the boundary layer or its complement. The relative contribution of the boundary layer is therefore bounded by

$$\frac{\operatorname{Leb}\{\mathcal{N}_\epsilon(\boldsymbol{\theta})\}}{\displaystyle\int_{[0,1]^{n_D}} \mathbf{1}\{G_D(\mathbf{u}_D, \boldsymbol{\theta}) = \mathbf{y}_D\}\, d\mathbf{u}_D} \le \frac{4 n_D B C_\Theta\, \epsilon}{L_D(\mathbf{y}_D \mid \boldsymbol{\theta})} = O(\epsilon), \tag{28}$$

uniformly in $\boldsymbol{\theta} \in \Theta$. Here $L_D(\mathbf{y}_D \mid \boldsymbol{\theta})$ is a finite product of NB masses that is continuous and strictly positive on the compact parameter space, and hence is bounded away from zero. Thus the boundary layer makes no contribution to the limiting local conditional law.

Combining (27) with the continuous contribution of Part 3, the limiting unnormalized joint density in $(\boldsymbol{\theta}, \mathbf{u}_D)$ is

$$L_C(\mathbf{y}_C \mid \boldsymbol{\theta}) \sqrt{\det[J_C(\boldsymbol{\theta})^\top J_C(\boldsymbol{\theta})]}\, \mathbf{1}\{G_D(\mathbf{u}_D, \boldsymbol{\theta}) = \mathbf{y}_D\}. \tag{29}$$

Integrating over the discrete auxiliary variables and using the identity of Part 2,

$$\int_{[0,1]^{n_D}} \mathbf{1}\{G_D(\mathbf{u}_D, \boldsymbol{\theta}) = \mathbf{y}_D\}\, d\mathbf{u}_D = L_D(\mathbf{y}_D \mid \boldsymbol{\theta}),$$

so that

$$\begin{aligned} r(\boldsymbol{\theta} \mid \mathbf{y}) &\propto L_D(\mathbf{y}_D \mid \boldsymbol{\theta})\, L_C(\mathbf{y}_C \mid \boldsymbol{\theta}) \sqrt{\det[J_C(\boldsymbol{\theta})^\top J_C(\boldsymbol{\theta})]} \\ &= L(\mathbf{y} \mid \boldsymbol{\theta}) \sqrt{\det[J_C(\boldsymbol{\theta})^\top J_C(\boldsymbol{\theta})]}. \end{aligned}$$

Applying the normalization convention of Part 3, under which the factor $n^{-q/2}$ is absorbed into the proportionality constant,

$$r(\boldsymbol{\theta} \mid \mathbf{y}) \propto L(\mathbf{y} \mid \boldsymbol{\theta})\, J(\boldsymbol{\theta}, \mathbf{y}_C), \qquad J(\boldsymbol{\theta}, \mathbf{y}_C) = \sqrt{\det[n^{-1} J_C(\boldsymbol{\theta})^\top J_C(\boldsymbol{\theta})]}. \tag{30}$$

It remains to verify that the right-hand side is normalizable. By Condition (D), $\Theta$ is a compact subset of the valid parameter space, has positive Lebesgue measure, and is bounded away from the boundary $\nu_j = 0$. For the fixed observed sample, $L(\mathbf{y} \mid \boldsymbol{\theta})$ is continuous and strictly positive on $\Theta$, and $J(\boldsymbol{\theta}, \mathbf{y}_C)$ is continuous on $\Theta$. The evaluated full-column-rank condition gives $J(\boldsymbol{\theta}, \mathbf{y}_C) > 0$ for every $\boldsymbol{\theta} \in \Theta$. Consequently the integrand is continuous and strictly positive on the compact set $\Theta$, and hence

$$0 < \int_\Theta L(\mathbf{y} \mid \boldsymbol{\theta})\, J(\boldsymbol{\theta}, \mathbf{y}_C)\, d\boldsymbol{\theta} < \infty.$$

Therefore, after normalization, $r(\boldsymbol{\theta} \mid \mathbf{y})$ is a proper probability density on $\Theta$. □

**4. Proof of Proposition 3 (Block-Diagonal Jacobian).**

PROOF. For $j \in C_i$, differentiating the structural equation with respect to $\boldsymbol{\theta}$ holds the auxiliary variable $U_{ij}^C$, and hence $\xi_{ij}$, fixed. The structural equation for $Y_{ij}$ depends on $\boldsymbol{\theta}$ only through the component-specific subvector $\boldsymbol{\theta}_j$. Hence $\frac{\partial Y_{ij}}{\partial \boldsymbol{\theta}_k} = \mathbf{0}$ for $k \neq j$. After ordering the continuous contributions by component, $J_C(\boldsymbol{\theta})$ is therefore block diagonal.

It remains to compute the partial derivatives within block $J_j$. Solving the structural equation $Y_{ij} = \mu_{ij} + \sigma_{ij}\xi_{ij}$ for $\xi_{ij}$ gives the identity

$$\sigma_{ij}\xi_{ij} = Y_{ij} - \mu_{ij}. \tag{31}$$

*Partial derivative with respect to $\beta_{jk}$.* Differentiating $Y_{ij} = \mu_{ij} + \sigma_{ij}\xi_{ij}$ with respect to the $k$-th regression coefficient yields

$$\begin{aligned}
\frac{\partial Y_{ij}}{\partial \beta_{jk}} &= \frac{X_{ik}\exp(X_i^\top \beta_j)t_i}{\nu_j} + \frac{X_{ik}\exp(\frac{1}{2}X_i^\top \beta_j)\sqrt{t_i(\nu_j+1)}}{2\nu_j}\xi_{ij} \\
&= X_{ik}\big(\mu_{ij} + \tfrac{1}{2}\sigma_{ij}\xi_{ij}\big) \\
&= X_{ik}\big(\mu_{ij} + \tfrac{1}{2}(Y_{ij} - \mu_{ij})\big) \quad \text{by (31)} \\
&= \tfrac{1}{2}X_{ik}\big(Y_{ij} + \mu_{ij}\big).
\end{aligned}$$

The scalar factor $\frac{1}{2}(Y_{ij} + \mu_{ij})$ is $w_{ij}^{(\beta)}$ of Proposition 3, and stacking across $k = 1, \dots, K$ gives the row $w_{ij}^{(\beta)} X_i^\top$ of $W_j^{(\beta)} X^{(j)}$.

*Partial derivative with respect to $\nu_j$.* Differentiating with respect to the dispersion parameter,

$$\begin{aligned}
\frac{\partial Y_{ij}}{\partial \nu_j} &= -\frac{\exp(X_i^\top \beta_j)t_i}{\nu_j^2} + \exp\Big(\tfrac{1}{2}X_i^\top \beta_j\Big)\sqrt{t_i} \cdot \frac{-(\nu_j+2)}{2\nu_j^2\sqrt{\nu_j+1}}\xi_{ij} \\
&= -\frac{\mu_{ij}}{\nu_j} - \frac{(\nu_j+2)}{2\nu_j(\nu_j+1)}\sigma_{ij}\xi_{ij} \\
&= -\frac{\mu_{ij}}{\nu_j} - \frac{(\nu_j+2)}{2\nu_j(\nu_j+1)}(Y_{ij} - \mu_{ij}) \quad \text{by (31)} \\
&= -\frac{1}{2\nu_j(\nu_j+1)}\big(\nu_j\mu_{ij} + (\nu_j+2)Y_{ij}\big),
\end{aligned}$$

which is $w_{ij}^{(\nu)}$ of Proposition 3. □

**5. Proof of Lemma 4 (Positive Definiteness of $V_j$).**

PROOF. Let $\mathcal{C} := \{j : P_Q(j \in C_i) > 0\}$ denote the set of active components under the convention of Section 3. Fix $j \in \mathcal{C}$ and define the oracle continuous-selection indicator

$$I_{ij}^C := \mathbf{1}\{j \in C_i\} = \mathbf{1}\Big\{\exp(X_i^\top \beta_{j,0})\, t_i \geq \omega\Big\}.$$

Because the oracle partition is determined at $\boldsymbol{\theta}_0$, $I_{ij}^C$ is a measurable function of the joint design $Z_i = (X_i, t_i)$ and is held fixed when $\boldsymbol{\theta}$ varies.

For a continuous contribution, define

$$g_{ij}(\boldsymbol{\theta}) = \begin{pmatrix} w_{ij}^{(\beta)}(\boldsymbol{\theta})X_i \\ w_{ij}^{(\nu)}(\boldsymbol{\theta}) \end{pmatrix},$$

where

$$w_{ij}^{(\beta)}(\boldsymbol{\theta}) = \frac{1}{2}\{Y_{ij} + \mu_{ij}(\boldsymbol{\theta})\} \quad \text{and} \quad w_{ij}^{(\nu)}(\boldsymbol{\theta}) = -\frac{\nu_j \mu_{ij}(\boldsymbol{\theta}) + (\nu_j + 2)Y_{ij}}{2\nu_j(\nu_j + 1)}.$$

The population counterpart of the normalized Jacobian Gram matrix is

$$V_j(\boldsymbol{\theta}) = \mathbb{E}_{\boldsymbol{\theta}_0}\left[I_{ij}^C g_{ij}(\boldsymbol{\theta}) g_{ij}(\boldsymbol{\theta})^\top\right].$$

We first prove that $V_j(\boldsymbol{\theta}) \succ 0$ for every fixed $\boldsymbol{\theta} \in \Theta$ and then establish a uniform lower bound over $\Theta$.

**Conditional moments.** Write

$$\mu_{ij,0} = \frac{\exp(X_i^\top \beta_{j,0}) t_i}{\nu_{j,0}}, \qquad \sigma_{ij,0}^2 = \frac{\exp(X_i^\top \beta_{j,0}) t_i (1 + \nu_{j,0})}{\nu_{j,0}^2},$$

and, for a fixed candidate parameter $\boldsymbol{\theta} \in \Theta$, write $\mu_{ij} = \mu_{ij}(\boldsymbol{\theta})$. On the event $\{I_{ij}^C = 1\}$, the true working model gives

$$Y_{ij} \mid Z_i \sim \mathcal{N}(\mu_{ij,0}, \sigma_{ij,0}^2).$$

Consequently,

$$\mathbb{E}_{\boldsymbol{\theta}_0}\left[w_{ij}^{(\beta)}(\boldsymbol{\theta}) \,\middle|\, Z_i\right] = \frac{\mu_{ij,0} + \mu_{ij}}{2}$$

and

$$\mathbb{E}_{\boldsymbol{\theta}_0}\left[\left\{w_{ij}^{(\beta)}(\boldsymbol{\theta})\right\}^2 \,\middle|\, Z_i\right] = \frac{(\mu_{ij,0} + \mu_{ij})^2 + \sigma_{ij,0}^2}{4}. \tag{32}$$

**Block decomposition.** For fixed $\boldsymbol{\theta} \in \Theta$, write

$$V_j(\boldsymbol{\theta}) = \begin{pmatrix} A_j(\boldsymbol{\theta}) & \mathbf{b}_j(\boldsymbol{\theta}) \\ \mathbf{b}_j(\boldsymbol{\theta})^\top & c_j(\boldsymbol{\theta}) \end{pmatrix}.$$

The upper-left block is

$$A_j(\boldsymbol{\theta}) = \mathbb{E}_{\boldsymbol{\theta}_0}\left[I_{ij}^C \left\{w_{ij}^{(\beta)}(\boldsymbol{\theta})\right\}^2 X_i X_i^\top\right].$$

By Conditions (A), (B), and (D), the quantities $\mu_{ij,0}$, $\mu_{ij}(\boldsymbol{\theta})$ and $\sigma_{ij,0}^2$ are bounded above and below by positive constants, uniformly over the design support and $\boldsymbol{\theta} \in \Theta$. In particular, there exists $a_j > 0$ such that

$$\mathbb{E}_{\boldsymbol{\theta}_0}\left[\left\{w_{ij}^{(\beta)}(\boldsymbol{\theta})\right\}^2 \,\middle|\, Z_i\right] \geq a_j$$

uniformly over $\boldsymbol{\theta} \in \Theta$ on $\{I_{ij}^C = 1\}$. Therefore, for every $v \in \mathbb{R}^K \setminus \{\mathbf{0}\}$,

$$\begin{aligned} v^\top A_j(\boldsymbol{\theta}) v &= \mathbb{E}_{\boldsymbol{\theta}_0}\left[I_{ij}^C \left\{w_{ij}^{(\beta)}(\boldsymbol{\theta})\right\}^2 (v^\top X_i)^2\right] \\ &\geq a_j \, v^\top \mathbb{E}\left[I_{ij}^C X_i X_i^\top\right] v > 0, \end{aligned}$$

where the final inequality follows from Condition (C)(ii). Hence $A_j(\boldsymbol{\theta}) \succ 0$ for every $\boldsymbol{\theta} \in \Theta$. Condition (C)(ii) also implies

$$\pi_j^C := P(I_{ij}^C = 1) > 0.$$

**Schur complement.** Define

$$U_i(\boldsymbol{\theta}) := I_{ij}^C w_{ij}^{(\beta)}(\boldsymbol{\theta}) X_i, \qquad W_i(\boldsymbol{\theta}) := I_{ij}^C w_{ij}^{(\nu)}(\boldsymbol{\theta}),$$

so that $\mathbb{E}[U_i(\boldsymbol{\theta})U_i(\boldsymbol{\theta})^\top] = A_j(\boldsymbol{\theta})$, $\mathbb{E}[U_i(\boldsymbol{\theta})W_i(\boldsymbol{\theta})] = \mathbf{b}_j(\boldsymbol{\theta})$, and $\mathbb{E}[W_i(\boldsymbol{\theta})^2] = c_j(\boldsymbol{\theta})$. Since $A_j(\boldsymbol{\theta}) \succ 0$, its Schur complement satisfies

(33)

$$S_j(\boldsymbol{\theta}) := c_j(\boldsymbol{\theta}) - \mathbf{b}_j(\boldsymbol{\theta})^\top A_j(\boldsymbol{\theta})^{-1}\mathbf{b}_j(\boldsymbol{\theta}) = \min_{\gamma\in\mathbb{R}^K} \mathbb{E}_{\boldsymbol{\theta}_0}\left[ I_{ij}^C \left\{ w_{ij}^{(\nu)}(\boldsymbol{\theta}) - \gamma^\top w_{ij}^{(\beta)}(\boldsymbol{\theta}) X_i \right\}^2 \right].$$

For every $\boldsymbol{\theta} \in \Theta$, the Jacobian weights satisfy

$$w_{ij}^{(\nu)}(\boldsymbol{\theta}) = -\frac{\nu_j + 2}{\nu_j(\nu_j+1)} w_{ij}^{(\beta)}(\boldsymbol{\theta}) + \frac{\mu_{ij}(\boldsymbol{\theta})}{\nu_j(\nu_j+1)}, \tag{34}$$

as substituting $w_{ij}^{(\beta)}(\boldsymbol{\theta}) = \{Y_{ij} + \mu_{ij}(\boldsymbol{\theta})\}/2$ into the right-hand side confirms.

For fixed $\gamma \in \mathbb{R}^K$, define

$$R_i(\gamma, \boldsymbol{\theta}) := w_{ij}^{(\nu)}(\boldsymbol{\theta}) - \gamma^\top w_{ij}^{(\beta)}(\boldsymbol{\theta}) X_i.$$

By (34),

$$R_i(\gamma, \boldsymbol{\theta}) = \delta_i(\gamma, \boldsymbol{\theta}) w_{ij}^{(\beta)}(\boldsymbol{\theta}) + d_i(\boldsymbol{\theta}), \tag{35}$$

where

$$\delta_i(\gamma, \boldsymbol{\theta}) = -\frac{\nu_j + 2}{\nu_j(\nu_j+1)} - \gamma^\top X_i \qquad \text{and} \qquad d_i(\boldsymbol{\theta}) = \frac{\mu_{ij}(\boldsymbol{\theta})}{\nu_j(\nu_j+1)} > 0.$$

Conditional on $Z_i$, both $\delta_i(\gamma, \boldsymbol{\theta})$ and $d_i(\boldsymbol{\theta})$ are deterministic. On $\{I_{ij}^C = 1\}$,

$$\mathrm{Var}_{\boldsymbol{\theta}_0}\left[ w_{ij}^{(\beta)}(\boldsymbol{\theta}) \,\middle|\, Z_i \right] = \frac{\sigma_{ij,0}^2}{4}, \qquad \mathbb{E}_{\boldsymbol{\theta}_0}\left[ w_{ij}^{(\beta)}(\boldsymbol{\theta}) \,\middle|\, Z_i \right] = \frac{\mu_{ij,0} + \mu_{ij}(\boldsymbol{\theta})}{2}.$$

Therefore,

$$\begin{aligned}\mathbb{E}_{\boldsymbol{\theta}_0}\left[ R_i(\gamma, \boldsymbol{\theta})^2 \,\middle|\, Z_i \right] &= \delta_i(\gamma, \boldsymbol{\theta})^2\, \mathbb{E}_{\boldsymbol{\theta}_0}\left[ \left\{ w_{ij}^{(\beta)}(\boldsymbol{\theta}) \right\}^2 \,\middle|\, Z_i \right] \\ &\quad + 2\delta_i(\gamma, \boldsymbol{\theta}) d_i(\boldsymbol{\theta})\, \mathbb{E}_{\boldsymbol{\theta}_0}\left[ w_{ij}^{(\beta)}(\boldsymbol{\theta}) \,\middle|\, Z_i \right] + d_i(\boldsymbol{\theta})^2.\end{aligned}$$

As a quadratic function of $\delta_i(\gamma, \boldsymbol{\theta})$, its minimum over all real values of $\delta$ is

$$m_{ij}(\boldsymbol{\theta}) = \frac{d_i(\boldsymbol{\theta})^2 \sigma_{ij,0}^2}{\{\mu_{ij,0} + \mu_{ij}(\boldsymbol{\theta})\}^2 + \sigma_{ij,0}^2} > 0. \tag{36}$$

Consequently, for every $\gamma \in \mathbb{R}^K$,

$$\mathbb{E}_{\boldsymbol{\theta}_0}\left[ R_i(\gamma, \boldsymbol{\theta})^2 \,\middle|\, Z_i \right] \geq m_{ij}(\boldsymbol{\theta}) \qquad \text{on } \{I_{ij}^C = 1\}. \tag{37}$$

Conditions (A), (B), and (D) imply that $\mu_{ij,0}$, $\mu_{ij}(\boldsymbol{\theta})$, $\sigma_{ij,0}^2$, $\nu_j$ and $\nu_j^{-1}$ are bounded above and below by positive constants, uniformly over the design support and $\boldsymbol{\theta} \in \Theta$. Therefore the expression in (36) is bounded uniformly away from zero: there exists $\underline{m}_j > 0$ such that

$$\inf_{\boldsymbol{\theta}\in\Theta} m_{ij}(\boldsymbol{\theta}) \geq \underline{m}_j > 0 \qquad \text{almost surely.} \tag{38}$$

Using the law of total expectation (Casella and Berger, 2002, Theorem 4.4.3) in (33), for every $\gamma \in \mathbb{R}^K$ and $\boldsymbol{\theta} \in \Theta$,

$$\mathbb{E}_{\boldsymbol{\theta}_0}\left[I_{ij}^C R_i(\gamma,\boldsymbol{\theta})^2\right] = \mathbb{E}\left[I_{ij}^C\,\mathbb{E}_{\boldsymbol{\theta}_0}\left[R_i(\gamma,\boldsymbol{\theta})^2 \,\middle|\, Z_i\right]\right] \geq \underline{m}_j\, P(I_{ij}^C = 1) = \underline{m}_j \pi_j^C > 0.$$

Taking the minimum over $\gamma$ gives

$$S_j(\boldsymbol{\theta}) \geq \underline{m}_j \pi_j^C > 0 \qquad \text{for every } \boldsymbol{\theta} \in \Theta.$$

Since $A_j(\boldsymbol{\theta}) \succ 0$ and $S_j(\boldsymbol{\theta}) > 0$, the Schur-complement criterion gives $V_j(\boldsymbol{\theta}) \succ 0$ for every $\boldsymbol{\theta} \in \Theta$.

**Uniform positive definiteness.** It remains to establish a uniform eigenvalue bound. For almost every $(Y_{ij}, X_i, t_i)$, the map

$$\boldsymbol{\theta} \longmapsto I_{ij}^C g_{ij}(\boldsymbol{\theta}) g_{ij}(\boldsymbol{\theta})^\top$$

is continuous on $\Theta$. Conditions (A), (B), and (D) imply that the Jacobian weights are affine in $Y_{ij}$ with coefficients uniformly bounded over $\Theta$, so for some finite constant $C_g$,

$$\sup_{\boldsymbol{\theta}\in\Theta} \left\| I_{ij}^C g_{ij}(\boldsymbol{\theta}) g_{ij}(\boldsymbol{\theta})^\top \right\|_F \leq C_g(1 + Y_{ij}^2).$$

On $\{I_{ij}^C = 1\}$, $Y_{ij}$ follows the normal component under $\boldsymbol{\theta}_0$ and has a finite second moment, so the right-hand side is integrable. The dominated convergence theorem (Keener, 2010, Theorem 2.5) therefore implies that $\boldsymbol{\theta} \mapsto V_j(\boldsymbol{\theta})$ is continuous on $\Theta$.

Since the smallest eigenvalue satisfies $\lambda_{\min}\{V_j(\boldsymbol{\theta})\} = \min_{\|v\|=1} v^\top V_j(\boldsymbol{\theta}) v$ (Harville, 1997, Theorem 21.5.6), it is a continuous function of the entries of a symmetric matrix, so

$$\boldsymbol{\theta} \longmapsto \lambda_{\min}\{V_j(\boldsymbol{\theta})\}$$

is continuous on $\Theta$. We have shown it is strictly positive at every $\boldsymbol{\theta} \in \Theta$, and compactness of $\Theta$ therefore gives

$$\inf_{\boldsymbol{\theta}\in\Theta} \lambda_{\min}\{V_j(\boldsymbol{\theta})\} > 0.$$

Because the set $\mathcal{C}$ of active components is finite, this bound may also be taken simultaneously over all such components:

$$\min_{j\in\mathcal{C}} \inf_{\boldsymbol{\theta}\in\Theta} \lambda_{\min}\{V_j(\boldsymbol{\theta})\} > 0.$$

□

**6. Proof of Theorem 5 (Bernstein–von Mises Theorem).** We verify the five regularity conditions of Borgert and Hannig (2026) for the hybrid NB–normal model under Conditions (A)–(E) of Section 4.1. Each assumption is stated in its general form before its verification. The conclusion of Theorem 5 then follows by direct application of Theorem 3.1 of Borgert and Hannig (2026).

Throughout this section, $P_{\boldsymbol{\theta}}$ denotes the joint random-design working model

$$P_{\boldsymbol{\theta}}(dz, d\mathbf{y}) = Q(dz) P_{\boldsymbol{\theta}}(d\mathbf{y} \mid z),$$

where the design distribution $Q$ does not depend on $\boldsymbol{\theta}$. Accordingly, the joint likelihood ratio equals the conditional response likelihood ratio. We therefore use the same notation $P_{\boldsymbol{\theta}}^{(n)}$ for the joint model while writing likelihood contributions conditionally on the observed design. All probabilities, expectations, and modes of convergence below refer to this joint model unless conditioning is explicitly indicated.

REMARK 2. The assumptions of Borgert and Hannig (2026) are stated for densities with respect to Lebesgue measure. For the hybrid model, the response coordinates are normal on $C_i$ and NB on $D_i$, and the oracle partition is determined by the design point $Z_i$. The natural dominating measure for $\mathbf{Y}_i$ is therefore the kernel $\Lambda_z$, which uses Lebesgue measure on the coordinates $j \in C_i$ and counting measure on the coordinates $j \in D_i$, with the assignment determined by $z$. The dominating measure for one joint observation $W_i = (Z_i, \mathbf{Y}_i)$ is

$$\Lambda_W(dz, d\mathbf{y}) := Q(dz)\,\Lambda_z(d\mathbf{y}),$$

and the corresponding full-sample dominating measure is $\Lambda_W^{\otimes n}$. Because the oracle partition is fixed at $\boldsymbol{\theta}_0$, the assignment map, and hence $\Lambda_z$, does not depend on $\boldsymbol{\theta}$. The relevant arguments of Borgert and Hannig (2026) use only $\sigma$-finiteness and $\boldsymbol{\theta}$-independence of the data-dominating measure, and therefore apply with $\Lambda_W^{\otimes n}$ in place of Lebesgue measure.

More generally, any finite or $\sigma$-finite, $\boldsymbol{\theta}$-independent measure equivalent to $\Lambda_W$ may serve as the common dominating measure for one observation, with its $n$-fold product serving as the full-sample dominating measure. Changing between two such measures multiplies every model density by the same observation-dependent, parameter-independent factor and therefore changes neither likelihood ratios nor the generalized fiducial density up to proportionality. We use this freedom in the verification of Assumption 4.

Unless stated otherwise, all data densities below are understood with respect to $\Lambda_W^{\otimes n}$. The dominating measure for $\boldsymbol{\theta}$ in Assumptions 2 and 4 is Lebesgue measure on $\mathbb{R}^q$, as stated.

6.1. *Verification of Assumption 1 (Local asymptotic normality).*

ASSUMPTION 1 (Local asymptotic normality). The sequence of models $\{P_{\boldsymbol{\theta}}^{(n)} : \boldsymbol{\theta} \in \Theta\}$ is locally asymptotically normal around $\boldsymbol{\theta}_0$ with norming matrices $\mathbf{r}_n = \sqrt{n}\,\mathbf{I}_q$, where $q = d(K+1)$ is the total parameter dimension.

PROOF. We show that for the sequence of models $\{P_{\boldsymbol{\theta}}^{(n)} : \boldsymbol{\theta} \in \Theta\}$, there exists a positive definite matrix $\mathcal{I}_{\boldsymbol{\theta}_0}$ and random vectors $\Delta_{n,\boldsymbol{\theta}_0}$ converging weakly to $\mathcal{N}(\mathbf{0}, \mathcal{I}_{\boldsymbol{\theta}_0})$ such that for every converging sequence $\mathbf{h}_n \to \mathbf{h}$,

$$\log \frac{dP^{(n)}_{\boldsymbol{\theta}_0 + \mathbf{h}_n/\sqrt{n}}}{dP^{(n)}_{\boldsymbol{\theta}_0}} = \mathbf{h}_n^\top \Delta_{n,\boldsymbol{\theta}_0} - \frac{1}{2}\mathbf{h}^\top \mathcal{I}_{\boldsymbol{\theta}_0}\mathbf{h} + o_{P^{(n)}_{\boldsymbol{\theta}_0}}(1), \quad n \to \infty.$$

*Setup.* Let $W_i = (Z_i, \mathbf{Y}_i)$, where $Z_i = (X_i, t_i)$. Under Condition (A), the design pairs $Z_i$ are i.i.d. with distribution $Q$. Conditional on the design sequence and oracle partition, the response vectors $\mathbf{Y}_1, \ldots, \mathbf{Y}_n$ are independent across observations. Within each observation, the component responses are also conditionally independent under the working model. Therefore, the joint observations $W_1, \ldots, W_n$ are i.i.d. under $P_{\boldsymbol{\theta}}$.

Because the design distribution $Q$ does not depend on $\boldsymbol{\theta}$, its density cancels from the likelihood ratio. Consequently, $P^{(n)}_{\boldsymbol{\theta}_0}$-almost surely,

$$\frac{dP^{(n)}_{\boldsymbol{\theta}}}{dP^{(n)}_{\boldsymbol{\theta}_0}}(\mathbf{Z}, \mathbf{Y}) = \frac{dP^{(n)}_{\boldsymbol{\theta}}(\,\cdot \mid \mathbf{Z})}{dP^{(n)}_{\boldsymbol{\theta}_0}(\,\cdot \mid \mathbf{Z})}(\mathbf{Y}),$$

and therefore

$$\log \frac{dP^{(n)}_{\boldsymbol{\theta}}}{dP^{(n)}_{\boldsymbol{\theta}_0}}(\mathbf{Z}, \mathbf{Y}) = \log \frac{dP^{(n)}_{\boldsymbol{\theta}}(\,\cdot \mid \mathbf{Z})}{dP^{(n)}_{\boldsymbol{\theta}_0}(\,\cdot \mid \mathbf{Z})}(\mathbf{Y}).$$

Accordingly, the LAN expansion may be established using the conditional working log-likelihood, while probabilities, expectations, and stochastic convergence are taken under the joint random-design model.

The oracle partition is determined at $\boldsymbol{\theta}_0$ and is held fixed under the local perturbation $\boldsymbol{\theta}_0 + \mathbf{h}_n/\sqrt{n}$. In particular, the partition is not redetermined as $\boldsymbol{\theta}$ varies in the Taylor expansion.

Let

$$\ell_i(\boldsymbol{\theta}; \mathbf{Y}_i, Z_i) = \sum_{j \in C_i} \ell_{ij}^{(C)}(\boldsymbol{\theta}_j; Y_{ij}, Z_i) + \sum_{j \in D_i} \ell_{ij}^{(D)}(\boldsymbol{\theta}_j; Y_{ij}, Z_i)$$

denote the conditional working log-likelihood contribution of observation $i$. The full working log-likelihood is

$$\ell_n(\boldsymbol{\theta}) = \sum_{i=1}^{n} \ell_i(\boldsymbol{\theta}; \mathbf{Y}_i, Z_i).$$

Because $Q$ is parameter-free, the joint log-likelihood ratio is

$$\log \frac{dP^{(n)}_{\boldsymbol{\theta}_0 + \mathbf{h}_n/\sqrt{n}}}{dP^{(n)}_{\boldsymbol{\theta}_0}} = \sum_{i=1}^{n} \left[ \ell_i\left(\boldsymbol{\theta}_0 + \frac{\mathbf{h}_n}{\sqrt{n}}; \mathbf{Y}_i, Z_i\right) - \ell_i(\boldsymbol{\theta}_0; \mathbf{Y}_i, Z_i) \right]. \tag{39}$$

Let $s_i(\boldsymbol{\theta}_0) = \nabla_{\boldsymbol{\theta}} \ell_i(\boldsymbol{\theta}_0; \mathbf{Y}_i, Z_i)$ and $H_i(\boldsymbol{\theta}_0) = \nabla^2_{\boldsymbol{\theta}} \ell_i(\boldsymbol{\theta}_0; \mathbf{Y}_i, Z_i)$ denote the observation-level score and Hessian derived in Section 2.

*Taylor expansion.* Because $\boldsymbol{\theta}_0$ is an interior point of $\Theta$ and $\mathbf{h}_n \to \mathbf{h}$, the perturbed parameter $\boldsymbol{\theta}_0 + \mathbf{h}_n/\sqrt{n}$ lies in a fixed neighborhood of $\boldsymbol{\theta}_0$ for all sufficiently large $n$. Therefore, a second-order Taylor expansion of each observation-level log-likelihood contribution gives:

$$\ell_i\left(\boldsymbol{\theta}_0 + \frac{\mathbf{h}_n}{\sqrt{n}}; \mathbf{Y}_i, Z_i\right) - \ell_i(\boldsymbol{\theta}_0; \mathbf{Y}_i, Z_i) = \frac{1}{\sqrt{n}} \mathbf{h}_n^\top s_i(\boldsymbol{\theta}_0) + \frac{1}{2n} \mathbf{h}_n^\top H_i(\boldsymbol{\theta}_0) \mathbf{h}_n + R_{i,n}, \tag{40}$$

where $R_{i,n}$ is the third order remainder. Summing over $i = 1, \ldots, n$ yields

$$\begin{aligned} \log \frac{dP^{(n)}_{\boldsymbol{\theta}_0 + \mathbf{h}_n/\sqrt{n}}}{dP^{(n)}_{\boldsymbol{\theta}_0}} = {} & \underbrace{\mathbf{h}_n^\top \left[ \frac{1}{\sqrt{n}} \sum_{i=1}^{n} s_i(\boldsymbol{\theta}_0) \right]}_{\text{Term 1}} \\ & + \underbrace{\frac{1}{2} \mathbf{h}_n^\top \left[ \frac{1}{n} \sum_{i=1}^{n} H_i(\boldsymbol{\theta}_0) \right] \mathbf{h}_n}_{\text{Term 2}} \\ & + \underbrace{\sum_{i=1}^{n} R_{i,n}}_{\text{Term 3}} . \end{aligned} \tag{41}$$

We first establish that the expectation defining $\mathcal{I}_{\boldsymbol{\theta}_0}$ is finite, and then analyze each term.

*Finiteness and convergence of the information.* By Conditions (A), (B), and (D), the covariates are bounded, the exposures are bounded above and away from zero, and the parameter space is a compact subset of the valid parameter space. Consequently, there exist constants $0 < r_{\min} < r_{\max} < \infty$ and $0 < \nu_{\min} < \nu_{\max} < \infty$ such that, for every observation $i$ and component $j$, $r_{ij,0} = \exp(X_i^\top \beta_{j,0}) t_i \in [r_{\min}, r_{\max}]$ and $\nu_{j,0} \in [\nu_{\min}, \nu_{\max}]$ almost surely.

For a normal contribution, every entry of $\mathcal{I}_i^{(C)}(\boldsymbol{\theta}_{j,0} \mid Z_i)$ in Equation (14) is bounded because it is a continuous function of $(r_{ij,0}, \nu_{j,0}, X_i)$ over a compact set.

For a NB contribution, the only term requiring additional consideration satisfies

$$0 \leq \psi^{(1)}(r_{ij,0}) - \mathbb{E}_{\boldsymbol{\theta}_0}\left[\psi^{(1)}(Y_{ij} + r_{ij,0}) \mid Z_i\right] \leq \psi^{(1)}(r_{ij,0}) \leq \sup_{r \in [r_{\min}, r_{\max}]} \psi^{(1)}(r) < \infty.$$

All remaining entries of $\mathcal{I}_i^{(D)}(\boldsymbol{\theta}_{j,0} \mid Z_i)$ in Equation (7) are bounded over the same compact set.

Because $d$ and $K$ are fixed, there exists $M_{\mathcal{I}} < \infty$ such that $\|\mathcal{I}_i(\boldsymbol{\theta}_0 \mid Z_i)\|_F \leq M_{\mathcal{I}}$ almost surely. Hence, $\mathbb{E}_Q\left[\|\mathcal{I}_i(\boldsymbol{\theta}_0 \mid Z_i)\|_F\right] < \infty$. Thus, the population information matrix $\mathcal{I}_{\boldsymbol{\theta}_0} = \mathbb{E}_Q\left[\mathcal{I}_i(\boldsymbol{\theta}_0 \mid Z_i)\right]$ is finite. Condition (E) further ensures that $\mathcal{I}_{\boldsymbol{\theta}_0} \succ \mathbf{0}$.

Because the design pairs $Z_i$ are i.i.d. and the oracle partition is a measurable function of $Z_i$ at $\boldsymbol{\theta}_0$, the conditional information matrices $\mathcal{I}_i(\boldsymbol{\theta}_0 \mid Z_i)$ are i.i.d.. Applying the strong law of large numbers (Durrett, 2010, Theorem 2.4.1) entrywise gives

$$\frac{1}{n}\sum_{i=1}^{n} \mathcal{I}_i(\boldsymbol{\theta}_0 \mid Z_i) \xrightarrow{a.s.} \mathbb{E}_Q\left[\mathcal{I}_i(\boldsymbol{\theta}_0 \mid Z_i)\right] = \mathcal{I}_{\boldsymbol{\theta}_0}.$$

*Term 1.* Now, define the central sequence

$$\Delta_{n,\boldsymbol{\theta}_0} := \frac{1}{\sqrt{n}}\sum_{i=1}^{n} s_i(\boldsymbol{\theta}_0). \tag{42}$$

Term 1 in Equation (41) is therefore $\mathbf{h}_n^\top \Delta_{n,\boldsymbol{\theta}_0}$.

By the conditional score identity in Section 2, $\mathbb{E}_{\boldsymbol{\theta}_0}\left[s_i(\boldsymbol{\theta}_0) \mid Z_i\right] = \mathbf{0}$. The law of total expectation therefore gives $\mathbb{E}_{\boldsymbol{\theta}_0}\left[s_i(\boldsymbol{\theta}_0)\right] = \mathbf{0}$. Similarly, the law of total variance and the conditional information identity give

$$\mathrm{Var}_{\boldsymbol{\theta}_0}\left[s_i(\boldsymbol{\theta}_0)\right] = \mathbb{E}_Q\left[\mathrm{Var}_{\boldsymbol{\theta}_0}\left(s_i(\boldsymbol{\theta}_0) \mid Z_i\right)\right] = \mathbb{E}_Q\left[\mathcal{I}_i(\boldsymbol{\theta}_0 \mid Z_i)\right] = \mathcal{I}_{\boldsymbol{\theta}_0}.$$

The central limit theorem requires only finite second moments, and $\mathrm{Var}_{\boldsymbol{\theta}_0}[s_i(\boldsymbol{\theta}_0)] = \mathcal{I}_{\boldsymbol{\theta}_0}$ was shown finite above.

Because the joint observations $(Z_i, \mathbf{Y}_i)$ are i.i.d., the multivariate central limit theorem (van der Vaart, 1998, Proposition 2.27) gives

$$\Delta_{n,\boldsymbol{\theta}_0} \xrightarrow{d} \mathcal{N}\left(\mathbf{0}, \mathcal{I}_{\boldsymbol{\theta}_0}\right). \tag{43}$$

Since $\mathbf{h}_n \to \mathbf{h}$, Slutsky's theorem (van der Vaart, 1998, Lemma 2.8(ii)) gives

$$\mathbf{h}_n^\top \Delta_{n,\boldsymbol{\theta}_0} = \mathbf{h}^\top \Delta_{n,\boldsymbol{\theta}_0} + o_{P^{(n)}_{\boldsymbol{\theta}_0}}(1),$$

and hence

$$\mathbf{h}_n^\top \Delta_{n,\boldsymbol{\theta}_0} \xrightarrow{d} \mathcal{N}\left(0,\, \mathbf{h}^\top \mathcal{I}_{\boldsymbol{\theta}_0}\mathbf{h}\right).$$

*Term 2.* By the conditional information identity established in Section 2,

$$\mathbb{E}_{\boldsymbol{\theta}_0}\left[H_i(\boldsymbol{\theta}_0) \mid Z_i\right] = -\mathcal{I}_i(\boldsymbol{\theta}_0 \mid Z_i).$$

Taking expectations over the joint design distribution $Q$ gives

$$\mathbb{E}_{\boldsymbol{\theta}_0}\left[H_i(\boldsymbol{\theta}_0)\right] = -\mathbb{E}_Q\left[\mathcal{I}_i(\boldsymbol{\theta}_0 \mid Z_i)\right] = -\mathcal{I}_{\boldsymbol{\theta}_0}.$$

Conditions (A), (B), and (D), together with the explicit Hessian expressions in Section 2, imply $\mathbb{E}_{\boldsymbol{\theta}_0}\left[\|H_i(\boldsymbol{\theta}_0)\|_F\right] < \infty$. The normal Hessian entries are polynomials of degree at most two in $Y_{ij}$ with coefficients uniformly bounded over the compact design and parameter sets. The

NB Hessian entries involve terms growing at most linearly or logarithmically in $Y_{ij}$, together with polygamma terms uniformly bounded for $r_{ij,0} \geq r_{\min} > 0$. The required moments are therefore finite.

Because the joint observations $W_i = (Z_i, \mathbf{Y}_i)$ are i.i.d., the weak law of large numbers (van der Vaart, 1998, Proposition 2.16) applied entrywise gives

$$\frac{1}{n}\sum_{i=1}^{n} H_i(\boldsymbol{\theta}_0) \xrightarrow{P^{(n)}_{\boldsymbol{\theta}_0}} \mathbb{E}_{\boldsymbol{\theta}_0}\left[H_i(\boldsymbol{\theta}_0)\right] = -\mathcal{I}_{\boldsymbol{\theta}_0}.$$

Since $\mathbf{h}_n \to \mathbf{h}$, the continuous mapping theorem gives

$$\frac{1}{2}\mathbf{h}_n^\top \left[\frac{1}{n}\sum_{i=1}^{n} H_i(\boldsymbol{\theta}_0)\right]\mathbf{h}_n \xrightarrow{P^{(n)}_{\boldsymbol{\theta}_0}} -\frac{1}{2}\mathbf{h}^\top \mathcal{I}_{\boldsymbol{\theta}_0}\mathbf{h}. \tag{44}$$

Equivalently,

$$\frac{1}{2}\mathbf{h}_n^\top \left[\frac{1}{n}\sum_{i=1}^{n} H_i(\boldsymbol{\theta}_0)\right]\mathbf{h}_n = -\frac{1}{2}\mathbf{h}^\top \mathcal{I}_{\boldsymbol{\theta}_0}\mathbf{h} + o_{P^{(n)}_{\boldsymbol{\theta}_0}}(1). \tag{45}$$

*Term 3.* Because $\boldsymbol{\theta}_0 \in \operatorname{int}(\Theta)$ and $\mathbf{h}_n \to \mathbf{h}$, there exists a compact neighborhood $\mathcal{N} \subset \operatorname{int}(\Theta)$ of $\boldsymbol{\theta}_0$ such that, for all sufficiently large $n$, the line segment joining $\boldsymbol{\theta}_0$ and $\boldsymbol{\theta}_0 + \mathbf{h}_n/\sqrt{n}$ is contained in $\mathcal{N}$.

The oracle partition is held fixed in the Taylor expansion. On $\mathcal{N}$, the normal and NB log-likelihood contributions are three times continuously differentiable with respect to $\boldsymbol{\theta}$. Moreover, Conditions (A), (B), and (D) imply that there exists a nonnegative random variable $M_i$ satisfying

$$\sup_{\boldsymbol{\theta}\in\mathcal{N}} \left\|D^3\ell_i(\boldsymbol{\theta}; \mathbf{Y}_i, Z_i)\right\|_{\mathrm{op}} \leq M_i \quad \text{and} \quad \mathbb{E}_{\boldsymbol{\theta}_0}[M_i] < \infty.$$

For the normal contributions, the third derivatives are bounded by polynomials of degree at most two in the response, with coefficients uniformly bounded over $\mathcal{N}$ and the support of the design. For the NB contributions, the third derivatives involve polygamma functions evaluated at $Y_{ij} + r_{ij}$ and $r_{ij}$, together with terms growing at most linearly or logarithmically in $Y_{ij}$. Since $r_{ij}$ is bounded away from zero and the NB response has moments of every finite order uniformly over the compact parameter and design sets, the required integrable envelope exists.

By the multivariate Taylor theorem (Keener, 2010, Appendix A.5), the third-order remainder satisfies

$$|R_{i,n}| \leq \frac{\|\mathbf{h}_n\|^3}{6n^{3/2}} M_i.$$

Because $\mathbf{h}_n \to \mathbf{h}$, the sequence $\{\|\mathbf{h}_n\|\}$ is bounded. Moreover, because $W_i = (Z_i, \mathbf{Y}_i)$ are i.i.d., the envelope variables $M_i$ are also i.i.d., and the weak law of large numbers (van der Vaart, 1998, Proposition 2.16) gives

$$\frac{1}{n}\sum_{i=1}^{n} M_i \xrightarrow{P^{(n)}_{\boldsymbol{\theta}_0}} \mathbb{E}_{\boldsymbol{\theta}_0}[M_i] < \infty.$$

In particular,

$$\frac{1}{n}\sum_{i=1}^{n} M_i = O_{P^{(n)}_{\boldsymbol{\theta}_0}}(1).$$

Consequently,

$$\left|\sum_{i=1}^{n} R_{i,n}\right| \leq \frac{\|\mathbf{h}_n\|^3}{6n^{3/2}}\sum_{i=1}^{n} M_i = \frac{\|\mathbf{h}_n\|^3}{6\sqrt{n}}\left(\frac{1}{n}\sum_{i=1}^{n} M_i\right) = o_{P^{(n)}_{\boldsymbol{\theta}_0}}(1). \tag{46}$$

*Conclusion.* Combining the Taylor expansion (41), the central-sequence convergence (43), the quadratic-term expansion (45), and the remainder bound (46), we obtain

$$\log \frac{dP^{(n)}_{\boldsymbol{\theta}_0+\mathbf{h}_n/\sqrt{n}}}{dP^{(n)}_{\boldsymbol{\theta}_0}} = \mathbf{h}_n^\top \Delta_{n,\boldsymbol{\theta}_0} - \frac{1}{2}\mathbf{h}^\top \mathcal{I}_{\boldsymbol{\theta}_0}\mathbf{h} + o_{P^{(n)}_{\boldsymbol{\theta}_0}}(1),$$

where

$$\Delta_{n,\boldsymbol{\theta}_0} = \frac{1}{\sqrt{n}}\sum_{i=1}^{n} s_i(\boldsymbol{\theta}_0) \xrightarrow{d} \mathcal{N}\left(\mathbf{0}, \mathcal{I}_{\boldsymbol{\theta}_0}\right).$$

Condition (E) ensures that $\mathcal{I}_{\boldsymbol{\theta}_0} \succ \mathbf{0}$. Therefore, the sequence of joint random-design working models $\{P^{(n)}_{\boldsymbol{\theta}} : \boldsymbol{\theta} \in \Theta\}$ is locally asymptotically normal at $\boldsymbol{\theta}_0$, with norming matrices $\mathbf{r}_n = \sqrt{n}\,\mathbf{I}_q$, central sequence $\Delta_{n,\boldsymbol{\theta}_0}$, and limiting information matrix $\mathcal{I}_{\boldsymbol{\theta}_0}$. This verifies Assumption 1. □

6.2. *Verification of Assumption 2 (Convergence of Jacobian function).*

ASSUMPTION 2 (Convergence of Jacobian function in the generalized fiducial density). There exists a sequence $D_n > 0$ with $D_n\sqrt{n} \to \infty$, and a $\sigma$-finite measure, absolutely continuous with respect to the Lebesgue measure on $\|\boldsymbol{\theta} - \boldsymbol{\theta}_0\| \le D_n$, with density $\pi$ such that as $n \to \infty$,

$$\sup_{\boldsymbol{\theta}:\|\boldsymbol{\theta}-\boldsymbol{\theta}_0\|\le D_n} \frac{|J(\mathbf{y},\boldsymbol{\theta}) - \pi(\boldsymbol{\theta})|}{J(\mathbf{y},\boldsymbol{\theta})} \xrightarrow{P_{\boldsymbol{\theta}_0}} 0.$$

PROOF. Because $\boldsymbol{\theta}_0 \in \operatorname{int}(\Theta)$, choose a fixed $D > 0$ such that

$$\mathcal{B}_D := \{\boldsymbol{\theta} \in \mathbb{R}^q : \|\boldsymbol{\theta} - \boldsymbol{\theta}_0\| \le D\} \subseteq \Theta.$$

We verify the assumption on this fixed ball by taking $D_n \equiv D$. This choice is admissible because $D\sqrt{n} \to \infty$, and convergence on $\mathcal{B}_D$ is stronger than the stated requirement. The oracle partition is determined at $\boldsymbol{\theta}_0$ and held fixed throughout the proof.

For each component $j$, define $I^C_{ij} := \mathbf{1}\{j \in C_i\}$ and

$$g_{ij}(\boldsymbol{\theta}) = \begin{pmatrix} w^{(\beta)}_{ij}(\boldsymbol{\theta})X_i \\ w^{(\nu)}_{ij}(\boldsymbol{\theta}) \end{pmatrix} \in \mathbb{R}^{K+1},$$

where

$$w^{(\beta)}_{ij}(\boldsymbol{\theta}) = \frac{1}{2}\{Y_{ij} + \mu_{ij}(\boldsymbol{\theta})\} \quad \text{and} \quad w^{(\nu)}_{ij}(\boldsymbol{\theta}) = -\frac{\nu_j \mu_{ij}(\boldsymbol{\theta}) + (\nu_j+2)Y_{ij}}{2\nu_j(\nu_j+1)}.$$

Only observations for which $I^C_{ij} = 1$ contribute rows to the continuous Jacobian block $J_j$. Consequently,

$$\widehat{V}_j(\boldsymbol{\theta}) := \frac{1}{n} J_j(\boldsymbol{\theta})^\top J_j(\boldsymbol{\theta}) = \frac{1}{n}\sum_{i=1}^{n} I^C_{ij} g_{ij}(\boldsymbol{\theta}) g_{ij}(\boldsymbol{\theta})^\top. \tag{47}$$

Its population counterpart under the true joint random-design model is

$$V_j(\boldsymbol{\theta}) := \mathbb{E}_{\boldsymbol{\theta}_0}\left[I^C_{ij} g_{ij}(\boldsymbol{\theta}) g_{ij}(\boldsymbol{\theta})^\top\right]. \tag{48}$$

By the block-diagonal structure established in Proposition 3,

$$J(\mathbf{y}, \boldsymbol{\theta}) = \prod_{j=1}^{d} \sqrt{\det \widehat{V}_j(\boldsymbol{\theta})}.$$

We define the candidate limiting density by

$$\pi(\boldsymbol{\theta}) := \prod_{j=1}^{d} \sqrt{\det V_j(\boldsymbol{\theta})}. \tag{49}$$

We verify the required uniform relative convergence in four steps.

**Step 1: Uniform convergence of the Jacobian Gram matrices.**

For every component $j$, the map $\boldsymbol{\theta} \longmapsto I_{ij}^C g_{ij}(\boldsymbol{\theta}) g_{ij}(\boldsymbol{\theta})^\top$ is continuous on $\mathcal{B}_D$ almost surely. The oracle indicator $I_{ij}^C$ is fixed at $\boldsymbol{\theta}_0$ and therefore does not vary with $\boldsymbol{\theta}$.

By Conditions (A), (B), and (D), the covariates and exposures are bounded and the parameter values range over a compact subset of the valid parameter space. Hence, uniformly over $\boldsymbol{\theta} \in \mathcal{B}_D$, the quantities $\mu_{ij}(\boldsymbol{\theta})$, $\nu_j$, and $\nu_j^{-1}$ are bounded. Since the Jacobian weights are affine in $Y_{ij}$ with uniformly bounded coefficients, there exists a constant $C_g < \infty$ such that

$$\sup_{\boldsymbol{\theta} \in \mathcal{B}_D} \left\| I_{ij}^C g_{ij}(\boldsymbol{\theta}) g_{ij}(\boldsymbol{\theta})^\top \right\|_F \le C_g (1 + Y_{ij}^2).$$

On the event $\{I_{ij}^C = 1\}$, the response follows the normal component of the working model at $\boldsymbol{\theta}_0$. Conditions (A), (B), and (D) imply that its mean and variance are uniformly bounded. Therefore, $\mathbb{E}_{\boldsymbol{\theta}_0}\left[I_{ij}^C (1 + Y_{ij}^2)\right] < \infty$. Thus, the matrix-valued function in (47) is continuous in $\boldsymbol{\theta}$ and admits an integrable envelope uniformly over $\mathcal{B}_D$.

Because the joint observations $W_i = (Z_i, \mathbf{Y}_i)$ are i.i.d., and the class $\{I_{ij}^C g_{ij}(\boldsymbol{\theta}) g_{ij}(\boldsymbol{\theta})^\top : \boldsymbol{\theta} \in \mathcal{B}_D\}$ has compact index set, is almost surely continuous in $\boldsymbol{\theta}$, and admits an integrable envelope, it is Glivenko–Cantelli (van der Vaart, 1998, Example 19.8). Applied entrywise, this gives

$$\sup_{\boldsymbol{\theta} \in \mathcal{B}_D} \left\| \widehat{V}_j(\boldsymbol{\theta}) - V_j(\boldsymbol{\theta}) \right\|_F \xrightarrow{a.s.} 0 \qquad \text{for every } j. \tag{50}$$

Continuity of $\boldsymbol{\theta} \longmapsto V_j(\boldsymbol{\theta})$ on $\Theta$, and hence on $\mathcal{B}_D$, was established in the proof of Lemma 4.

**Step 2: Uniform positive definiteness and determinant convergence.**

By Lemma 4, for every active component $j$ there exists $\kappa_j > 0$ such that

$$\inf_{\boldsymbol{\theta} \in \Theta} \lambda_{\min}\{V_j(\boldsymbol{\theta})\} \ge \kappa_j > 0. \tag{51}$$

In particular, because $\mathcal{B}_D \subseteq \Theta$,

$$\inf_{\boldsymbol{\theta} \in \mathcal{B}_D} \lambda_{\min}\{V_j(\boldsymbol{\theta})\} \ge \kappa_j > 0.$$

For symmetric matrices, Weyl's inequality (Franklin, 2000, Chapter 6.7, Theorem 2) gives

$$\left| \lambda_{\min}\left(\widehat{V}_j(\boldsymbol{\theta})\right) - \lambda_{\min}(V_j(\boldsymbol{\theta})) \right| \le \left\| \widehat{V}_j(\boldsymbol{\theta}) - V_j(\boldsymbol{\theta}) \right\|_{\text{op}}.$$

Since the operator norm is bounded by the Frobenius norm, (50) and (51) imply

$$P_{\boldsymbol{\theta}_0}\left\{ \inf_{\boldsymbol{\theta} \in \mathcal{B}_D} \lambda_{\min}\left(\widehat{V}_j(\boldsymbol{\theta})\right) \ge \frac{\kappa_j}{2} \right\} \longrightarrow 1. \tag{52}$$

Thus, with probability approaching one, $\widehat{V}_j(\boldsymbol{\theta})$ is positive definite uniformly over $\mathcal{B}_D$. In particular,

$$\inf_{\boldsymbol{\theta}\in\mathcal{B}_D} \sqrt{\det \widehat{V}_j(\boldsymbol{\theta})} \geq \left(\frac{\kappa_j}{2}\right)^{(K+1)/2} > 0. \tag{53}$$

Because the determinant is a polynomial in the matrix entries, and the matrices $\widehat{V}_j(\boldsymbol{\theta})$ and $V_j(\boldsymbol{\theta})$ are uniformly bounded on $\mathcal{B}_D$ with probability approaching one, (50) gives $\sup_{\boldsymbol{\theta}\in\mathcal{B}_D} |\det \widehat{V}_j(\boldsymbol{\theta}) - \det V_j(\boldsymbol{\theta})| \xrightarrow{P_{\boldsymbol{\theta}_0}} 0$. Moreover, $\det V_j(\boldsymbol{\theta}) \geq \kappa_j^{K+1}$ and, on the event in (52), $\det \widehat{V}_j(\boldsymbol{\theta}) \geq (\kappa_j/2)^{K+1}$, so both determinants are bounded below by $(\kappa_j/2)^{K+1} > 0$, on which range $x \longmapsto \sqrt{x}$ is Lipschitz. Hence

$$\sup_{\boldsymbol{\theta}\in\mathcal{B}_D} \left| \frac{\sqrt{\det V_j(\boldsymbol{\theta})}}{\sqrt{\det \widehat{V}_j(\boldsymbol{\theta})}} - 1 \right| \xrightarrow{P_{\boldsymbol{\theta}_0}} 0. \tag{54}$$

**Step 3: Convergence of the full Jacobian ratio.**

Let $A_n$ denote the event that $\widehat{V}_j(\boldsymbol{\theta})$ is positive definite uniformly over $\mathcal{B}_D$ for every $j$. By (52), $P_{\boldsymbol{\theta}_0}(A_n) \to 1$. Define

$$R_{j,n}(\boldsymbol{\theta}) := \frac{\sqrt{\det V_j(\boldsymbol{\theta})}}{\sqrt{\det \widehat{V}_j(\boldsymbol{\theta})}}.$$

On $A_n$,

$$\frac{|J(\mathbf{y},\boldsymbol{\theta}) - \pi(\boldsymbol{\theta})|}{J(\mathbf{y},\boldsymbol{\theta})} = \left|1 - \frac{\pi(\boldsymbol{\theta})}{J(\mathbf{y},\boldsymbol{\theta})}\right| = \left|1 - \prod_{j=1}^{d} R_{j,n}(\boldsymbol{\theta})\right|.$$

Equation (54) gives

$$\sup_{\boldsymbol{\theta}\in\mathcal{B}_D} |R_{j,n}(\boldsymbol{\theta}) - 1| \xrightarrow{P_{\boldsymbol{\theta}_0}} 0 \qquad \text{for every } j.$$

Since $d$ is fixed, the product of $d$ terms each converging uniformly to 1 also converges uniformly to 1. Combined with $P_{\boldsymbol{\theta}_0}(A_n) \to 1$, this gives

$$\sup_{\boldsymbol{\theta}:\|\boldsymbol{\theta}-\boldsymbol{\theta}_0\|\leq D} \frac{|J(\mathbf{y},\boldsymbol{\theta}) - \pi(\boldsymbol{\theta})|}{J(\mathbf{y},\boldsymbol{\theta})} \xrightarrow{P_{\boldsymbol{\theta}_0}} 0.$$

**Step 4: The limiting density $\pi$ is well defined.**

By Lemma 4, each $V_j(\boldsymbol{\theta})$ is uniformly positive definite on $\Theta$. Together with continuity of $\boldsymbol{\theta} \mapsto V_j(\boldsymbol{\theta})$, this implies that the function $\pi$ in (49) is continuous and strictly positive on $\Theta$, and hence on the compact set $\mathcal{B}_D$. It is therefore bounded and Lebesgue-integrable on $\mathcal{B}_D$, and hence defines a finite, and therefore $\sigma$-finite, measure absolutely continuous with respect to Lebesgue measure. This verifies Assumption 2. □

6.3. *Verification of Assumption 3 (Mass of limiting measure).*

ASSUMPTION 3 (Mass of limiting measure). The limiting density $\pi(\boldsymbol{\theta})$ is continuous and positive at $\boldsymbol{\theta}_0$.

PROOF. Recall from (49) that

$$\pi(\boldsymbol{\theta}) = \prod_{j=1}^{d} \sqrt{\det(V_j(\boldsymbol{\theta}))}.$$

Step 4 of the proof of Assumption 2 established that $\pi$ is continuous and strictly positive on $\mathcal{B}_D$. Since $\boldsymbol{\theta}_0 \in \mathcal{B}_D$, it follows in particular that $\pi$ is continuous at $\boldsymbol{\theta}_0$ and $\pi(\boldsymbol{\theta}_0) > 0$, which verifies Assumption 3. □

6.4. *Verification of Assumption 4 (Likelihood splitting).*

ASSUMPTION 4 (Likelihood splitting). There exist finite measures $P^1_{\boldsymbol{\theta}}, P^2_{\boldsymbol{\theta}}$ absolutely continuous with respect to the Lebesgue measure having densities $p^1(\mathbf{y},\boldsymbol{\theta}), p^2(\mathbf{y},\boldsymbol{\theta})$, respectively, such that

$$p(\mathbf{y},\boldsymbol{\theta}) = p^1(\mathbf{y},\boldsymbol{\theta}) p^2(\mathbf{y},\boldsymbol{\theta}).$$

Moreover, there exists a finite measure on $\boldsymbol{\theta}$, absolutely continuous with respect to the Lebesgue measure, with density $\gamma(\boldsymbol{\theta})$ so that for the same sequence $D_n$ as in Assumption 2, as $n \to \infty$,

$$P_{\boldsymbol{\theta}_0}\left(p^2(\mathbf{y},\boldsymbol{\theta}) J(\mathbf{y},\boldsymbol{\theta}) \mathbf{1}_{\{\|\boldsymbol{\theta}-\boldsymbol{\theta}_0\|>D_n\}} \le \gamma(\boldsymbol{\theta})\right) \to 1.$$

PROOF. We construct both factors on the full sample space under a common, parameter-independent dominating probability measure. Throughout, $\mathbf{w} = (w_1, \ldots, w_n)$ with $w_i = (Z_i, \mathbf{Y}_i)$ collects the design and response, and the densities $p^1$ and $p^2$ of Assumption 4 are densities in $\mathbf{w}$ rather than in $\mathbf{y}$ alone.

Fix an arbitrary reference parameter $\boldsymbol{\theta}^\dagger \in \Theta$ and let $H := P_{\boldsymbol{\theta}^\dagger}$ denote the distribution of one joint observation $W_i = (Z_i, \mathbf{Y}_i)$ under that parameter. Conditional on the design, the normal density is strictly positive on $\mathbb{R}$ and the NB probability mass function is strictly positive on $\mathbb{N}_0$. Because the oracle partition and the design distribution $Q$ do not depend on $\boldsymbol{\theta}$, it follows that

$$P_{\boldsymbol{\theta}} \ll H \qquad \text{and} \qquad H \ll P_{\boldsymbol{\theta}} \qquad \text{for every } \boldsymbol{\theta} \in \Theta.$$

We use $H^{\otimes n}$ as the dominating measure on the full sample space. This is permitted by Remark 2, because $H^{\otimes n}$ is a probability measure, hence $\sigma$-finite, and does not depend on $\boldsymbol{\theta}$. Define the full-data density

$$p_n(\mathbf{w},\boldsymbol{\theta}) := \frac{dP^{(n)}_{\boldsymbol{\theta}}}{dH^{\otimes n}}(\mathbf{w}) = \prod_{i=1}^{n} \frac{dP_{\boldsymbol{\theta}}}{dH}(w_i).$$

Writing

$$h_n(\mathbf{w}) := \frac{dH^{\otimes n}}{d\Lambda_W^{\otimes n}}(\mathbf{w}),$$

which is strictly positive and does not depend on $\boldsymbol{\theta}$, we have

$$p_n(\mathbf{w},\boldsymbol{\theta}) = \frac{dP^{(n)}_{\boldsymbol{\theta}}/d\Lambda_W^{\otimes n}}{h_n(\mathbf{w})}.$$

Thus, changing the dominating measure multiplies the original likelihood by the observation-dependent but parameter-independent factor $h_n(\mathbf{w})^{-1}$, and therefore changes neither likelihood ratios nor the generalized fiducial density up to proportionality.

Set

$$p_n^1(\mathbf{w},\boldsymbol{\theta}) := p_n(\mathbf{w},\boldsymbol{\theta}), \qquad p_n^2(\mathbf{w},\boldsymbol{\theta}) := 1,$$

so that $p_n(\mathbf{w},\boldsymbol{\theta}) = p_n^1(\mathbf{w},\boldsymbol{\theta})p_n^2(\mathbf{w},\boldsymbol{\theta})$. Both factors are densities of finite measures on the same full-data space with respect to the same dominating measure $H^{\otimes n}$: indeed,

$$\int p_n^1(\mathbf{w},\boldsymbol{\theta})\, dH^{\otimes n}(\mathbf{w}) = 1, \qquad \int p_n^2(\mathbf{w},\boldsymbol{\theta})\, dH^{\otimes n}(\mathbf{w}) = 1.$$

The corresponding finite measures are therefore

$$P_{\boldsymbol{\theta}}^{1,(n)} = P_{\boldsymbol{\theta}}^{(n)}, \qquad P_{\boldsymbol{\theta}}^{2,(n)} = H^{\otimes n}.$$

**High-probability bound for the Jacobian.** It remains to verify the domination condition. Recall that

$$J(\mathbf{y},\boldsymbol{\theta}) = \prod_{j=1}^{d} \sqrt{\det \widehat{V}_j(\boldsymbol{\theta})}, \qquad \widehat{V}_j(\boldsymbol{\theta}) = \frac{1}{n}\sum_{i=1}^{n} I_{ij}^C g_{ij}(\boldsymbol{\theta}) g_{ij}(\boldsymbol{\theta})^\top.$$

Because $\|g_{ij}(\boldsymbol{\theta})g_{ij}(\boldsymbol{\theta})^\top\|_F = \|g_{ij}(\boldsymbol{\theta})\|^2$, and Conditions (A), (B), and (D) uniformly bound the design, exposure, component means, dispersions, and inverse dispersions over $\boldsymbol{\theta} \in \Theta$, the constant $C_g$ of Step 1 of the proof of Assumption 2 may be taken uniformly over all of $\Theta$, so that

$$\sup_{\boldsymbol{\theta}\in\Theta} \|g_{ij}(\boldsymbol{\theta})\|^2 \le C_g(1 + Y_{ij}^2).$$

Consequently,

$$\sup_{\boldsymbol{\theta}\in\Theta} \operatorname{tr} \widehat{V}_j(\boldsymbol{\theta}) \le \frac{C_g}{n} \sum_{i=1}^{n} I_{ij}^C (1 + Y_{ij}^2). \tag{55}$$

On the event $\{I_{ij}^C = 1\}$, the true working distribution of $Y_{ij}$ is normal, with mean and variance uniformly bounded under Conditions (A), (B), and (D). Hence

$$L_j := \mathbb{E}_{\boldsymbol{\theta}_0}\left[I_{ij}^C(1 + Y_{ij}^2)\right] < \infty,$$

and $L_j$ does not depend on $i$ because the observations are i.i.d.. The strong law of large numbers (Durrett, 2010, Theorem 2.4.1) gives

$$\frac{1}{n}\sum_{i=1}^{n} I_{ij}^C (1 + Y_{ij}^2) \xrightarrow{a.s.} L_j,$$

so for each of the finitely many components $j$, $\frac{1}{n}\sum_{i=1}^n I_{ij}^C(1+Y_{ij}^2) \le L_j + 1$ for all sufficiently large $n$, almost surely. Setting $M_V := C_g \max_{j\le d}(L_j + 1)$, which is finite because $d$ is fixed, and combining with (55), we obtain

$$P_{\boldsymbol{\theta}_0}^{(n)}\left(\sup_{\boldsymbol{\theta}\in\Theta} \operatorname{tr} \widehat{V}_j(\boldsymbol{\theta}) \le M_V \text{ for every } j \in \{1,\dots,d\}\right) \longrightarrow 1. \tag{56}$$

Each $\widehat{V}_j(\boldsymbol{\theta})$ is a positive semidefinite $(K+1)\times(K+1)$ matrix, so its eigenvalues $\lambda_1,\dots,\lambda_{K+1}$ are nonnegative. The arithmetic–geometric mean inequality (Casella and Berger, 2002, Example 4.7.8) applied to these eigenvalues gives

$$\det \widehat{V}_j(\boldsymbol{\theta}) = \prod_{l=1}^{K+1} \lambda_l \le \left\{\frac{1}{K+1}\sum_{l=1}^{K+1} \lambda_l\right\}^{K+1} = \left\{\frac{\operatorname{tr} \widehat{V}_j(\boldsymbol{\theta})}{K+1}\right\}^{K+1}.$$

Therefore, on the event in (56),

$$\sup_{\boldsymbol{\theta}\in\Theta} J(\mathbf{y},\boldsymbol{\theta}) \le \left(\frac{M_V}{K+1}\right)^{d(K+1)/2} =: C_J < \infty,$$

and consequently

$$P_{\boldsymbol{\theta}_0}^{(n)}\left(\sup_{\boldsymbol{\theta}\in\Theta} J(\mathbf{y},\boldsymbol{\theta}) \le C_J\right) \longrightarrow 1. \tag{57}$$

**Construction of the dominating density.** Define $\gamma(\boldsymbol{\theta}) := C_J\,\mathbf{1}_\Theta(\boldsymbol{\theta})$. By Condition (D), $\Theta$ is compact and therefore has finite Lebesgue measure, so

$$\int_{\mathbb{R}^q} \gamma(\boldsymbol{\theta})\,d\boldsymbol{\theta} = C_J\,\mathrm{Leb}(\Theta) < \infty,$$

and $\gamma$ is the density of a finite measure absolutely continuous with respect to Lebesgue measure on the parameter space. Because $p_n^2(\mathbf{w},\boldsymbol{\theta}) \equiv 1$, on the event in (57) we have, for every $\boldsymbol{\theta}\in\Theta$,

$$p_n^2(\mathbf{w},\boldsymbol{\theta})J(\mathbf{y},\boldsymbol{\theta})\mathbf{1}_{\{\|\boldsymbol{\theta}-\boldsymbol{\theta}_0\|>D_n\}} \le C_J\,\mathbf{1}_\Theta(\boldsymbol{\theta}) = \gamma(\boldsymbol{\theta}).$$

It follows that

$$P_{\boldsymbol{\theta}_0}^{(n)}\big(p_n^2(\mathbf{w},\boldsymbol{\theta})J(\mathbf{y},\boldsymbol{\theta})\mathbf{1}_{\{\|\boldsymbol{\theta}-\boldsymbol{\theta}_0\|>D_n\}} \le \gamma(\boldsymbol{\theta}) \text{ for every } \boldsymbol{\theta}\in\Theta\big) \longrightarrow 1.$$

Thus, the required likelihood splitting and domination condition hold, which verifies Assumption 4. □

6.5. *Verification of Assumption 5 (Exponentially consistent tests).*

ASSUMPTION 5 (Exponentially consistent tests). For every $\epsilon > 0$, there exists a sequence of tests $\delta_n$ and constants $c_1, c_2 > 0$ such that for $\|\boldsymbol{\theta}-\boldsymbol{\theta}_0\| \ge \epsilon$,

$$P_{\boldsymbol{\theta}_0}\delta_n \to 0 \quad \text{and} \quad P_{\boldsymbol{\theta}}(1-\delta_n) \le e^{-c_1 n(\|\boldsymbol{\theta}-\boldsymbol{\theta}_0\|^2\wedge 1)} \quad \text{and} \quad P_{\boldsymbol{\theta}}^1(1-\delta_n) \le e^{-c_2 n(\|\boldsymbol{\theta}-\boldsymbol{\theta}_0\|^2\wedge 1)},$$

where $P_{\boldsymbol{\theta}}^1$ is the measure with density $p^1(\mathbf{y},\boldsymbol{\theta})$ from Assumption 4.

PROOF. We use likelihood-ratio tests and exponential moment bounds based on the Hellinger affinity. Under the reference-measure splitting construction in the proof of Assumption 4,

$$P_{\boldsymbol{\theta}}^{1,(n)} = P_{\boldsymbol{\theta}}^{(n)}, \tag{58}$$

so the tests may be constructed from all observations $W_1,\dots,W_n$, with $W_i = (Z_i, \mathbf{Y}_i)$, and the resulting type-I and type-II bounds under the full model apply identically to $P_{\boldsymbol{\theta}}^{1,(n)}$.

Let $p_{\boldsymbol{\theta}}(w)$ denote a density of one joint observation $W_i$ under the joint random-design working model, with respect to the dominating measure $\Lambda_W$ of Remark 2. Because the design distribution $Q$ does not depend on $\boldsymbol{\theta}$, likelihood ratios depend only on the conditional response densities. In particular,

$$\frac{p_{\boldsymbol{\theta}}(z,\mathbf{y})}{p_{\boldsymbol{\theta}_0}(z,\mathbf{y})} = \frac{p_{\boldsymbol{\theta}}(\mathbf{y}\mid z)}{p_{\boldsymbol{\theta}_0}(\mathbf{y}\mid z)}$$

wherever the ratio is defined.

**Step 1: Identifiability of the joint working model.**

We first show that $P_{\boldsymbol{\theta}} = P_{\boldsymbol{\theta}_0}$ implies $\boldsymbol{\theta} = \boldsymbol{\theta}_0$. Because the design distribution $Q$ is the same under every parameter value, equality of the joint distributions implies equality of the conditional response distributions for $Q$-almost every design value $z = (x, t)$.

Fix a component $j$. The oracle partition is determined at $\boldsymbol{\theta}_0$ and is held fixed as $\boldsymbol{\theta}$ varies. Consequently, for a given design value $z$, the same component family is used under $\boldsymbol{\theta}$ and $\boldsymbol{\theta}_0$: either both conditional distributions are NB or both are normal.

Under either component family, the conditional mean and variance are

$$\mu_j(z;\boldsymbol{\theta}_j) = \frac{\exp(x^\top \beta_j)t}{\nu_j} \quad \text{and} \quad \sigma_j^2(z;\boldsymbol{\theta}_j) = \frac{\exp(x^\top \beta_j)t(1+\nu_j)}{\nu_j^2} = \mu_j(z;\boldsymbol{\theta}_j)\left(1+\frac{1}{\nu_j}\right).$$

Equality of the conditional distributions implies equality of their first two moments. Hence, for $Q$-almost every $z$,

$$\mu_j(z;\boldsymbol{\theta}_j) = \mu_j(z;\boldsymbol{\theta}_{j,0}) \quad \text{and} \quad \sigma_j^2(z;\boldsymbol{\theta}_j) = \sigma_j^2(z;\boldsymbol{\theta}_{j,0}).$$

Because the conditional means are strictly positive,

$$1+\frac{1}{\nu_j} = \frac{\sigma_j^2(z;\boldsymbol{\theta}_j)}{\mu_j(z;\boldsymbol{\theta}_j)} = \frac{\sigma_j^2(z;\boldsymbol{\theta}_{j,0})}{\mu_j(z;\boldsymbol{\theta}_{j,0})} = 1+\frac{1}{\nu_{j,0}}.$$

Hence, $\nu_j = \nu_{j,0}$, and equality of the conditional means then gives

$$\exp(x^\top \beta_j)t = \exp(x^\top \beta_{j,0})t.$$

Since $t > 0$ almost surely, $x^\top(\beta_j - \beta_{j,0}) = 0$ $Q$-almost surely. Therefore,

$$0 = \mathbb{E}_Q\left[\left\{X_i^\top(\beta_j - \beta_{j,0})\right\}^2\right] = (\beta_j - \beta_{j,0})^\top \mathbb{E}[X_i X_i^\top](\beta_j - \beta_{j,0}).$$

Condition (C)(i) gives $\mathbb{E}[X_i X_i^\top] \succ 0$, and hence $\beta_j = \beta_{j,0}$. Applying this argument to every component $j \in \{1,\dots,d\}$ gives $\boldsymbol{\theta} = \boldsymbol{\theta}_0$. Thus, the joint working model is identifiable.

**Step 2: Uniform separation in Hellinger affinity.**

For $\boldsymbol{\theta}, \boldsymbol{\theta}' \in \Theta$, define the one-observation Hellinger affinity

$$\rho(\boldsymbol{\theta}, \boldsymbol{\theta}') := \int \sqrt{p_{\boldsymbol{\theta}}(w) p_{\boldsymbol{\theta}'}(w)}\, d\Lambda_W(w).$$

Because the design distribution is common to all parameter values, this affinity may equivalently be written as

$$\rho(\boldsymbol{\theta}, \boldsymbol{\theta}') = \int \left[\int \sqrt{p_{\boldsymbol{\theta}}(\mathbf{y} \mid z) p_{\boldsymbol{\theta}'}(\mathbf{y} \mid z)}\, d\Lambda_z(\mathbf{y})\right] dQ(z).$$

The affinity satisfies $0 \le \rho(\boldsymbol{\theta}, \boldsymbol{\theta}') \le 1$. Moreover, $\rho(\boldsymbol{\theta}, \boldsymbol{\theta}') = 1$ if and only if the two probability measures are equal. By identifiability, $\rho(\boldsymbol{\theta}, \boldsymbol{\theta}_0) < 1$ whenever $\boldsymbol{\theta} \neq \boldsymbol{\theta}_0$.

The map $\boldsymbol{\theta} \longmapsto \rho(\boldsymbol{\theta}, \boldsymbol{\theta}_0)$ is continuous on $\Theta$. For each design value, the normal densities and NB probability mass functions are continuous in $\boldsymbol{\theta}$. Moreover, for any two probability densities or mass functions $p$ and $q$, the arithmetic–geometric mean inequality (Casella and Berger, 2002, Example 4.7.8) gives $\sqrt{pq} \le (p+q)/2$, so the conditional Hellinger integrands are dominated by integrable probability densities. The dominated convergence theorem (Keener, 2010, Theorem 2.5) therefore gives continuity of the conditional affinity, and because $d$ is fixed, taking products over components and integrating over the common design distribution $Q$ preserves continuity.

Fix $\epsilon > 0$ and define the compact alternative

$$\Theta_\epsilon := \{\boldsymbol{\theta} \in \Theta : \|\boldsymbol{\theta} - \boldsymbol{\theta}_0\| \geq \epsilon\}.$$

Since $\Theta$ is compact, $\Theta_\epsilon$ is compact. By continuity and identifiability,

$$\rho_\epsilon := \sup_{\boldsymbol{\theta} \in \Theta_\epsilon} \rho(\boldsymbol{\theta}, \boldsymbol{\theta}_0) < 1. \tag{59}$$

**Step 3: Pairwise likelihood-ratio tests.**

Fix $\boldsymbol{\theta}' \in \Theta_\epsilon$. Define the log-likelihood ratio based on the full sample of $n$ observations by

$$T_{n,\boldsymbol{\theta}'} := \sum_{i=1}^{n} \log \frac{p_{\boldsymbol{\theta}'}(W_i)}{p_{\boldsymbol{\theta}_0}(W_i)}.$$

Consider the likelihood-ratio test

$$\delta_{n,\boldsymbol{\theta}'} := \mathbf{1}\{T_{n,\boldsymbol{\theta}'} \geq 0\}.$$

Under the null distribution, Markov's inequality (van der Vaart, 1998, Example 2.6) gives

$$P_{\boldsymbol{\theta}_0}^{(n)}\left(\delta_{n,\boldsymbol{\theta}'} = 1\right) = P_{\boldsymbol{\theta}_0}^{(n)}\left(\exp\{T_{n,\boldsymbol{\theta}'}/2\} \geq 1\right) \leq \mathbb{E}_{\boldsymbol{\theta}_0}^{(n)}\left[\exp\{T_{n,\boldsymbol{\theta}'}/2\}\right].$$

Because the observations are i.i.d.,

$$\mathbb{E}_{\boldsymbol{\theta}_0}^{(n)}\left[\exp\{T_{n,\boldsymbol{\theta}'}/2\}\right] = \prod_{i=1}^{n} \mathbb{E}_{\boldsymbol{\theta}_0}\left[\left\{\frac{p_{\boldsymbol{\theta}'}(W_i)}{p_{\boldsymbol{\theta}_0}(W_i)}\right\}^{1/2}\right] = \rho(\boldsymbol{\theta}', \boldsymbol{\theta}_0)^n.$$

It follows that

$$P_{\boldsymbol{\theta}_0}^{(n)}\left(\delta_{n,\boldsymbol{\theta}'} = 1\right) \leq \rho(\boldsymbol{\theta}', \boldsymbol{\theta}_0)^n \leq \rho_\epsilon^n. \tag{60}$$

Similarly, under $P_{\boldsymbol{\theta}'}^{(n)}$, Markov's inequality (van der Vaart, 1998, Example 2.6) applied to $\exp\{-T_{n,\boldsymbol{\theta}'}/2\}$ gives

$$\begin{aligned} P_{\boldsymbol{\theta}'}^{(n)}\left(\delta_{n,\boldsymbol{\theta}'} = 0\right) &\leq P_{\boldsymbol{\theta}'}^{(n)}\left(\exp\{-T_{n,\boldsymbol{\theta}'}/2\} \geq 1\right) \\ &\leq \mathbb{E}_{\boldsymbol{\theta}'}^{(n)}\left[\exp\{-T_{n,\boldsymbol{\theta}'}/2\}\right] \\ &= \prod_{i=1}^{n} \mathbb{E}_{\boldsymbol{\theta}'}\left[\left\{\frac{p_{\boldsymbol{\theta}_0}(W_i)}{p_{\boldsymbol{\theta}'}(W_i)}\right\}^{1/2}\right] \\ &= \rho(\boldsymbol{\theta}', \boldsymbol{\theta}_0)^n \\ &\leq \rho_\epsilon^n. \end{aligned}$$

Thus, for every fixed $\boldsymbol{\theta}' \in \Theta_\epsilon$, the likelihood-ratio test has exponentially decreasing type-I and type-II errors. It remains to combine these pairwise tests into a single test that is uniform over the composite alternative $\Theta_\epsilon$.

**Step 4: Uniformization over the composite alternative.**

For $\boldsymbol{\theta}, \boldsymbol{\theta}' \in \Theta$, define the cross-affinity moment

$$\kappa(\boldsymbol{\theta}, \boldsymbol{\theta}') := \mathbb{E}_{\boldsymbol{\theta}}\left[\left\{\frac{p_{\boldsymbol{\theta}_0}(W)}{p_{\boldsymbol{\theta}'}(W)}\right\}^{1/2}\right] = \int p_{\boldsymbol{\theta}}(w)\left\{\frac{p_{\boldsymbol{\theta}_0}(w)}{p_{\boldsymbol{\theta}'}(w)}\right\}^{1/2} d\Lambda_W(w).$$

At $\boldsymbol{\theta} = \boldsymbol{\theta}'$,

$$\begin{aligned}\kappa(\boldsymbol{\theta}', \boldsymbol{\theta}') &= \int p_{\boldsymbol{\theta}'}(w) \left\{ \frac{p_{\boldsymbol{\theta}_0}(w)}{p_{\boldsymbol{\theta}'}(w)} \right\}^{1/2} d\Lambda_W(w) \\ &= \int \sqrt{p_{\boldsymbol{\theta}_0}(w) p_{\boldsymbol{\theta}'}(w)}\, d\Lambda_W(w) \\ &= \rho(\boldsymbol{\theta}', \boldsymbol{\theta}_0) \\ &\le \rho_\epsilon \\ &< 1.\end{aligned}$$

For every fixed $\boldsymbol{\theta}' \in \Theta_\epsilon$, the map $\boldsymbol{\theta} \longmapsto \kappa(\boldsymbol{\theta}, \boldsymbol{\theta}')$ is continuous in a neighborhood of $\boldsymbol{\theta}'$. Because $p_{\boldsymbol{\theta}'}$ appears in the denominator, this requires an integrable envelope rather than pointwise continuity alone. The oracle partition and the support of each component distribution are fixed as $\boldsymbol{\theta}$ varies, so, conditional on the design, the integrand factorizes across components and it suffices to dominate each factor.

For a normal contribution, the dependence on $y$ is $\exp\{-c_{ij}(\boldsymbol{\theta}) y^2 + O(|y|)\}$, where

$$c_{ij}(\boldsymbol{\theta}) = \frac{1}{2\sigma_{ij}^2(\boldsymbol{\theta})} + \frac{1}{4\sigma_{ij}^2(\boldsymbol{\theta}_0)} - \frac{1}{4\sigma_{ij}^2(\boldsymbol{\theta}')}, \qquad c_{ij}(\boldsymbol{\theta}') = \frac{1}{4\sigma_{ij}^2(\boldsymbol{\theta}')} + \frac{1}{4\sigma_{ij}^2(\boldsymbol{\theta}_0)} > 0.$$

For a NB contribution, the dependence on $y$ is bounded by a polynomial multiple of

$$\left[ b_j(\boldsymbol{\theta}) \left\{ \frac{b_j(\boldsymbol{\theta}_0)}{b_j(\boldsymbol{\theta}')} \right\}^{1/2} \right]^y, \qquad b_j(\boldsymbol{\theta}) = (1 + \nu_j)^{-1},$$

whose base at $\boldsymbol{\theta} = \boldsymbol{\theta}'$ equals $\sqrt{b_j(\boldsymbol{\theta}_0) b_j(\boldsymbol{\theta}')} < 1$. The polynomial bound follows from the corresponding ratios of Gamma functions.

Both $c_{ij}$ and the base of the geometric factor are continuous in $\boldsymbol{\theta}$. Moreover, Conditions (A), (B), and (D) imply that the conditional means, variances, shape parameters, and dispersions are bounded above and below by positive constants, uniformly over the compact design and parameter sets. On a sufficiently small neighborhood of $\boldsymbol{\theta}'$, the normal factors are therefore dominated by an integrable envelope of the form $\exp\{-c_0 y^2 + C_0 |y|\}$ with $c_0 > 0$, while the NB factors are dominated by a polynomial times a geometric sequence whose rate is bounded strictly below one, uniformly in the design. Since the number of components is fixed, the product of these componentwise envelopes is integrable. Integration with respect to the design probability measure $Q$ preserves this domination, and the dominated convergence theorem (Keener, 2010, Theorem 2.5) therefore gives the required local continuity.

Choose a constant $\bar{\rho}$ satisfying $\rho_\epsilon < \bar{\rho} < 1$. For each $\boldsymbol{\theta}' \in \Theta_\epsilon$, continuity of $\kappa(\boldsymbol{\theta}, \boldsymbol{\theta}')$ and the inequality $\kappa(\boldsymbol{\theta}', \boldsymbol{\theta}') \le \rho_\epsilon < \bar{\rho}$ imply that there exists an open neighborhood $\mathcal{U}_{\boldsymbol{\theta}'}$ of $\boldsymbol{\theta}'$ such that

$$\sup_{\boldsymbol{\theta} \in \mathcal{U}_{\boldsymbol{\theta}'}} \kappa(\boldsymbol{\theta}, \boldsymbol{\theta}') \le \bar{\rho}. \tag{61}$$

The collection $\left\{ \mathcal{U}_{\boldsymbol{\theta}'} : \boldsymbol{\theta}' \in \Theta_\epsilon \right\}$ is an open cover of the compact set $\Theta_\epsilon$. Hence, there exist finitely many centers $\boldsymbol{\theta}^{(1)}, \ldots, \boldsymbol{\theta}^{(L)} \in \Theta_\epsilon$ such that

$$\Theta_\epsilon \subseteq \bigcup_{\ell=1}^{L} \mathcal{U}_{\boldsymbol{\theta}^{(\ell)}}.$$

For each center $\boldsymbol{\theta}^{(\ell)}$, define

$$T_{n,\ell} := \sum_{i=1}^{n} \log \frac{p_{\boldsymbol{\theta}^{(\ell)}}(W_i)}{p_{\boldsymbol{\theta}_0}(W_i)}$$

and

$$\delta_{n,\ell} := \mathbf{1}\{T_{n,\ell} \geq 0\}.$$

Define the composite test

$$\delta_n := \max_{1\leq \ell \leq L} \delta_{n,\ell}.$$

Under the null, the union bound and (60) give

$$P^{(n)}_{\boldsymbol{\theta}_0}(\delta_n = 1) \leq \sum_{\ell=1}^{L} P^{(n)}_{\boldsymbol{\theta}_0}(\delta_{n,\ell} = 1) \leq L\rho_\epsilon^n \longrightarrow 0. \tag{62}$$

Now fix $\boldsymbol{\theta} \in \Theta_\epsilon$. Choose $\ell \in \{1, \dots, L\}$ such that $\boldsymbol{\theta} \in \mathcal{U}_{\boldsymbol{\theta}^{(\ell)}}$. Since $\delta_n \geq \delta_{n,\ell}$,

$$P^{(n)}_{\boldsymbol{\theta}}(1-\delta_n) \leq P^{(n)}_{\boldsymbol{\theta}}(1-\delta_{n,\ell}) = P^{(n)}_{\boldsymbol{\theta}}(T_{n,\ell} < 0).$$

Applying Markov's inequality (van der Vaart, 1998, Example 2.6) to $\exp\{-T_{n,\ell}/2\}$ gives

$$\begin{aligned}
P^{(n)}_{\boldsymbol{\theta}}(T_{n,\ell} < 0) &\leq \mathbb{E}^{(n)}_{\boldsymbol{\theta}}\left[\exp\{-T_{n,\ell}/2\}\right] \\
&= \prod_{i=1}^{n} \mathbb{E}_{\boldsymbol{\theta}}\left[\left\{\frac{p_{\boldsymbol{\theta}_0}(W_i)}{p_{\boldsymbol{\theta}^{(\ell)}}(W_i)}\right\}^{1/2}\right] \\
&= \kappa(\boldsymbol{\theta}, \boldsymbol{\theta}^{(\ell)})^n \\
&\leq \bar{\rho}^n,
\end{aligned}$$

where the final inequality follows from (61). Consequently,

$$\sup_{\boldsymbol{\theta}\in\Theta_\epsilon} P^{(n)}_{\boldsymbol{\theta}}(1-\delta_n) \leq \bar{\rho}^n. \tag{63}$$

**Step 5: Exponential form of the bounds.**

By (58), the bounds established above apply to $P^{1,(n)}_{\boldsymbol{\theta}}$ as well as to $P^{(n)}_{\boldsymbol{\theta}}$. In particular, (62) gives

$$P^{(n)}_{\boldsymbol{\theta}_0}\delta_n \leq L\rho_\epsilon^n \longrightarrow 0,$$

and (63) gives, uniformly over $\boldsymbol{\theta} \in \Theta_\epsilon$,

$$P^{(n)}_{\boldsymbol{\theta}}(1-\delta_n) = P^{1,(n)}_{\boldsymbol{\theta}}(1-\delta_n) \leq \bar{\rho}^n. \tag{64}$$

Let $a := -\log\bar{\rho} > 0$, so that $\bar{\rho}^n = \exp(-an)$. Because $\|\boldsymbol{\theta} - \boldsymbol{\theta}_0\|^2 \wedge 1 \leq 1$,

$$\exp(-an) \leq \exp\left\{-an\left(\|\boldsymbol{\theta} - \boldsymbol{\theta}_0\|^2 \wedge 1\right)\right\}.$$

Thus, taking $c_1 = c_2 = a > 0$, Equation (64) gives, for every $\boldsymbol{\theta} \in \Theta_\epsilon$,

$$P^{(n)}_{\boldsymbol{\theta}}(1-\delta_n) \leq \exp\left\{-c_1 n\left(\|\boldsymbol{\theta} - \boldsymbol{\theta}_0\|^2 \wedge 1\right)\right\}$$

and

$$P^{1,(n)}_{\boldsymbol{\theta}}(1-\delta_n) \leq \exp\left\{-c_2 n\left(\|\boldsymbol{\theta} - \boldsymbol{\theta}_0\|^2 \wedge 1\right)\right\}$$

uniformly over $\|\boldsymbol{\theta} - \boldsymbol{\theta}_0\| \geq \epsilon$. Together with $P^{(n)}_{\boldsymbol{\theta}_0}\delta_n \longrightarrow 0$, this verifies Assumption 5. □

| Regime | Parameter | Sample size $n$ 50 | 100 | 200 | 500 |
|---|---|---|---|---|---|
| Sparse | Intercept | 0.951 | 0.948 | 0.946 | 0.945 |
| | $\beta_1$ (continuous) | 0.948 | 0.959 | 0.945 | 0.950 |
| | $\beta_2$ (group A) | 0.952 | 0.955 | 0.942 | 0.958 |
| | $\beta_3$ (group B) | 0.952 | 0.949 | 0.949 | 0.954 |
| | $\nu$ | 0.951 | 0.953 | 0.947 | 0.947 |
| | *Average* | *0.951* | *0.953* | *0.946* | *0.951* |
| Borderline | Intercept | 0.954 | 0.938 | 0.946 | 0.941 |
| | $\beta_1$ (continuous) | 0.941 | 0.940 | 0.943 | 0.940 |
| | $\beta_2$ (group A) | 0.949 | 0.951 | 0.943 | 0.936 |
| | $\beta_3$ (group B) | 0.949 | 0.952 | 0.935 | 0.931 |
| | $\nu$ | 0.948 | 0.931 | 0.938 | 0.936 |
| | *Average* | *0.948* | *0.942* | *0.941* | *0.937* |
| Heavy | Intercept | 0.949 | 0.941 | 0.951 | 0.945 |
| | $\beta_1$ (continuous) | 0.949 | 0.948 | 0.949 | 0.952 |
| | $\beta_2$ (group A) | 0.951 | 0.942 | 0.946 | 0.946 |
| | $\beta_3$ (group B) | 0.949 | 0.951 | 0.942 | 0.948 |
| | $\nu$ | 0.951 | 0.941 | 0.951 | 0.945 |
| | *Average* | *0.950* | *0.944* | *0.948* | *0.947* |

TABLE 1

*Empirical coverage of* $95\%$ *generalized fiducial intervals by regime, parameter and sample size (*$R = 2000$ *replications per cell, Monte Carlo standard error approximately* $0.005$*). Regimes are ordered by count, as in Table 2.*

*Conclusion.* Assumptions 1–5 are the five regularity conditions of Theorem 3.1 of Borgert and Hannig (2026), verified above under Conditions (A)–(E). That theorem therefore applies to the hybrid working model with norming matrices $\mathbf{r}_n = \sqrt{n}\,\mathbf{I}_q$, central sequence $\Delta_{n,\boldsymbol{\theta}_0}$ and information $\mathcal{I}_{\boldsymbol{\theta}_0}$, which is the statement of Theorem 5.

**7. Simulation Study: Full Coverage Results.** Section 5.2 summarizes the coverage sweep by regime. Table 1 of this supplement gives the individual cells: empirical coverage of the $95\%$ generalized fiducial intervals for each regime, parameter and sample size. The flat-prior Bayesian comparator agrees closely enough that a second panel would be uninformative: fit to the same datasets, the two differ by at most $0.004$ in coverage and $1.2\%$ in median interval width in every cell. In the sparse regime the fiducial intervals are from the discrete construction (Section 8.3 of this supplement).

*Assignment by the realized count.* Table 2 of this supplement repeats the borderline design of Section 5.2 ($\nu = 0.05$, $R = 2000$ replications per cell) with each observation assigned to the normal component when its own count reaches $\tau = 500$, rather than when its expected count does. About $49\%$ of observations are assigned to the normal component at every sample size. Coverage of the intercept and the dispersion falls steadily with $n$, while the slopes stay near nominal. Because assignment by the response selects the observations with positive errors into the normal component, the working likelihood is biased, and the bias does not shrink as the sample size grows. The three fits whose chains did not move on the first attempt were refit, and all converged.

*Convergence in the partition-rule comparison.* Of the 25,999 hybrid fits behind Table 2, those with $\widehat{R} > 1.05$ or with no $\widehat{R}$ were refit with four chains and a longer warmup, and the $118$ that still had not converged were excluded. Ninety-one are oracle fits in the heavy regime with $\nu \leq 0.008$, concentrated enough that retaining them would lower the heavy-regime oracle cells by $0.005$, $0.016$, $0.016$ and $0.022$ at $n = 50$, 100, 200 and 500. Every

| Parameter | Sample size $n$ | | | |
|---|---|---|---|---|
| | 50 | 100 | 200 | 500 |
| Intercept | 0.915 | 0.889 | 0.843 | 0.714 |
| $\beta_1$ (continuous) | 0.949 | 0.944 | 0.949 | 0.941 |
| $\beta_2$ (group A) | 0.952 | 0.944 | 0.955 | 0.945 |
| $\beta_3$ (group B) | 0.954 | 0.947 | 0.952 | 0.941 |
| $\nu$ | 0.901 | 0.876 | 0.830 | 0.689 |

TABLE 2

*Empirical coverage of* $95\%$ *generalized fiducial intervals in the borderline regime when each observation is assigned by its own realized count (*$R = 2000$ *replications per cell, Monte Carlo standard error approximately* $0.005$*).*

other cell moves by less than $0.002$. The remaining 27 are 13 heavy-regime fits under the other three rules, four borderline fits and the 10 sparse fits described below. The $22{,}001$ entirely NB fits come from the $7{,}175$ datasets on which at least one rule assigned nothing to the normal component, and all are included. Of these, 62 had a $\log\nu$ chain whose effective sample size fell below $0.8$ times that of the coefficients. In the sparse regime all 10 excluded hybrid fits are mean-rule fits with $\nu \leq 0.008$ and every observation assigned to the normal component.

**8. Computation of the Discrete Generalized Fiducial Distribution.** This section supports Section 2.1. It proves Lemma 1 and describes the Metropolis–Hastings (MH) sampler used for the 26 entirely discrete bins of the application in Section 6, together with its validation against exact draws from the definition. It then reports the diagnostics for those bins and compares the discrete construction with the flat-prior posterior on $(\beta, \log\nu)$, both in the application and in the simulation study of Section 5.1. The code supplement contains the MH sampler, the single-site Gibbs sampler used as a reference at small $n$, the exact rejection sampler and the high-precision recomputation of the proposals the MH sampler could not evaluate.

8.1. *Setting and Selection Rule.* Fix an entirely discrete component and suppress its index, so $\boldsymbol{\theta} = (\beta^\top, \nu)^\top$, $\mathbf{y} \in \mathbb{N}_0^n$ and $\mathbf{u} \in (0,1)^n$. Write $\lambda = \log\nu$. Computations take place on a bounded parameter set

$$\Theta = (-b, b)^K \times (\lambda_{\text{lo}}, \lambda_{\text{hi}}),$$

with $b = 100$ and $(\lambda_{\text{lo}}, \lambda_{\text{hi}}) = (-10, 0)$ in the application and in every simulation cell of Section 5.1. The bounds sit more than eight standard deviations from every fitted distribution, and no chain state reached either bound of $\lambda$ or any bound of $\beta$ in any run. Both are checked on every chain (Section 8.6 of this supplement). The validation problems of Section 8.4 of this supplement are small enough that the fiducial distribution reaches lower dispersions, and there $(\lambda_{\text{lo}}, \lambda_{\text{hi}}) = (-6, 6)$, stated with each table. With $\mathcal{Q}(\mathbf{y}, \mathbf{u}) = \{\boldsymbol{\theta} \in \Theta : G_D(\mathbf{u}, \boldsymbol{\theta}) = \mathbf{y}\}$ and $\mathcal{B}(\boldsymbol{\theta}) = \{\mathbf{u} : G_D(\mathbf{u}, \boldsymbol{\theta}) = \mathbf{y}\}$, let $\mathcal{S} = \{\mathbf{u} : \mathcal{Q}(\mathbf{y}, \mathbf{u}) \neq \varnothing\}$ and $L(\boldsymbol{\theta}) = \prod_i f_{\text{NB}}(y_i \mid \boldsymbol{\theta})$.

At fixed $\nu$, each observation constrains the linear predictor $\eta_i = X_i^\top \beta$ through

$$F_{\text{NB}}(y_i - 1; e^{\eta_i} t_i, \nu) < u_i \leq F_{\text{NB}}(y_i; e^{\eta_i} t_i, \nu).$$

Since $F_{\text{NB}}(y; \mu, \nu)$ is decreasing in $\mu$, this places $\eta_i$ in a window

$$\eta_i^{\text{lo}}(u_i, \lambda) < \eta_i \leq \eta_i^{\text{hi}}(u_i, \lambda)$$

whose endpoints solve

$$F_{\text{NB}}\big(y_i - 1; e^{\eta_i^{\text{lo}}(u_i,\lambda)} t_i, \nu\big) = u_i, \qquad F_{\text{NB}}\big(y_i; e^{\eta_i^{\text{hi}}(u_i,\lambda)} t_i, \nu\big) = u_i,$$

with $\eta_i^{\mathrm{lo}}(u_i,\lambda)=-\infty$ when $y_i=0$. Hence

$$\mathcal{Q}(\mathbf{y},\mathbf{u})=\{(\beta,\lambda):\lambda\in I(\mathbf{u}),\ \beta\in P_{\mathbf{u}}(\lambda)\},$$

where $P_{\mathbf{u}}(\lambda)=\{\beta:\eta_i^{\mathrm{lo}}(u_i,\lambda)<X_i^\top\beta\le\eta_i^{\mathrm{hi}}(u_i,\lambda)\ \forall i,\ |\beta_k|<b\ \forall k\}$ is a convex polytope and $I(\mathbf{u})=\{\lambda\in(\lambda_{\mathrm{lo}},\lambda_{\mathrm{hi}}):P_{\mathbf{u}}(\lambda)$ has nonempty interior$\}$, which differs from the set where $P_{\mathbf{u}}(\lambda)$ is merely nonempty only at endpoints of its pieces, a set of $\lambda$ of measure zero. The set $I(\mathbf{u})$ is a union of disjoint open intervals, which we call its pieces, and $|I(\mathbf{u})|$ denotes its total length. The selection rule is

$$V:\qquad \lambda\sim\mathrm{Uniform}(I(\mathbf{u})),\qquad \beta\mid\lambda\sim\mathrm{Uniform}(P_{\mathbf{u}}(\lambda)),\tag{65}$$

with density $v_{\mathbf{u}}(\beta,\lambda)=\{|I(\mathbf{u})|\ \mathrm{vol}_K\,P_{\mathbf{u}}(\lambda)\}^{-1}$ on $\mathcal{Q}(\mathbf{y},\mathbf{u})$ with respect to $d\beta\,d\lambda$.

8.2. *Proof of Lemma 1.* The proof uses two facts.

(i) $\boldsymbol{\theta}\in\mathcal{Q}(\mathbf{y},\mathbf{u})$ if and only if $\mathbf{u}\in\mathcal{B}(\boldsymbol{\theta})$, since each restates $G_D(\mathbf{u},\boldsymbol{\theta})=\mathbf{y}$.
(ii) $\mathcal{B}(\boldsymbol{\theta})=\prod_i\big(F_{\mathrm{NB}}(y_i-1;\boldsymbol{\theta}),F_{\mathrm{NB}}(y_i;\boldsymbol{\theta})\big]$ has Lebesgue measure $L(\boldsymbol{\theta})$. In particular $P(\mathbf{U}^\star\in\mathcal{S})\ge L(\boldsymbol{\theta})>0$ for every $\boldsymbol{\theta}\in\Theta$, so the conditioning below is well defined.

PROOF. On $\Theta\times[0,1]^n$ define

$$p(\boldsymbol{\theta},\mathbf{u})=\frac{v_{\mathbf{u}}(\boldsymbol{\theta})\,\mathbf{1}\{\mathbf{u}\in\mathcal{B}(\boldsymbol{\theta})\}}{P(\mathbf{U}^\star\in\mathcal{S})}.$$

We show that $p$ is a probability density, that its $\mathbf{u}$-marginal is uniform on $\mathcal{S}$ with conditionals $V(\cdot\mid\mathcal{Q}(\mathbf{y},\mathbf{u}))$, and hence that its $\boldsymbol{\theta}$-marginal is the discrete GFD.

*The* $\mathbf{u}$*-marginal.* By (i) the indicator equals $\mathbf{1}\{\boldsymbol{\theta}\in\mathcal{Q}(\mathbf{y},\mathbf{u})\}$, so fixing $\mathbf{u}$ and integrating over $\boldsymbol{\theta}$ gives $1/P(\mathbf{U}^\star\in\mathcal{S})$ when $\mathbf{u}\in\mathcal{S}$, because $v_{\mathbf{u}}$ is a probability density supported on $\mathcal{Q}(\mathbf{y},\mathbf{u})$, and 0 otherwise. The $\mathbf{u}$-marginal of $p$ is therefore uniform on $\mathcal{S}$, and $p$ integrates to one.

*The conditionals.* Dividing $p$ by that marginal shows that, given $\mathbf{u}\in\mathcal{S}$, the conditional density of $\boldsymbol{\theta}$ is $v_{\mathbf{u}}$, which is the selection rule $V(\cdot\mid\mathcal{Q}(\mathbf{y},\mathbf{u}))$.

*Identification.* The discrete GFD is by definition the law of $V(\mathcal{Q}(\mathbf{y},\mathbf{U}^\star))$ given $\mathbf{U}^\star\in\mathcal{S}$, which is exactly the law of $\boldsymbol{\theta}$ under $p$. Its density is therefore the $\boldsymbol{\theta}$-marginal of $p$,

$$r(\boldsymbol{\theta}\mid\mathbf{y})=\frac{1}{P(\mathbf{U}^\star\in\mathcal{S})}\int v_{\mathbf{u}}(\boldsymbol{\theta})\,\mathbf{1}\{\mathbf{u}\in\mathcal{B}(\boldsymbol{\theta})\}\,d\mathbf{u}=\frac{L(\boldsymbol{\theta})}{P(\mathbf{U}^\star\in\mathcal{S})}\,\mathbb{E}_{\boldsymbol{\theta}}\big[v_{\mathbf{U}}(\boldsymbol{\theta})\big],$$

the last equality by (ii), with $\mathbf{U}\sim\mathrm{Uniform}(\mathcal{B}(\boldsymbol{\theta}))$. □

8.3. *The Sampler.* Lemma 1 exhibits $r(\boldsymbol{\theta}\mid\mathbf{y})$ as the $\boldsymbol{\theta}$-marginal of a joint density on $\Theta\times\mathcal{S}$, which suggests sampling the pair. Sampling $(\boldsymbol{\theta},\mathbf{u})$ directly is awkward, because a proposal that moves $\boldsymbol{\theta}$ must move $\mathbf{u}$ with it: the support constraint $\mathbf{u}\in\mathcal{B}(\boldsymbol{\theta})$ ties the two together, and the reverse move must recover the uniforms the forward move started from. The difficulty disappears if each uniform is recorded not by its value but by its position within its cell: the interval $\big(F_{\mathrm{NB}}(y_i-1;\boldsymbol{\theta}),F_{\mathrm{NB}}(y_i;\boldsymbol{\theta})\big]$ in which $u_i$ must lie for $\boldsymbol{\theta}$ to reproduce $y_i$. Let

$$u_i=\varphi_{\boldsymbol{\theta}}(\mathbf{w})_i=F_{\mathrm{NB}}(y_i-1;\boldsymbol{\theta})+w_if_{\mathrm{NB}}(y_i\mid\boldsymbol{\theta}),\qquad w_i\in(0,1).$$

For fixed $\boldsymbol{\theta}$ the map $\varphi_{\boldsymbol{\theta}}$ carries $(0,1)^n$ onto $\mathcal{B}(\boldsymbol{\theta})$ with Jacobian determinant $L(\boldsymbol{\theta})$, so the density of $(\boldsymbol{\theta},\mathbf{w})$ is

$$q(\boldsymbol{\theta},\mathbf{w})\ \propto\ L(\boldsymbol{\theta})\,v_{\varphi_{\boldsymbol{\theta}}(\mathbf{w})}(\boldsymbol{\theta})\qquad\text{on }\Theta\times(0,1)^n,$$

**Algorithm 1** Metropolis–Hastings sampler for the discrete GFD of $(\beta, \nu)$

**Input:** counts $\mathbf{y}$, design rows $X_1, \ldots, X_n \in \mathbb{R}^K$, exposures $t_1, \ldots, t_n$; the maximum likelihood estimate $\hat{\boldsymbol{\theta}} = (\hat{\beta}, \hat{\lambda})$ and its inverse observed information $\hat{\Sigma}$; the parameter set $\Theta$; number of steps $S$; block size $m = \max(1, \lfloor n/10 \rfloor)$; scan spacing $h$
**Start:** Set $\boldsymbol{\theta} \leftarrow \hat{\boldsymbol{\theta}}$. Draw $\mathbf{w} \sim \mathrm{Uniform}(0,1)^n$ until $q(\boldsymbol{\theta}, \mathbf{w}) > 0$. Set $\zeta \leftarrow \log q(\boldsymbol{\theta}, \mathbf{w})$.
**for** $s = 1$ to $S$ **do**
  *Move $\boldsymbol{\theta}$, keeping every $u_i$ at the same relative position in its cell.*
  Draw $\boldsymbol{\theta}' \sim g$
  **if** $\boldsymbol{\theta}' \in \Theta$ **then**
    $\zeta' \leftarrow \log q(\boldsymbol{\theta}', \mathbf{w})$
    Accept $\boldsymbol{\theta}'$ with probability $\min\{1,\ e^{\zeta' - \zeta}\, g(\boldsymbol{\theta})/g(\boldsymbol{\theta}')\}$; on acceptance set $\boldsymbol{\theta} \leftarrow \boldsymbol{\theta}'$, $\zeta \leftarrow \zeta'$
  **end if**
  *Move $\mathbf{w}$, keeping $\boldsymbol{\theta}$ fixed.*
  Choose $m$ indices $B \subset \{1, \ldots, n\}$ uniformly at random, without replacement; set $\mathbf{w}' \leftarrow \mathbf{w}$ and redraw $w_i' \sim \mathrm{Uniform}(0,1)$ for $i \in B$
  $\zeta' \leftarrow \log q(\boldsymbol{\theta}, \mathbf{w}')$
  Accept $\mathbf{w}'$ with probability $\min\{1,\ e^{\zeta' - \zeta}\}$; on acceptance set $\mathbf{w} \leftarrow \mathbf{w}'$, $\zeta \leftarrow \zeta'$
  Record $\boldsymbol{\theta}$
**end for**
**Output:** the $S$ recorded values of $\boldsymbol{\theta}$

**procedure** $\log q(\boldsymbol{\theta}, \mathbf{w})$
  $\mathbf{u} \leftarrow \varphi_{\boldsymbol{\theta}}(\mathbf{w})$ ▷ the uniforms at these relative positions
  Find the pieces of $I(\mathbf{u})$ and their total length $|I(\mathbf{u})|$ ▷ **return** $-\infty$ if $\lambda$ lies in none of them or $|I(\mathbf{u})| = 0$
  Compute $\mathrm{vol}_K P_{\mathbf{u}}(\lambda)$ ▷ **return** $-\infty$ if the computation fails
  **return** $\log L(\boldsymbol{\theta}) - \log |I(\mathbf{u})| - \log \mathrm{vol}_K P_{\mathbf{u}}(\lambda)$
**end procedure**

whose $\boldsymbol{\theta}$-marginal is $r(\boldsymbol{\theta} \mid \mathbf{y})$: for fixed $\boldsymbol{\theta}$, $\mathbf{w} \sim \mathrm{Uniform}(0,1)^n$ makes $\mathbf{U} = \varphi_{\boldsymbol{\theta}}(\mathbf{w})$ uniform on $\mathcal{B}(\boldsymbol{\theta})$, so

$$\int_{(0,1)^n} q(\boldsymbol{\theta}, \mathbf{w})\, d\mathbf{w} \ \propto\ L(\boldsymbol{\theta})\, \mathbb{E}_{\boldsymbol{\theta}}\big[v_{\mathbf{U}}(\boldsymbol{\theta})\big] \ \propto\ r(\boldsymbol{\theta} \mid \mathbf{y})$$

by Lemma 1. The expectation is what the lemma identifies as the correction to the likelihood, and the sampler never computes it: the chain samples $(\boldsymbol{\theta}, \mathbf{w})$ jointly, and keeping only the $\boldsymbol{\theta}$ values integrates $\mathbf{w}$ out. In these coordinates, the support no longer depends on $\boldsymbol{\theta}$: $\mathbf{w}$ can be held fixed while $\boldsymbol{\theta}$ moves, every uniform following it, and the reverse move is the same map run backwards. Algorithm 1 alternates a move of this kind with a refresh of a block of $\mathbf{w}$. Both are MH moves on $(\boldsymbol{\theta}, \mathbf{w})$ with target $q$, so the chain is exact for $q$ and its recorded $\boldsymbol{\theta}$ values are draws from $r(\boldsymbol{\theta} \mid \mathbf{y})$. No separate Jacobian term enters the acceptance ratios, because the Jacobian $L(\boldsymbol{\theta})$ of $\mathbf{w} \mapsto \mathbf{u}$ is already part of $q$. In the $\mathbf{w}$-move the likelihood cancels and only $v$ changes.

The $\boldsymbol{\theta}$-proposal is a multivariate $t$ distribution with five degrees of freedom, centered at the maximum likelihood estimate with scale matrix $1.3^2$ times its estimated covariance, so that its spread is $1.3$ times the estimated standard deviations: $g = t_5(\hat{\boldsymbol{\theta}}, 1.3^2\hat{\Sigma})$, drawn independently of the current state. The center and scale place most proposals where the target has most of its mass, since the target is close to $N(\hat{\boldsymbol{\theta}}, \hat{\Sigma})$ at moderate $n$ in our runs (Section 8.5 of this supplement). The heavy tails and the factor 1.3 make sure the proposal also reaches wherever the target does. An independence sampler with a proposal narrower than its target can get stuck in the tails (Mengersen and Tweedie, 1996). The $\boldsymbol{\theta}$-move was accepted at a median rate of $0.52$ in the simulations at $n = 50$ and $0.54$ in the application at $n = 216$.

*Locating pieces.* To find the pieces of $I(\mathbf{u})$, we test whether $P_{\mathbf{u}}(\lambda)$ has interior at a grid of $\lambda$ values spaced $h$ apart across $(\lambda_{\text{lo}}, \lambda_{\text{hi}})$, and also at the current $\lambda$, which is always feasible. Each run of consecutive feasible test points is one piece, and its endpoints are found by a root search between the last feasible point and the first infeasible one on either side. The piece containing the current $\lambda$ is therefore always found, however narrow. A second piece can be missed only if it fits entirely between two grid points, and two pieces separated by a gap narrower than $h$ are read as one. The chain is exact for the target in which $|I(\mathbf{u})|$ is the length the scan finds. The only thing the scan can get wrong is a missed or a merged second piece. Every run records the narrowest piece it found away from the current $\lambda$, to be compared with $h$. The volume of $P_{\mathbf{u}}(\lambda)$ is computed with the box constraints $|\beta_k| < b$ included as faces. These bind only when the design rows of the positive-count observations fail to span $\mathbb{R}^K$, which did not occur in any run reported here.

*Numerical guards.* A proposal whose $\log q$ cannot be evaluated, because $\lambda$ is not located in any piece or the volume cannot be computed, is rejected. Both failures happen only when $P_{\mathbf{u}}(\lambda)$ is too thin for the solver, which is where $v_{\mathbf{u}}$ is largest, so such a rejection could in principle discard a proposal that should have been accepted. We therefore record, for each, how large $v_{\mathbf{u}'}(\boldsymbol{\theta}')$ would have needed to be for acceptance, and check the rejections by re-computing the cells, $\mathbf{u}'$, $I(\mathbf{u}')$, the polytope volume and both likelihoods in 50-digit arithmetic, which gives the acceptance probability each proposal truly had. In the validation, 1,272 of 256,000 $\boldsymbol{\theta}$-proposals were rejected this way and no $\mathbf{w}$-proposal was. The least demanding would have needed $\log v_{\mathbf{u}'}(\boldsymbol{\theta}') \geq 33.8$. Recomputing all 147 that needed less than 100 gave true acceptance probabilities of at most $3 \times 10^{-12}$. The other 1,125 would have needed $\log v_{\mathbf{u}'}(\boldsymbol{\theta}') \geq 100$, against at most 20.7 at any state the chains held, so these proposals would have been rejected in any case.

Second pieces of $I(\mathbf{u})$ arise when a second, widely separated dispersion regime is compatible with every count. They are rare: in 4,000 exact draws of $\mathbf{u}$ at $n = 4$, none had a second piece at a scan spacing of 0.005, the Gibbs chain met two in 20,000 sweeps and the MH chains 12 in 256,000 states, and none was found in any simulation replication or application bin. The scan spacing was $h = 0.02$ in the validation and $h = 0.25$ in the simulations and the application. In a pilot of eight replications at each of $n = 50$ and $n = 500$, spacings of 0.05, 0.1 and 0.25 gave identical chains, as they must when no feasible set has a second piece.

The single-site Gibbs sampler of Hannig, Iyer and Wang (2007), which updates one $u_i$ at a time and then applies $V$, targets the same joint density and is the more efficient choice at small $n$, where its effective sample size is a third to two thirds of the chain length. At $n = 216$ it is not, for the reason given in Section 2.1: on application bins its draws moved across less than a tenth of a posterior standard deviation in several hundred sweeps.

*Implementation.* The samplers are written in Python with NumPy and SciPy (Harris et al., 2020; Virtanen et al., 2020). The NB distribution function is evaluated through the regularized incomplete beta function (`scipy.special.betainc`), and the cell endpoints in $r_{ij}$ implied by each $u_i$ are found by bisection on the log scale. Feasibility of $P_{\mathbf{u}}(\lambda)$ and a point in its interior come from a Chebyshev-center (maximum-slack) linear program, solved with HiGHS (Huangfu and Hall, 2018) through `highspy`. The pieces of $I(\mathbf{u})$ are bracketed by the scan and each end is located by Brent's method on the program's slack. The polytope volume is computed by vertex enumeration with Qhull (Barber, Dobkin and Huhdanpaa, 1996), through `scipy.spatial.HalfspaceIntersection`, after whitening by the Cholesky factor of $\hat{\Sigma}$. The maximum likelihood estimate that centers the $t$ proposal is obtained by BFGS on the NB log-likelihood, and the 50-digit recomputations of rejected proposals use `mpmath` (The mpmath development team, 2026). The Gibbs sampler draws $\beta$ by hit-and-run (Smith, 1984) and shares the linear-program and piece-location code with the

| Sampler | Quantity | Parameter $\beta_0$ | $\beta_1$ | $\log\nu$ |
|---|---|---|---|---|
| Reference | Mean (sd) | $-0.96$ (1.09) | 0.46 (0.43) | $-3.18$ (1.58) |
| Gibbs | Offset (sd units) | $-0.017$ | 0.025 | 0.005 |
| | KS distance | 0.026 | 0.019 | 0.011 |
| MH | Offset (sd units) | 0.019 (0.015) | 0.025 (0.014) | 0.041 (0.025) |

TABLE 3

*Validation on* $n = 4$*,* $K = 2$*,* $y = (8, 27, 2, 0)$*,* $\log\nu \in (-6, 6)$*,* $b = 100$*. The reference row gives means and standard deviations of* 4,000 *exact draws. Offsets are differences of means in reference standard deviations. The Gibbs sampler gives* 20,000 *sweeps, and the first* 20% *are discarded. Its effective sample sizes are about* 9,300, 4,900 *and* 10,900*, so the offsets it would show from sampling noise alone have standard errors of about* 0.02*, and its KS distances have* 95% *points of* 0.026, 0.029 *and* 0.025*. The sampler of Section 8.3 of this supplement is pooled over* 16 *independent chains of* 16,000 *steps, with between-chain standard errors in parentheses. The text explains why no KS distance is reported for it.*

MH sampler. The reference sampler of Section 8.4 shares no code with either: it enumerates the polygon's vertices directly and draws $\beta$ by area-weighted triangulation. Package versions are listed in the code supplement.

8.4. *Validation.* The reference is rejection sampling from the definition: draw $\mathbf{U}^\star \sim \mathrm{Uniform}(0,1)^n$, keep it if $\mathcal{Q}(\mathbf{y}, \mathbf{U}^\star) \neq \varnothing$, and apply $V$ of (65), with $\lambda$ uniform on the exact feasible set and $\beta$ uniform on the polytope by area-weighted triangulation. It shares no code with the sampler, and it is feasible only for very small $n$, because the probability that a uniform draw is consistent with the data falls rapidly with $n$: at $n = 4$, $13\%$ of draws were kept. A second reference run with independent seeds agrees with the first to within sampling noise (offsets of at most $0.027$ reference standard deviations, standard error $0.022$, and Kolmogorov–Smirnov (KS) distances of at most $0.024$, against a $95\%$ point of $0.030$).

Table 3 of this supplement compares the sampler and the Gibbs sampler with the reference on a problem with a zero count, chosen so that the fiducial distribution sits about one reference standard deviation from the maximum likelihood estimate. Both samplers agree with the reference to within sampling noise. The Gibbs offsets are within one standard error of zero and its KS distances are at or below their $95\%$ points, with $\beta_0$ at the edge. The MH offsets are within two between-chain standard errors of zero.

At these sample sizes the sampler's error is reported from independent chains with distinct seeds rather than from effective sample sizes. A single chain at $n = 4$ occasionally remains for hundreds of steps in states where $I(\mathbf{u})$ is very short and $v_{\mathbf{u}}$ correspondingly large, and its effective sample size, estimated from the excursions the chain happened to take, then understates its error. For the same reason a KS distance computed from one such chain is not interpretable, and none is reported for the sampler. This behavior requires a single move to change $v_{\mathbf{u}}$ by a large fraction of its spread, which does not occur at $n = 216$ (Section 8.6 of this supplement).

8.5. *Algorithmic Cost.* The sampler's cost per step grows roughly linearly in $n$, from about $0.2$ seconds at $n = 50$ to $1.9$ seconds at $n = 500$ in the simulation study, and a fit takes about three orders of magnitude longer than the corresponding Stan fit. In the sparse regime of Section 5.1, a single chain per fit suffices: the effective sample size of $\log\nu$ had median $0.99$ times that of the coefficients, and the $41$ replications ($0.51\%$) below $0.8$ are included. Fit to the same datasets, the discrete GFD and the flat-prior posterior agree in coverage to within $0.004$ and in median interval width to within $1.2\%$ in every cell, and the fiducial means sit

| | Intercept | Age | Forward | Midfielder | $\log\nu$ |
|---|---|---|---|---|---|
| Mean shift (sd units) | $-0.016$ | $+0.002$ | $+0.010$ | $+0.015$ | $-0.024$ |
| Standard error | 0.004 | 0.004 | 0.004 | 0.004 | 0.004 |
| sd ratio (se 0.003) | 1.000 | 1.002 | 0.998 | 1.000 | 1.001 |

TABLE 4

*Fiducial minus flat-prior-posterior means in posterior standard deviations, pooled over* 26 *bins and two chains, with standard errors, and the ratio of fiducial to posterior standard deviation.*

below the posterior's by $0.074$, $0.049$, $0.034$ and $0.022$ posterior standard deviations for $\log\nu$ at $n = 50$, 100, 200 and 500 (0.040 to 0.011 for $\beta_0$), vanishing at an observed rate of $n^{-1/2}$. No replication was capped or had an active box constraint, and no state had a second piece of $I(\mathbf{u})$.

8.6. *Application: Diagnostics and Comparison.* The 26 entirely discrete bins have $n = 216$, $K = 4$, minimum counts 39–676, no zero counts, and positive-count design rows of full rank, so the box never binds. Each was fit with two independent chains of 6,000 steps (first 1,200 discarded), with $(\lambda_{\rm lo}, \lambda_{\rm hi}) = (-10, 0)$ and $b = 100$, giving about 2,000 effective draws per coordinate per chain, with the effective sample size of $\log\nu$ between $0.92$ and $1.11$ times that of the coefficients in every bin. Every scan of the feasible $\lambda$-set found a single interval, and the draws of $\log\nu$ span $[-6.40, -2.78]$, so no chain approached a bound on $\lambda$. Of the 312,000 $\boldsymbol{\theta}$-proposals across the 52 chains, 19 could not be evaluated and were rejected by the guards, as were two $\mathbf{w}$-proposals. For 15 of the 19, acceptance would have required $\log v_{\mathbf{u}'}(\boldsymbol{\theta}')$ at least 148 above the largest value the chain held. Recomputing the other four in 50-digit arithmetic, as for the validation, gave acceptance probabilities below $10^{-5}$ for three and $0.83$ for one, a proposal whose polytope was ordinary but which the solver failed to evaluate. That chain therefore held its state for one of its 6,000 steps where it would probably have moved. The difference of each coordinate's mean between the two chains matched that predicted from their effective sample sizes (root-mean-square ratio 1.07 over the 130 comparisons), so the Monte Carlo errors reported from these chains are reliable.

For the same bins we also fit the flat-prior posterior on $(\beta, \log\nu)$, i.e. the NB likelihood under Lebesgue measure on the unconstrained scale, which is the Stan model of the main analysis with no continuous rows and hence no Jacobian term. By Lemma 1 this posterior is the discrete GFD without the factor $\mathbb{E}_{\boldsymbol{\theta}}[v_{\mathbf{U}}(\boldsymbol{\theta})]$, so the comparison measures that factor's effect. Table 4 of this supplement reports, pooled over the 26 bins and both chains, the shift of the fiducial mean from the posterior mean in units of the posterior standard deviation and the ratio of standard deviations. The shift is small, reproduced by both chains, and structured: $\log\nu$ and the intercept move down by about $0.02$ standard deviations and the position contrasts up by about $0.01$. Widths are unchanged, and bin-to-bin variation of the shift beyond Monte Carlo error is about $0.005$ standard deviations. In individual bins the $\log\nu$ shift ranges from $-0.065$ to $+0.018$. Recomputing the sign declarations and effect summaries of Sections 6.3–6.4 with the posterior draws in place of the fiducial draws for these bins changes no Benjamini–Hochberg (BH) declaration, and no reported effect changes by more than $0.001$.

## 9. Application Robustness Checks.

9.1. *Athlete-Level Permutation Check.* For each of the 100 permutations of the athletes' position labels, and separately of their mean ages, described in Section 6.3, the partition was recomputed, all 100 bins were refit, and BH at level $0.05$ was applied. The active bins were fit by Hamiltonian Monte Carlo with shorter chains than the main analysis (four chains of 1,000 draws after 500 warmup), and the entirely discrete bins by the sampler of Section 8.3 of this

| Permuted | Covariate | Observed | Null distribution (100 permutations) | | | | $\#\{\text{null} \geq \text{obs}\}$ |
|---|---|---|---|---|---|---|---|
| | | | Median | IQR | Max | $\Pr(\geq 1)$ | |
| Position | Forward | 24 | 31 | 16–48 | 77 | 0.97 | 59 |
| | Midfielder | 48 | 27 | 10–47 | 84 | 0.90 | 25 |
| | Age | 13 | 24 | 13–39 | 65 | 0.96 | 76 |
| Age | Age | 13 | 28 | 12–49 | 83 | 0.92 | 73 |
| | Forward | 24 | 35 | 31–40 | 57 | 0.99 | 87 |
| | Midfielder | 48 | 45 | 43–47 | 57 | 1.00 | 20 |

TABLE 5

*Number of bins declared significant under the working model (BH at level* $0.05$*) when the athletes' position labels, or their mean ages, are permuted across the* $17$ *athletes and the whole procedure is rerun. "Observed" is the count on the unpermuted data at the reduced sampler settings used for the permutations.* $\Pr(\geq 1)$ *is the fraction of permutations with at least one significant bin.*

supplement with one chain of $3{,}000$ steps. On the observed data these settings reproduce the main analysis's counts to within four bins. With the shorter chains, $86$ of the $7{,}322$ active-bin fits under position permutations and $79$ of $7{,}447$ under age permutations had $\widehat{R} > 1.05$. They are retained, and excluding them from the BH count changes no null median by more than half a bin.

Table 5 of this supplement gives the results. The observed counts are unremarkable against these null distributions, and the largest estimated contrasts are likewise within the null range ($0.14$ and $0.29$ observed against null medians of $0.22$ and $0.26$ for forward and midfielder). With three labels assigned to $17$ athletes, any assignment yields bins in which the labeled groups differ, and a model that counts $216$ sessions resolves those differences as significant.

9.2. *Sensitivity to the Partition Threshold.* Figure 6 places substantial mass near $\omega = 10$, so we refit the per-bin model at $\omega = 7$ and $\omega = 15$. The number of entirely discrete bins is $16$ and $50$, against $26$ at $\omega = 10$. The numbers of bins with a significant age, forward and midfielder effect under the working model (Section 6.3) are $17$, $23$ and $50$ at $\omega = 7$, and $18$, $14$ and $48$ at $\omega = 15$, against $17$, $23$ and $48$ at $\omega = 10$. The age and midfielder results are stable. The forward count falls only at $\omega = 15$, where half the bins become entirely discrete: six of the nine bins it loses are mixed bins at $\omega = 10$ that become entirely discrete at $\omega = 15$.



\end{document}